\documentclass[reprint,amsfonts, amssymb, amsmath,  showkeys, nofootinbib,pra, superscriptaddress, twocolumn,longbibliography]{revtex4-2}
\usepackage{float}
\makeatletter
\let\newfloat\newfloat@ltx
\makeatother
\usepackage[english]{babel}
\usepackage[utf8]{inputenc}
\usepackage{graphics}
\usepackage{selinput}

\usepackage{braket}
\usepackage{amsthm}
\usepackage{mathtools}
\usepackage{physics}
\usepackage{xcolor}
\usepackage{graphicx}
\usepackage[left=16mm,right=16mm,top=35mm,columnsep=15pt]{geometry} 
\usepackage{adjustbox}
\usepackage{placeins}
\usepackage[T1]{fontenc}
\usepackage{lipsum}
\usepackage{csquotes}

\usepackage{braket}
\usepackage{algorithm}
\usepackage[noend]{algpseudocode}
\def\Cbb{\mathbb{C}}
\def\P{\mathbb{P}}

\def\HC{\mathcal{H}}

\def\LC{\mathcal{L}}

\def\ad{^{\dagger}}

\usepackage[makeroom]{cancel}
\usepackage[toc,page]{appendix}
\usepackage[colorlinks=true,citecolor=blue,linkcolor=magenta]{hyperref}

\newcommand{\dya}[1]{\ket{#1}\!\bra{#1}}

\newcommand{\poly}{\operatorname{poly}}

\newcommand{\DC}{\mathcal{D}}

\newcommand{\GC}{\mathcal{G}}
\newcommand{\IC}{\mathcal{I}}

\newcommand{\OC}{\mathcal{O}}

\newcommand{\SC}{\mathcal{S}}

\newcommand{\UC}{\mathcal{U}}

\newcommand{\supp}{\text{supp}}

\renewcommand{\geq}{\geqslant}
\renewcommand{\leq}{\leqslant}

\renewcommand{\Re}{\text{Re}}
\renewcommand{\Im}{\text{Im}}

\DeclareMathOperator*{\argmin}{arg\,min}
\renewcommand{\vec}[1]{\boldsymbol{#1}}  

\newcommand*{\id}{\openone}

\newcommand{\bs}{\textsf{BS}}

\renewcommand{\th}{\theta } 

\newcommand{\thv}{\vec{\theta}}
\newcommand{\delv}{\vec{\delta}} 

\newcommand{\uba}{Departamento de F\'isica ``J. J. Giambiagi'' and IFIBA, FCEyN, Universidad de Buenos Aires, 1428 Buenos Aires, Argentina}
\newcommand{\losalamos}{Theoretical Division, Los Alamos National Laboratory, Los Alamos, New Mexico 87545, USA}

\def\be{\begin{equation}}
\def\ee{\end{equation}}
\def\bs{\begin{split}}
\def\e{\end{split}}
\def\ba{\begin{eqnarray}}
\def\bea{\begin{eqnarray}}

\def\tea{\end{eqnarray}}
\def\ea{\end{eqnarray}}
\def\eea{\end{eqnarray}}

\def\R{\mathds{R}}

\def\k{^{(k)}}

\def\liea{\mathfrak{k}}
\def\liea{\mathfrak{g}}

\def\lieg{\mathds{G}}

\def\H{\tilde{H}}

\def\P{\mathbb{P}}

\def\k{^{(k)}}

\def\liea{\mathfrak{k}}
\def\liea{\mathfrak{g}}

\def\lieg{\mathds{G}}

\def\P{\mathbb{P}}

\newcommand{\veryshortarrow}[1][3pt]{\mathrel{%
   \hbox{\rule[\dimexpr\fontdimen22\textfont2-.2pt\relax]{#1}{.4pt}}%
   \mkern-4mu\hbox{\usefont{U}{lasy}{m}{n}\symbol{41}}}}
\def\vsa{\veryshortarrow}

\newcommand\TFIM{\text{TFIM}}

\newtheorem{theorem}{Theorem}

\newtheorem{proposition}{Proposition}

\newtheorem{definition}{Definition}

\usepackage{amssymb}
\usepackage{dsfont}

\def\be{\begin{equation}}
\def\te{\end{equation}}
\def\ee{\end{equation}}
\def\ba{\begin{eqnarray}}
\def\bea{\begin{eqnarray}}

\def\tea{\end{eqnarray}}
\def\ea{\end{eqnarray}}
\def\eea{\end{eqnarray}}
\def\S{^{(\SC)}}

\begin{document}
\title{Theory of overparametrization in quantum neural networks}

\author{Mart\'{i}n Larocca}
\thanks{The first two authors contributed equally to this work.}
\affiliation{\losalamos}
\affiliation{\uba}

\author{Nathan Ju}
\thanks{The first two authors contributed equally to this work.}
\affiliation{\losalamos}

\author{Diego Garc\'{i}a-Mart\'{i}n}
\affiliation{\losalamos}
\affiliation{Barcelona Supercomputing Center, Barcelona 08034, Spain}
\affiliation{Instituto de F\'{i}sica Te\'{o}rica, UAM-CSIC, Madrid 28049, Spain}

\author{Patrick J. Coles}
\affiliation{\losalamos}

\author{M. Cerezo}
\affiliation{Information Sciences, Los Alamos National Laboratory, Los Alamos, NM 87545, USA}
\affiliation{\losalamos}
\affiliation{Center for Nonlinear Studies, Los Alamos National Laboratory, Los Alamos, New Mexico 87545, USA}


\begin{abstract}
The prospect of achieving quantum advantage with Quantum Neural Networks (QNNs) is exciting. Understanding how QNN properties (e.g., the number of parameters $M$)  affect the loss landscape  is crucial to the design of scalable QNN architectures. Here, we rigorously analyze the overparametrization phenomenon in QNNs with periodic structure. We define overparametrization as the regime where the QNN has more than a critical number of parameters $M_c$ that allows it to explore all relevant directions in state space. Our main results show that the dimension of the Lie algebra obtained from the generators of the QNN is an upper bound for $M_c$, and for the maximal rank that the quantum Fisher information and Hessian matrices can reach. Underparametrized  QNNs  have spurious local minima in the loss landscape that start disappearing when $M\geq M_c$. Thus, the overparametrization onset corresponds to a computational phase transition where the QNN trainability is greatly improved by a more favorable landscape. We then connect the notion of overparametrization to the QNN capacity, so that when a QNN is overparametrized, its capacity achieves its maximum possible value. We run numerical simulations for eigensolver, compilation, and autoencoding applications to showcase the overparametrization computational phase transition. We note that our results also apply to variational quantum algorithms and quantum optimal control.
\end{abstract}

\maketitle

\section{Introduction}

The development of Neural Networks (NNs) and  Machine Learning (ML) is one of the greatest scientific revolutions of the twentieth century. Traditionally,  computers were explicitly programmed to solve a task, so that a user-created code would take an input and produce a desired output. In ML, however, one follows a fundamentally different approach. Here, a computer is trained to learn from data, with the goal that  it can accurately solve the problem when presented with new and previously unseen cases~\cite{mohri2018foundations}. Currently, ML is used in virtually all areas of science, with applications such as drug discovery~\cite{vamathevan2019applications}, new materials exploration~\cite{schmidt2019recent}, and self-driving cars~\cite{grigorescu2020survey}.

Despite their tremendous success, training NNs is a difficult task that has even been shown to be  NP-hard~\cite{blum1992training,daniely2016complexity,boob2020complexity}. Thus, finding ways to improve the NNs trainability and generalization capacity has always been a coveted goal. Towards this end, one of the most surprising phenomena in ML  is that of overparametrization. Here, one trains a NN with a capacity larger than that which is necessary to represent the distribution of the training data~\cite{neyshabur2018role}. Usually, this implies having a number of parameters in the NN that is much larger than the number of training points~\cite{zhang2021understanding}. Naively, one could expect that a model with a large capacity would have training difficulties and also have overfitting (poor generalization). However, overparametrizing a  NN can improve its performance and reduce its training and generalization errors~\cite{zhang2021understanding,allen2019convergence,allen2019learning,du2019gradient,buhai2020empirical}, and even lead to provable convergence results~\cite{du2018gradient,brutzkus2018sgd}.

The advent of quantum computers~\cite{nielsen2000quantum,preskill2018quantum} has brought a tremendous interest in using these devices for data science. Here, researchers have embedded ML into the framework of quantum mechanics, with the new, generalized theory being called  Quantum Machine Learning (QML)~\cite{schuld2015introduction,biamonte2017quantum,cerezo2020variationalreview}. With QML, the end goal is not formal generalization but rather to exploit entanglement and superposition to achieve a quantum advantage~\cite{huang2021information,huang2021power,kubler2021inductive,abbas2020power}, that is, to solve the problem more efficiently than any classical algorithm run on a classical supercomputer.

Naturally, as a generalized theory, QML has the potential to exhibit many of the issues and phenomena exhibited by (classical) ML. For instance, like the classical case, it has been shown that training QML models is NP-hard~\cite{bittel2021training}. Since a QML model may consist of a data embedding followed by a parametrized quantum circuit that is often called a Quantum Neural Network (QNN), its training  requires optimizing the QNN's parameters~\cite{cerezo2020variationalreview,benedetti2019parameterized,beer2020training,cong2019quantum,farhi2018classification}. Recently, much effort has gone towards developing so-called Quantum Landscape Theory~\cite{arrasmith2021equivalence}, which studies the properties of QML loss function landscapes. Indeed, there are results analyzing the presence of sub-optimal local minima~\cite{rivera2021avoiding,wierichs2020avoiding}, the existence of barren plateaus~\cite{mcclean2018barren,cerezo2020cost,marrero2020entanglement,patti2020entanglement,larocca2021diagnosing,holmes2021connecting,cerezo2020impact,holmes2021barren,arrasmith2020effect,huembeli2021characterizing,pesah2020absence,sharma2020trainability}, and how quantum noise affects the loss landscape~\cite{wang2020noise,fontana2020optimizing,wang2021can,campos2021training,franca2020limitations}.

Similar to classical NNs, some examples of QNNs that exhibit overparametrization have been constructed~\cite{wiersema2020exploring,zhang2020overparametrization,wierichs2020avoiding,wiersema2020exploring,kiani2020learning,funcke2021best,lee2021towards,anschuetz2021critical}. Some of these works have heuristically shown that increasing the number of parameters in the QNN can improve its trainability and lead to faster convergence. However, there is still need for a detailed theoretical analysis of this overparametrization phenomenon. Understanding overparametrization is crucial for Quantum Landscape Theory and for engineering QNNs to enhance their trainability.

In this work we provide a theoretical framework for the overparametrization of QNNs. Our main results indicate that, for a general type of periodic-structured QNNs, one can reach an overparametrized regime by increasing the number of parameters  past some threshold critical value $M_c$  (see Fig.~\ref{fig:schematic}(a)). Moreover, we prove that $M_c$  is related to the dimension of the Dynamical Lie Algebra (DLA)~\cite{dalessandro2010introduction,zeier2011symmetry} associated with the generators of the QNN.

We here define overparametrization as the QNN  having enough parameters so that the quantum Fisher information matrix saturates its achievable rank. In this case, one can explore all relevant directions in the state space by varying the QNN parameters. We then relate this notion of overparametrization to different measures of the model's capacity ~\cite{abbas2020power,haug2021capacity}, so that a model is overparametrized when its capacity is saturated.  Then, as shown in Fig.~\ref{fig:schematic}(b), our results have direct implications in understanding why overparametrization can improve the model's trainability, as the overparametrization onset corresponds to a computational phase transition~\cite{kiani2020learning}. We verify our theoretical results by performing numerical simulations. In all cases, we find the predicted computational phase transition, where the success probability of solving the optimization problem is greatly increased after a critical number of parameters.

These results provide theoretical grounds for  recent observations of the overparametrization phenomenon in QML~\cite{wiersema2020exploring,kiani2020learning,kim2021universal}. Moreover, our theorems have direct consequences for the field of quantum optimal control~\cite{d2007introduction,chakrabarti2007quantum,larocca2020exploiting,larocca2020fourier}.

\begin{figure}[t]
	\includegraphics[width= .9\columnwidth]{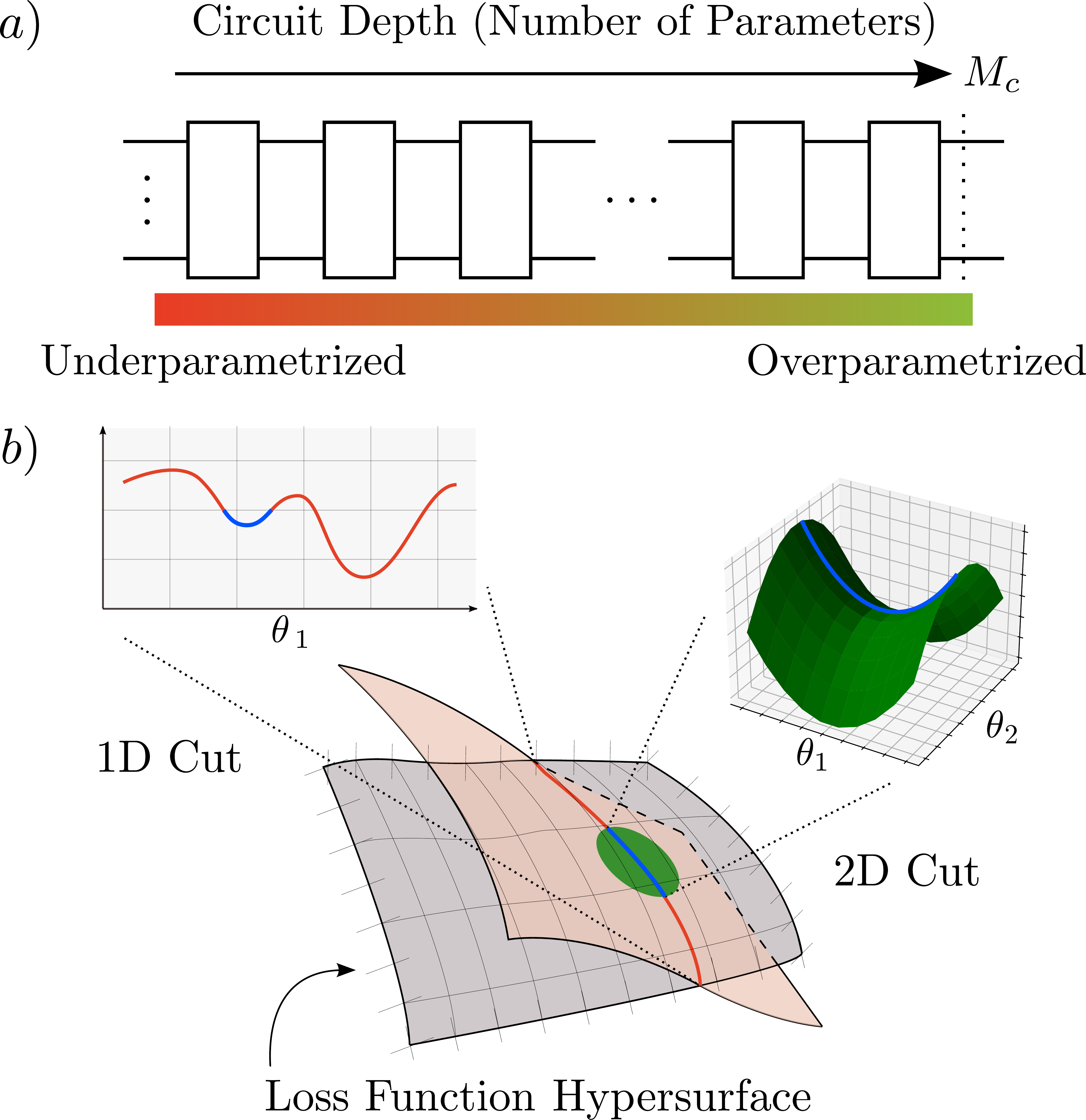}
	\caption{\textbf{Overparametrization in quantum neural networks (QNNs).} a) Quantum circuit description of the QNN. By having a low (high) number of parameters one is not able (is able) to explore all relevant directions in the Hilbert space, and thus the QNN is underparametrized (overparametrized).  b)~The gray surface corresponds to the unconstrained loss function landscape. An underparametrized QNN explores a low dimensional cut of the loss function (1D cut over the red lines). Here, the optimizer can get trapped in spurious local minima (blue segment) that negatively impact the parameter optimization. By increasing the number of parameters past some threshold $M_c$, one can explore a higher dimensional cut of the landscape (2D cut over the green region). As shown, some previous spurious local minima correspond to saddle points (blue segment), and the optimizer can escape the false trap.  }
	\label{fig:schematic}
\end{figure}

\section{Results}

\subsection{Quantum Neural Networks}

 \begin{figure*}[th!]
	\includegraphics[width=1\linewidth]{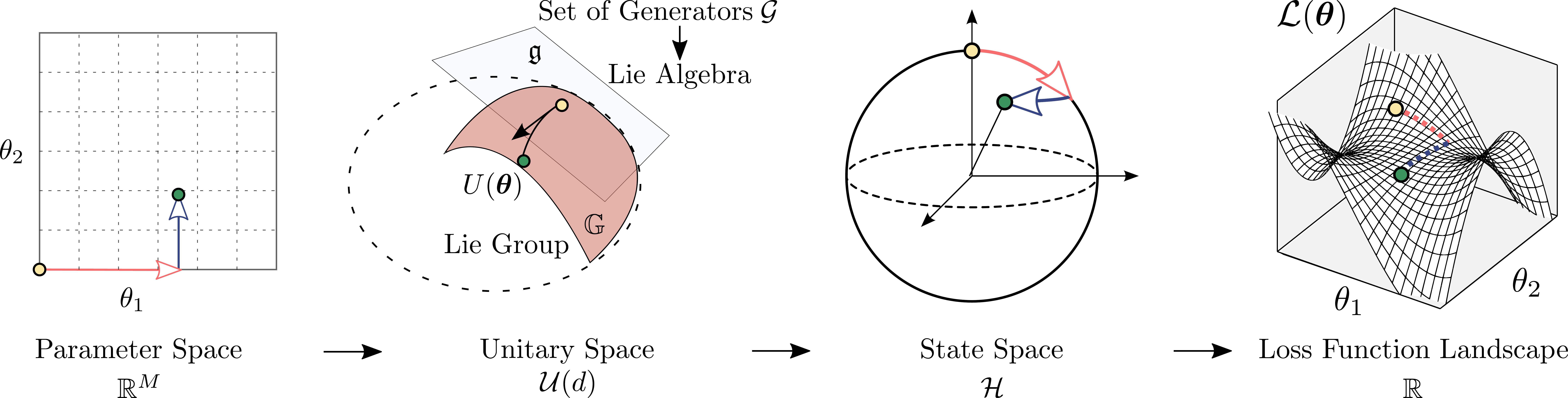}
	\caption{\textbf{Relevant mathematical spaces for QNNs.} QNNs employ a set of $M$  trainable parameters $\thv\in\mathbb{R}^M$, which live in parameter space. The QNN itself is represented by a $d$-dimensional unitary $U(\thv)$, which lives in unitary space. An exemplary form of $U(\thv)$ is that of Eq.~\eqref{eq:PSA_ansatz}, where the set of unitaries depends on the dynamical Lie algebra $\liea$, which in turn is obtained from the set of generators $\GC$ in Definition~\ref{def:generators}. In most QML applications, the QNN acts on an input state $\ket{\psi_\mu}$ from a training set. Thus, the set of reachable unitaries of the ansatz translates into a set of reachable states in the Hilbert space $\HC$. Finally, by performing measurements on a quantum computer one estimates the loss function or its gradient. This information is then used to navigate through the loss function landscape $\LC(\thv)$.  Understanding the connections between these mathematical spaces is fundamental for the theory of quantum landscapes.       }
	\label{fig:surjectivity}
\end{figure*}

Quantum Neural Networks (QNNs)~\cite{schuld2015introduction,biamonte2017quantum,cerezo2020variationalreview} employ parametrized quantum circuits to allow for task-oriented programming of quantum computers. Here, one encodes the problem of interest in a loss function $\LC(\thv)$, whose minima correspond to the task's solution.  Using data from a training dataset $\SC$ composed of  quantum states $\ket{\psi_\mu}\in\SC$, one optimizes the QNN parameters to solve the  problem
\begin{equation}\label{eq:optimization}
    \thv_*=\argmin_{\thv}\LC(\thv).
\end{equation}
Measurements on a quantum computer assist in estimating the loss function (or its gradients), while a classical optimizer is used to update the parameters and solve Eq.~\eqref{eq:optimization}. This hybrid scheme allows the QML model to access the exponentially large dimension of the Hilbert space, with the hope that if the whole process is hard to classically simulate, then a quantum advantage could be achieved~\cite{boixo2018characterizing,google2019supremacy,huang2021power}.

We consider the case when the QNN  is a parametrized quantum circuit $U(\thv)$ that acts on the quantum states in the training set as $U(\thv)\ket{\psi_\mu}$. Here,  $U(\thv)$ has an $L$-layered periodic structure  of the form
\begin{equation}\label{eq:PSA_ansatz}
    U(\thv)=\prod_{l=1}^LU_l(\thv_l)\,, \quad U_l(\thv_l)=\prod_{k=1}^K e^{-i \theta_{lk}H_k}\,,    
\end{equation}
where the index $l$ indicates the layer,  and the index $k$ spans the traceless Hermitian operators $H_k$ that generate the unitaries in the ansatz. Moreover,  $\thv_{l}=(\theta_{l1},\ldots\theta_{lK})$ are the parameters in a single layer, and $\thv=\{\thv_{1},\ldots,\thv_L\}$ denotes the set of $M=K\cdot L$ trainable parameters in the QNN.

As discussed in~\cite{larocca2021diagnosing}, Eq.~\eqref{eq:PSA_ansatz} contains as special cases the hardware-efficient ansatz~\cite{kandala2017hardware}, quantum alternating operator ansatz (QAOA)~\cite{farhi2014quantum,hadfield2019quantum}, Adaptive QAOA~\cite{zhu2020adaptive}, Hamiltonian Variational Ansatz (HVA)~\cite{wecker2015progress}, and Quantum Optimal Control Ansatz~\cite{choquette2020quantum}, among others~\cite{lee2021towards}. As we discuss in the Methods section,  due to the close connection between training a parametrized quantum circuit and the control pulses used to evolved a quantum state in a quantum optimal control protocol, all the results derived hereon can be directly applied to the field of quantum optimal control.

 \subsection{Quantum Landscape Theory}

The usefulness of a QNN for a given task hinges on several factors. First and foremost, it is crucial that a solution (or a good approximation to it) actually exists within the ansatz. Then, even if that solution exists, one must be able to find the associated optimal parameters. The goals of Quantum Landscape Theory are to study properties of the QML loss landscape, how they emerge, and how they affect the optimization process. Here we recall the basic theoretical framework of Quantum Landscape Theory.



First, we note that there are several  aspects of the problem that play a key role in how the loss function landscape arises. Specifically, as shown in Fig.~\ref{fig:surjectivity}, $\thv$  is  a vector in $\mathbb{R}^M$, and each set of parameters corresponds to a unitary $U(\thv)$ in the unitary group $\UC(d)$ of degree $d$. Then, one applies the unitary  $U(\thv)$ to an  $n$-qubit input state $\ket{\psi_\mu}$ (from the dataset $\SC$) in a Hilbert space $\HC$ of dimension $d=2^n$.   Finally, the loss function value $\LC(\thv)\in \mathbb{R}$ is determined by performing measurements over the states $U(\thv)\ket{\psi_\mu}$. In this sense, the action of the QML model arises from the composition of the following three maps:
\begin{equation}\label{eq:maps}
\mathbb{R}^M\xrightarrow[]{} \UC(d) \xrightarrow[]{} \HC \xrightarrow[]{} \mathbb{R}\,.
\end{equation}
Since the landscape is essentially the collection of values obtained at the end of the maps in Eq.~\eqref{eq:maps}, understanding each step of this process is crucial to understanding the properties of the landscape.

Let us consider the first  map in Eq.~\eqref{eq:maps}, i.e., the map between the space of parameters and the unitary group. It has been shown that the unitaries generated by the ansatz in Eq.~\eqref{eq:PSA_ansatz} are characterized via the so-called Dynamical Lie Algebra (DLA)~\cite{dalessandro2010introduction,zeier2011symmetry}. Specifically, consider the following definition. 
\begin{definition}[Set of generators $\GC$]\label{def:generators} Consider a parametrized quantum circuit of the form \eqref{eq:PSA_ansatz}. The set of generators $\GC=\{H_k\}_{k=1}^K$ is defined as the set (of size $|\GC|=K$) of the Hermitian operators that generate the unitaries in a single layer of $U(\thv)$.
\end{definition}
Then, the DLA is defined as follows.
\begin{definition}[Dynamical Lie Algebra (DLA)]\label{def:dynamical_lie_algebra} Consider a set of generators $\GC$ according to Definition~\ref{def:generators}. The DLA $\liea$ is generated by repeated nested commutators of the operators in $\GC$. That is, 
\begin{equation}
\liea={\rm span}\left\langle iH_1, \ldots, iH_K \right\rangle_{Lie}\,,
\end{equation}
where $\left\langle S\right\rangle_{Lie}$ denotes the Lie closure, i.e., the set obtained by repeatedly taking the commutator of the elements in $S$. 
\end{definition}

Recall that the set of reachable unitaries $\{ U(\thv)\}_{\thv} \subseteq \mathbb{G} \subseteq \SC\UC(d)$ obtained from arbitrary choices of $\thv$ forms itself a Lie group, known as the dynamical Lie group $\mathbb{G}$. Then, we note that $\mathbb{G}$ is fully obtained from the DLA as  $\mathbb{G} = e^{\liea}$~\cite{larocca2021diagnosing,morales2020universality}.
We refer the reader to the Methods section for some intuitive understanding on the role of the DLA.

Here, we should remark that the optimal choice of ansatz (or equivalently, the best choice of generators) for a given task is still an open question. While a natural choice would be to use a QNN that is as expressible as possible~\cite{sim2019expressibility}, it has been shown that such choice can lead to trainability issues such as barren plateaus~\cite{mcclean2018barren,holmes2021connecting,larocca2021diagnosing}. 

We can now analyze the second map in Eq.~\eqref{eq:maps}, i.e., the map leading to quantum states in a Hilbert space. Given the fact that $\liea$ determines the set of reachable unitaries, and recalling that the QNN acts on the states $\ket{\psi_\mu}$ in the training set $\SC$ as $U(\thv)\ket{\psi_\mu}$, then the set of reachable states (i.e., the orbit) is, in turn, also directly determined by the DLA. We note that in many cases the set of generators can have symmetries, in which case the DLA is of the form $\liea=\bigoplus_\nu \liea_\nu $. Here, $\nu$ is an index over the invariant subspaces. The states in the training set need not respect some, or any, of the symmetries of the QNN. In this work, we consider the case where the states in the training set respect some of the symmetries, and we denote as $\liea_{\SC}$ the DLA associated with the symmetries preserved by the states in $\SC$. The limiting case when the states in $\SC$ break all symmetries in the ansatz (or when the ansatz has no symmetries) corresponds to $\liea_{\SC}=\liea$.

Here, one can study the set of reachable states through the action of $U(\thv)$ on $\ket{\psi_\mu}$ as follows. Given a set of parameters $\thv$ and an infinitesimal perturbation $\vec{\delta}$ (possibly obtained from some update rule), it is  useful to quantify the distance $\DC$ between the quantum states $\ket{\psi_\mu(\thv)}=U(\thv)\ket{\psi_\mu}$ and $\ket{\psi_\mu(\thv+\vec{\delta})}=U(\thv+\vec{\delta})\ket{\psi_\mu}$. The second-order Taylor expansion of $\DC$ is given by the Fubini-Study metric~\cite{cheng2010quantum,meyer2021fisher} as 
\begin{equation}
    \DC(\ket{\psi_\mu(\thv)},\ket{\psi_\mu(\thv+\vec{\delta})})=\frac{1}{2}\delv^T\cdot F_\mu(\thv)\cdot\delv\,.
\end{equation}
Here, $F_\mu(\thv)$ is the Quantum Fisher Information Matrix (QFIM) for the state $\ket{\psi_\mu}$. The QFIM is an  $M\times M$ matrix whose elements are~\cite{liu2019quantum}
\begin{align}\label{eq:QFIM-elem}
\begin{split}
[F_\mu(\thv)]_{ij}\!=\!4\Re[&\braket{\partial_i\psi_\mu(\thv)}{\partial_j\psi_\mu(\thv)}\\
&-\braket{\partial_i\psi_\mu(\thv)}{\psi_\mu(\thv)}\braket{\psi_\mu(\thv)}{\partial_j\psi_\mu(\thv)}]\,,
\end{split}
\end{align}
where $\ket{\partial_i\psi_\mu(\thv)}=\partial \ket{\psi_\mu(\thv)}/\partial\theta_i=\partial_i\ket{\psi_\mu(\thv)}$ for $\theta_i\in\thv$. 
The QFIM  plays a crucial role in imaginary time evolution algorithms~\cite{mcardle2019variational}, and in quantum-aware optimizers such as the quantum natural gradient descent~\cite{stokes2020quantum,koczor2019quantum,gacon2021simultaneous,haug2021natural}. Moreover, we recall that the rank of the QFIM quantifies the number of independent directions in state space that can be explored by making an infinitesimal change in $\thv$.

Finally, consider the third  map in Eq.~\eqref{eq:maps}, i.e., the map leading to the loss function value. Similar to how the QFIM is related to the changes in state space arising by a change in the parameters, one can also quantify  how much the loss function value changes by a small parameter update. In this case, one can study the curvature of the loss landscape via the Hessian matrix $\nabla^2 \LC(\thv)$, an $M\times M$ matrix whose elements are defined as
\begin{equation}
    [\nabla^2\LC(\thv)]_{ij}=\partial_i\partial_j \LC(\thv)\,.
\end{equation}
Evaluating the gradient and the Hessian at a given point allows one to construct a quadratic model of the loss function, with the Hessian eigenvectors associated with positive (negative) eigenvalues determining directions of positive (negative) curvature. Thus, the rank of $\nabla^2\LC(\thv)$ is related to the number of directions that lead to (second order) changes in the loss, as a zero-valued eigenvalue indicates a zero-curvature flat direction. We finally note that the Hessian has been used to characterize the loss landscapes of variational quantum algorithms~\cite{huembeli2021characterizing,kim2021quantum,dalgaard2021predicting,dalgaard2020hessian}.

\subsection{Theoretical Results} \label{sec:Main_results}

Here we present our main results, where we rigorously analyze the overparametrization phenomenon in QNNs. Our results prove that: 1) there exists a critical number of parameters $M_c$ needed to overparametrize a QNN, and 2) that $M_c$, and the onset of overparametrization, can be related to the dimension of the associated DLA. The proofs of our main results are sketched in the Methods section and formally derived in the Supplementary Information. For our main results in Theorem~\ref{theo:1} and Theorem~\ref{theo:2}, we make no assumption on the loss function other than the QNN  acting on the states $\ket{\psi_\mu}$ in the training set $\SC$ as $U(\thv)\ket{\psi_\mu}$ and that the loss is estimated via measurements on these evolved states. Then, for Theorem~\ref{theo:3} we consider special cases of such loss functions. 

First, consider the following definition. 
\begin{definition}[Overparametrization]\label{def:overparametrization}
    A QNN is said to be overparametrized if the number of parameters $M$ is such that the QFI matrices, for all the states in the training set, simultaneously saturate their achievable rank $R_\mu$ at least in one point of the loss landscape. That is, if increasing the number of parameters past some minimal (critical) value $M_c$ does not further increase the rank of any QFIM:
    \begin{equation}\label{eq:Rmu}
        \max_{M\geq M_c,\thv}\rank[F_\mu(\thv)] = R_\mu\,.
    \end{equation}
\end{definition}
In the Methods section we give additional motivation for this definition, as well as present an equivalent definition that further highlights the geometrical nature of the overparametrization phenomenon. 

According to Definition~\ref{def:overparametrization}, when the QNN is overparametrized, one can explore all relevant and independent directions in the state space by changing the parameters of the ansatz. Evidently, since the rank of the QFIM is at most equal to $M$, then Definition~\ref{def:overparametrization} implies that $M_c$ must be such that $M_c\geq \max_\mu R_\mu$. We also remark that the overparametrization  is here defined for the  QFIM ranks to be equal to $R_\mu$ on a single point in the landscape. In principle, the QFIM could achieve its maximum  rank in a given point, and not in others. However, as we numerically verify (see Supplementary Information), at the overparametrization onset the QFIM saturates its rank almost everywhere in the landscape simultaneously. Here, increasing the number of parameters will not further increase the number of accessible directions in state space. However, it can still be beneficial to add more parameters as this will lead to global minima with higher degeneracy~\cite{moore2012exploring,larocca2020exploiting,larocca2020fourier,fontana2020optimizing}.

In light of Definition~\ref{def:overparametrization}, overparametrization has implications for the trainability of the QNN parameters.  If the QNN is underparametrized, the loss landscape can exhibit spurious, or false, local minima~\cite{wu2012singularities,riviello2014searching,rach2015dressing,larocca2018quantum}. However, by increasing the number of parameters and overparametrizing the QNN, one can explore more directions in state space, and hence the optimizer is able to escape these false minima. As such, crossing the overparametrization threshold can be considered as a computational phase transition~\cite{kiani2020learning} where a more favorable landscape ameliorates the optimization.

Then, our first main result to understand how overparametrization can improve the trainability is as follows.
\begin{theorem}\label{theo:1}
    For each state $\ket{\psi_\mu}$ in the training set $\SC$, the maximum rank $R_\mu$ of its associated QFIM (defined in Eq.~\eqref{eq:QFIM-elem}) is upper bounded  as
    \begin{equation}
        R_\mu\leq\dim(\liea_\SC)\,.
    \end{equation}
\end{theorem}
We remark that $\dim(\liea_\SC)\leq \dim(\liea)$, and hence $\dim(\liea)$ also upper bounds $R_\mu$. Theorem~\ref{theo:1} shows that, at most, the QNN can explore  $\dim(\liea_\SC)$ relevant and independent directions in the state space. And thus, a sufficient condition for overparametrization is that
\begin{equation}
    M\geq \dim(\liea_\SC)\,.
\end{equation}
Note here that the number of parameters for overparametrization depends on the data in $\SC$ and on the set of generators $\GC$.  The latter implies that: 1) Different ansatzes for the QNN can be overparametrized for different depths even when using the same dataset, 2) The same QNN ansatz can reach overparametrization for different depths when used for two different datasets.

Then, as shown in the numerical results below, in many cases the QNN is found to be overparametrized when
\begin{equation}\label{eq:crit-params}
    M\sim\dim(\liea_\SC)\,.
\end{equation}
Evidently, $M$ in Eq.~\eqref{eq:crit-params} can be intractable for ansatzes where $\dim(\liea_\SC)\in\OC(b^n)$ with $b>1$ (e.g. controllable systems~\cite{larocca2021diagnosing}). More promising, however, are QNNs where $\rm dim(\liea_\SC) \in\OC(\poly(n))$, as here the QNN can be overparametrized for a number of parameters $M\in\OC(\poly(n))$. Below we show examples of ansatzes that can achieve overparametrization with polynomially deep circuits.

Here we note that Definition~\ref{def:overparametrization} allows us to connect the notion of overparametrization to that of the QNN's capacity. We recall that the capacity (or power) of a QNN quantifies the breadth of functions
that it can capture~\cite{coles2021seeking}. While there is no unique definition of capacity, we here consider two definitions for the so-called effective quantum dimension, which measures the power of the QNN. First, following~\cite{haug2021capacity}, we can define the average effective quantum dimension of a QNN:
\begin{equation}\label{eq:eff_dim_1}
    D_1(\thv)=\mathbb{E}\left[\sum_{i=1}^M\IC(\lambda^{i}_\mu(\thv))\right]\,,
\end{equation}
where $\lambda^{i}_\mu(\thv)$ are the eigenvalues of the QFIM for the state $\ket{\psi_\mu}$, and where $\IC(x)=0$ for $x=0$, and $\IC(x)=1$ for $x\neq1$. Here the expectation value is taken over the probability distribution  that samples input states from the dataset.

The second definition follows from~\cite{abbas2020power}. In the $n\rightarrow\infty$ limit, the effective quantum dimension of~\cite{abbas2020power} converges to
\begin{equation}\label{eq:eff_dim_2}
    D_2=\max_{\thv}\left(\rank\left[\widetilde{F}(\thv)\right]\right)\,,
\end{equation}
where $\widetilde{F}(\thv)$ is the classical Fisher Information matrix obtained as
\small
\begin{equation}
    \widetilde{F}(\thv)=\mathbb{E}\left[\frac{\partial\log(p(\ket{\psi_\mu},y_\mu;\thv))}{\partial\thv}\frac{\partial\log(p(\ket{\psi_\mu},y_\mu;\thv))}{\partial\thv}^T\right].
\end{equation}
\normalsize 
Here, $p(\ket{\psi},y;\thv)$, describes the joint relationship between an input $\ket{\psi}$ and an output $y$ of the QNN. In addition, the expectation value is taken over the probability distribution that samples input states from the dataset.

Then, the following theorem holds.
\begin{theorem}\label{theo:2}
The model capacity, as quantified by the effective dimensions of Eqs.~\eqref{eq:eff_dim_1} or~\eqref{eq:eff_dim_2}, is upper bounded as
\begin{equation}
    D_1(\thv)\leq \dim(\liea_\SC),\quad D_2\leq \dim(\liea_\SC)\,.
\end{equation}
Moreover, when the QNN is overparametrized according to Definition~\ref{def:overparametrization}, $D_1(\thv)$ achieves its maximum value on at least one point of the landscape. 
\end{theorem}
Theorem~\ref{theo:2} provides an operational meaning to the overparametrization definition in terms of the model's capacity. Specifically, the onset of the overparametrization arises when the model's capacity in Eq.~\eqref{eq:eff_dim_1} can get saturated. Moreover, we here see that increasing the number of parameters can never increase the model capacity beyond $\dim(\liea_\SC)$.

Note that Definition~\ref{def:overparametrization} relates the overparametrization phenomenon with the rank of the QFIM and the possibility of exploring all relevant directions in the state space. One can also relate the notion of overparametrization with the rank of the Hessian and the relevant directions in the loss function landscape. Consider the case when the loss function is of the form
\begin{equation}\label{eq:lineal-loss}
    \LC(\thv)=\sum_{\ket{\psi_\mu}\in\SC}c_\mu \Tr[U(\thv)\dya{\psi_\mu}U\ad(\thv)O]\,,
\end{equation}
where $c_\mu$ are real coefficients associated with each state $\ket{\psi_\mu}$ in $\SC$, and where $O$ is a Hermitian operator. Such loss functions arise for supervised quantum machine learning~\cite{havlivcek2019supervised,pesah2020absence,sharma2020trainability}, autoencoding~\cite{romero2017quantum}, principal component analysis~\cite{larose2019variational,bravo2020quantum,cerezo2020variational}, dynamical simulation~\cite{yuan2019theory,cirstoiu2020variational,commeau2020variational,gibbs2021long}, and, more generally, for variational quantum algorithms~\cite{cerezo2020variationalreview,bharti2021noisy}. Then, the following theorem holds.
\begin{theorem}\label{theo:3}
  Let $\nabla^2\LC(\thv_*)$  be the Hessian for a loss function of the form of Eq.~\eqref{eq:lineal-loss} evaluated at the optimal set of parameters $\thv_*$. Then, its rank is upper bounded as
    \begin{equation}
        \rank[\nabla^2\LC(\thv_*)]\leq\min \{ \dim(\liea_\SC), 2dr - r^2 - r\}\,,
    \end{equation}
    where $r=\min\{\rank[\sum_{\mu} c_{\mu} \dya{\psi_{\mu}}],\rank[O]\}$, and $d$ is the Hilbert space dimension.
\end{theorem}
Theorem~\ref{theo:3} shows that the maximum number of relevant directions around the global minima of the optimization problem is always smaller than $\dim(\liea_\SC)$. Here, we again numerically find that in the overparametrization regime adding more parameters only adds zero-valued eigenvalues to the Hessian.  We finally remark that Theorem~\ref{theo:3} imposes a maximal rank on the Hessian when evaluated at the solution, but in general the Hessian can have a rank larger than $\dim(\liea_\SC)$ at other points in the landscape. 

Note that in principle one can define overparametrization as the rank of the Hessian being saturated at the solution. However, as discussed in the Methods, this definition could have potential issues.

\subsection{Numerical Results}

Here we numerically illustrate the overparametrization phenomenon and the associated computational phase transition. We consider three different optimization tasks: the Variational Quantum Eigensolver (VQE), unitary compilation, and quantum autoencoding. We note that the overparametrization phenomenon has been empirically observed for the first two tasks respectively in~\cite{wiersema2020exploring,kim2021universal} and~\cite{kiani2020learning}. The simulations were performed with the open-source library \texttt{Qibo} \cite{efthymiou2020qibo,efthymiou2021qibo_zenodo}, and the details can be found in the Supplemental Information.

\subsubsection{Variational Quantum Eigensolver}

First, we use the VQE algorithm \cite{peruzzo2014variational,bravo2020scaling,consiglio2021variational} to minimize the loss function 
\begin{equation} \label{eqn:hva_cost_function}
    E(\vec{\theta}) = \langle \psi(\vec{\theta}) |H_{\TFIM}| \psi(\vec{\theta}) \rangle\,,
\end{equation}
and find the ground state of the Hamiltonian of the transverse field Ising model $H_{\TFIM}$. Here, $| \psi(\vec{\theta}) \rangle=U(\thv)\ket{+}^{\otimes n}$ and  $H_{\TFIM} = -\sum_{i=1}^{n_f} \sigma^z_i\sigma^z_{i+1} - h \sum_{i=1}^n \sigma^x_i$, where  $\sigma^\mu_i$ denotes the $\mu$-Pauli matrix (with $\mu=x,z$) acting on qubit $i$, and $h$ is the strength of the transverse field. We set $h=1$ and consider both open ($n_f=n-1$) and closed ($n_f=n$) boundary conditions. In the latter, $\sigma^\mu_{n+1}=\sigma^\mu_1$. We employ a Hamiltonian variational ansatz for the QNN~\cite{wecker2015progress,wiersema2020exploring}. This ansatz has two parameters per layer and is precisely of the form in~\eqref{eq:PSA_ansatz} (see Methods for a detailed description of the ansatz).

\begin{figure}[t!]
    \centering
    \includegraphics[width=1\columnwidth]{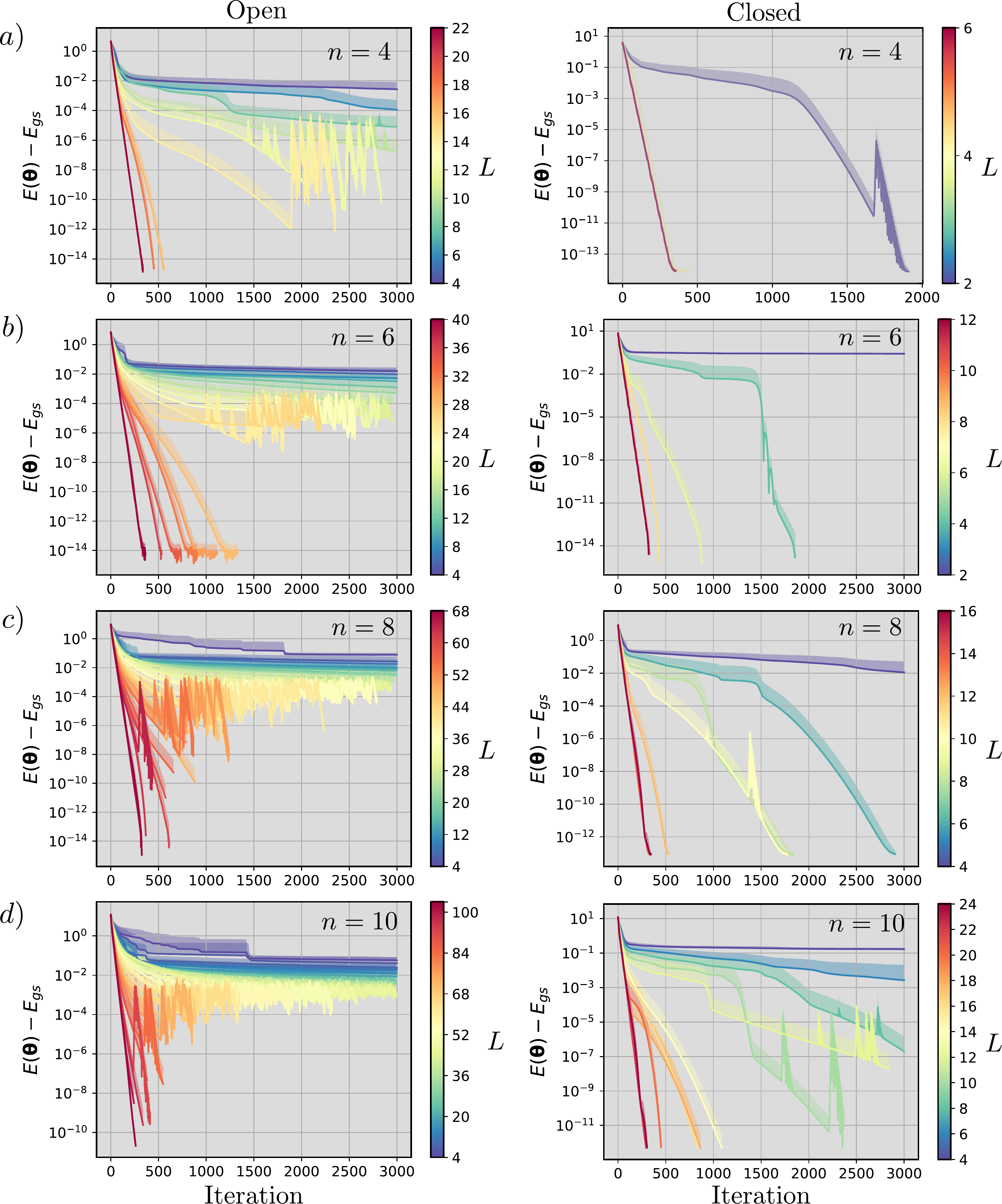}
    \caption{\textbf{Training curves for VQE implementation.} The loss function value minus the exact ground-state energy ($E_{gs}$) is plotted versus iteration. We used a Hamiltonian variational ansatz with open (left) and closed (right) boundary conditions to solve the VQE task in Eq.~\eqref{eqn:hva_cost_function} for a) $n=4$, b) $n=6$, c) $n=8$, and d) $n=10$ qubits. Solid lines represent the average over 50 random initialized runs while the shaded regions correspond to the standard deviation. }
    \label{fig:hva_loss_vs_iter}
\end{figure}

As shown in~\cite{larocca2021diagnosing}, the dimension of the DLA associated with the ansatz is given by~\cite{SI-overparam}
\begin{equation}\label{eqn:HVA_DLAdim} 
 {\rm dim}(\liea^{\rm closed}_\SC) = \frac{3}{2}\,n \quad,\quad {\rm dim}(\liea^{\rm open}_\SC) = n^2\,, 
\end{equation}
where the superscripts indicate closed and open boundary conditions in the ansatz and in $H_{\TFIM}$.   Hence, from our theoretical results, we expect that both of these ansatzes can be overparametrized with only a polynomial number of  parameters.

Figure~\ref{fig:hva_loss_vs_iter} shows the results of minimizing the loss in Eq.~\eqref{eqn:hva_cost_function}, for problem sizes of $n=4,6,8,10$ qubits and for ansatzes with different depths $L$ (i.e.,  $2L$ parameters), with both open and closed boundary conditions. In all cases, we averaged over 50 random parameter initializations.  First, we note that one can always  observe the onset of overparametrization through a computational phase transition  whereby the convergence of the optimization dramatically increases when increasing the number of parameters past some threshold. That is, for a small number of layers, the algorithm is unable to accurately find the ground state, while for a large number of layers the algorithm always rapidly converges to the solution. In fact,  we observe that the loss function decreases exponentially with each optimization step when the number of layers is large enough.

To analyze the number of parameters for which the overparametrization occurs, Figure~\ref{fig:hva_traps}(a) shows the success probability, i.e., the fraction of randomly-initialized instances that converged within $10^{-7}$ of the true solution. Here, one can see the phase transition at the onset of overparametrization. Indeed, at  $M\sim\dim(\liea_\SC)$, the success probability rapidly goes to one. This is due to the fact that the optimization hypersurface becomes more favorable by the removal of false local minima, and thus one can  obtain higher-quality solutions with less iterations. Figure~\ref{fig:hva_traps}(a) also shows that  further increasing the number of parameters past  $\dim(\liea_\SC)$ can in fact lead to the QNN having a higher probability of converging to the solution. There exists a point, however, for  which the overparametrization saturates and there is no visible improvement in convergence speed or quality of the solution found. We found that the saturation number of parameters grows linearly for closed boundary conditions and quadratically for open boundary conditions, and thus these saturation numbers have the same scaling as their corresponding $\dim(\liea_\SC)$.

\begin{figure}[t!]
    \centering
    \includegraphics[width=1\columnwidth]{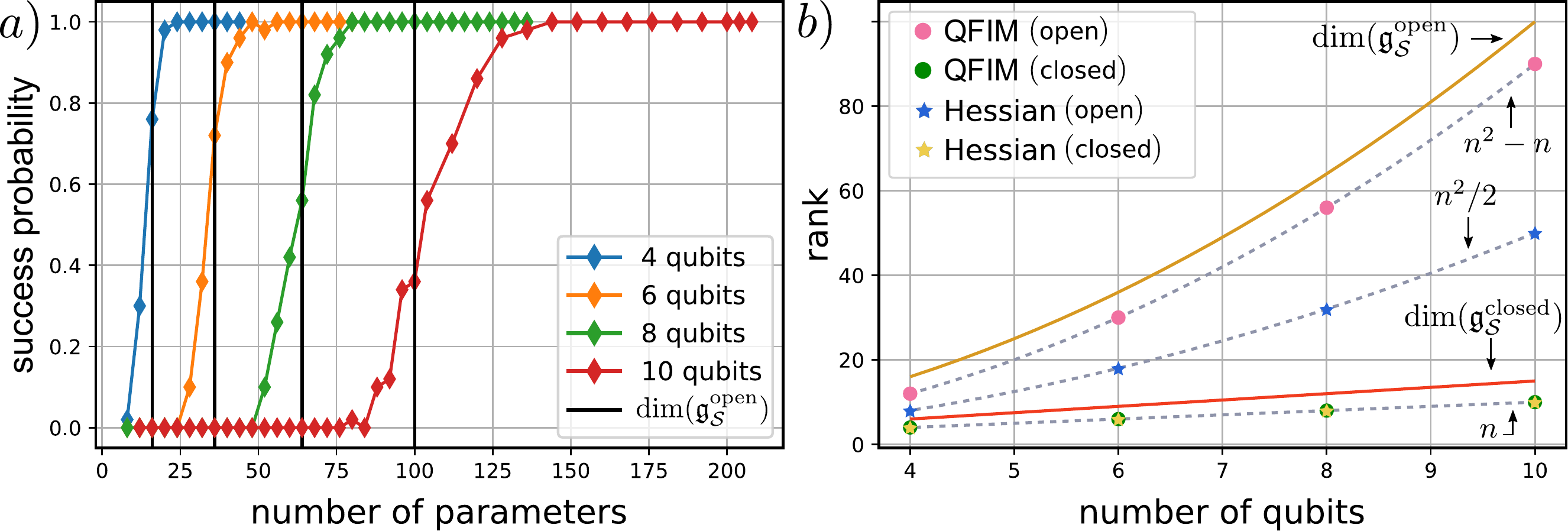}
    \caption{\textbf{Overparametrization threshold for VQE implementation.} a) The success probability (i.e., fraction of instances  that converged to the global optimum within an error of $10^{-7}$) is plotted versus number of parameters. Results are obtained from $50$ randomly-initialized instances, for $n=4-10$ qubits and for the Hamiltonian variational ansatz with open boundary conditions. The vertical black lines indicate the dimension of the DLA. b) The rank of the QFIM and Hessian is plotted versus number of qubits. The rank of the QFIM was evaluated at the optima and at random points in the landscape. The rank of the Hessian was evaluated at the optima. Dashed lines indicate the functional dependence of the computed ranks. Here, the ansatz had a number of parameters for which overparametrization had been fully achieved. }
    \label{fig:hva_traps}
\end{figure}

Finally, Fig.~\ref{fig:hva_traps}(b) shows the computations of the ranks of the QFIM and Hessian at the overparametrization threshold for Hamiltonians and ansatzes with open and closed boundary conditions.  The rank of the QFIM was computed at the global optima, and also at random points in the landscape, and it was found to be the same in all cases. The rank of the Hessian was computed at the global optima. First, let us note that these results show that Theorem~\ref{theo:1} and Theorem~\ref{theo:3} hold, as the dimension of the associated DLA is always an upper bound for the ranks. As shown in the figure by the dashed lines, we can find the explicit dependence for the ranks as a function of the system size. In the Supplementary Information we present additional plots for the ranks of the QFIM and Hessian.

\subsubsection{Unitary compilation}

Let us now consider a unitary compilation task. Unitary compilation refers to decomposing a target unitary into a sequence of control pulses or quantum gates that can be directly implemented on quantum hardware~\cite{chong2017programming,haner2018software,venturelli2018compiling,khatri2019quantum,sharma2019noise}. 

In variational unitary compiling~\cite{khatri2019quantum,sharma2019noise}, one trains a parametrized quantum circuit $U(\thv)$ so that its action matches that of a target unitary $V$ (up to a global phase). Thus, one minimizes the loss function 
\begin{equation}
\label{eq:unitary_cost}
    \LC(\thv)= 1-|\Tr[V\ad U(\vec{\theta})]|^2/d^2\,.
\end{equation}
Here, $\LC(\thv)$ can be efficiently evaluated on a quantum computer with the Hilbert-Schmidt test~\cite{khatri2019quantum}. While  $\LC(\thv)$ is not exactly of the form in~\eqref{eq:lineal-loss}, we also prove in the Methods section a theorem showing that the rank of its Hessian is also upper bounded by $\dim(\liea_\SC)$ at the global optima.

We employ a hardware efficient ansatz~\cite{kandala2017hardware} for $U(\thv)$ composed of alternating layers of single qubit rotations and entangling gates. The number of parameters is therefore $M=2n + L(4n-4)$ (see Methods for a detailed description of the ansatz). We sample the target unitary $V$ from the Haar measure in the unitary group of degree $d$. As shown in~\cite{larocca2021diagnosing}, the dimension of the DLA associated with this ansatz is
\begin{equation} 
\dim(\liea_\SC)= 4^n\,,
\end{equation}
and thus grows exponentially with the number of qubits. 

\begin{figure}[t!]
    \centering
    \includegraphics[width=1\columnwidth]{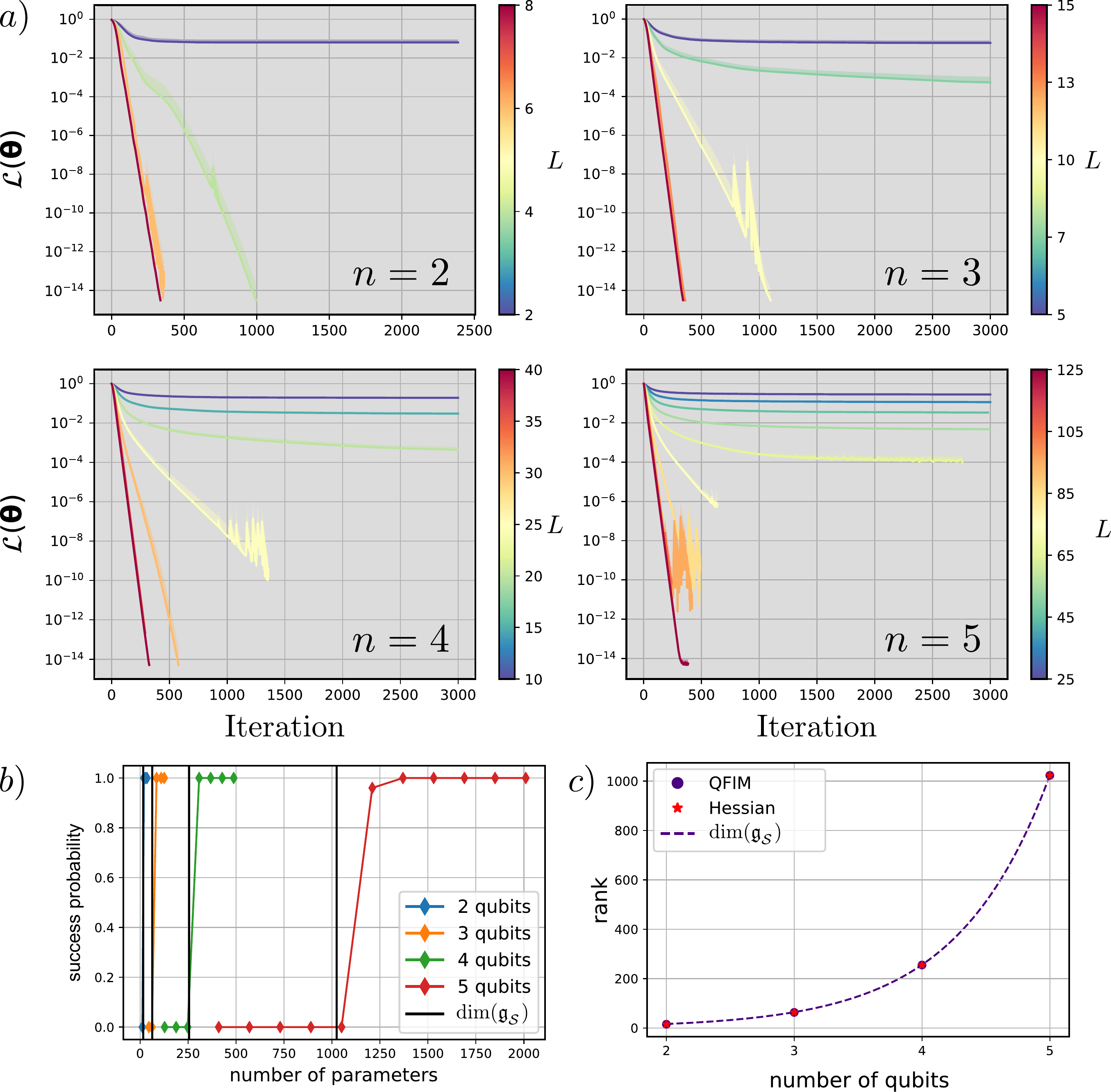}
    \caption{\textbf{Unitary compilation implementation.} a) The loss function is plotted versus iteration. The results are obtained for problem sizes of $n=2,3,4,5$ qubits and for an $L$-layered hardware efficient ansatz. For each point, we averaged the results over $50$ randomly-initialized problem instances. b)~The success probability (i.e., fraction of instances that converged to the global optimum within an error of $10^{-7}$) is plotted versus number of parameters. The vertical lines indicate the dimension of the DLA. c)~The rank of the QFIM and Hessian is plotted versus number of qubits. The rank of the QFIM was evaluated at the optimum and at random points in the landscape. The rank of the Hessian was evaluated at the optimum.  Here, the ansatz had a number of parameters for which overparametrization had been fully achieved.  }
    \label{fig:hea_results}
\end{figure}

Figure~\ref{fig:hea_results}(a) shows the results of minimizing the loss function in Eq.~\eqref{eq:unitary_cost}, for problem sizes of $n=2,3,4,5$ qubits, and for ansatzes with different depths $L$.  In all cases we averaged over $50$ random parameter initializations. Here we can again observe that as the depth of the circuit increases, the convergence towards the global optimum improves dramatically until reaching a saturation point. Figure~\ref{fig:hea_results}(b) plots the success probability for randomly-initialized instances. Similar to the VQE implementation, one finds that around $\dim(\liea_\SC)$ parameters are required to consistently find high-quality solutions, and that the probability of convergence to the global optimum undergoes a drastic phase transition when the number of parameters is around $\dim(\liea_\SC)$. This result again implies a simplification of the optimization landscape, where local traps disappear.  We also numerically verify Theorem~\ref{theo:1} and Theorem~\ref{theo:3}. Namely, Fig.~\ref{fig:hea_results}(c) plots the rank of the QFIM and the Hessian. The QFIM was evaluated at the global optima and at random points in the landscape, while the Hessian was evaluated at the global optima.  For all cases we found that the ranks are equal to $\dim(\liea)-1$. 

\subsubsection{Quantum autoencoding}

Finally, we present results for the archetypal QML task of quantum autoencoding \cite{romero2017quantum,ma2020compression}. A quantum autoencoder is a special type of QNN that can be used to compress quantum information. Analogously to classical autoencoders, the idea is to reduce the dimensionality of the states in a dataset through the action of an encoder $U(\vec{\theta})$. Once compressed, the states belong  to a smaller dimensional Hilbert space known as the latent space. The states compressed into this space can subsequently be recovered with high fidelity at a later time by a decoder $U^\dagger(\vec{\theta})$.

\begin{figure}[t!]
    \centering
    \includegraphics[width=.9\columnwidth]{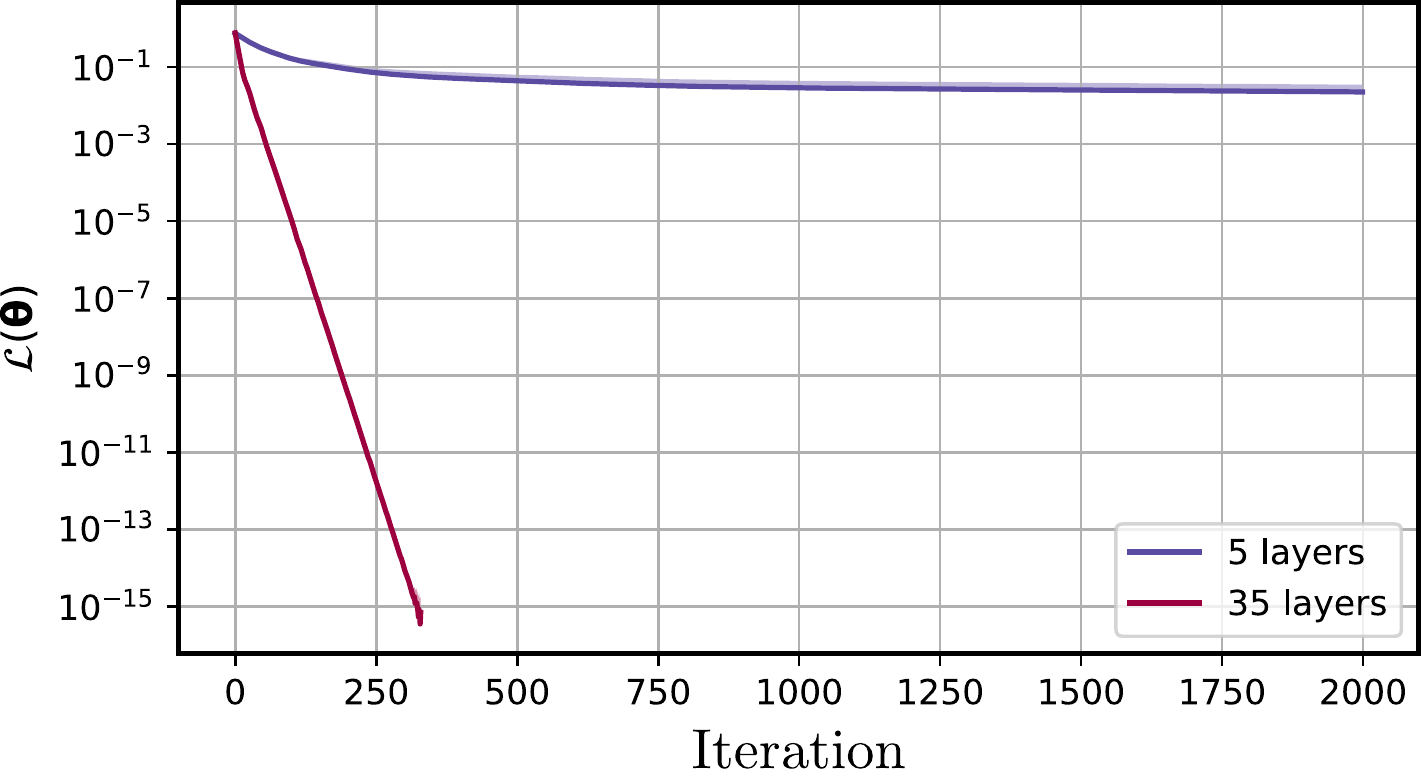}
    \caption{\textbf{Quantum autoencoder implementation.} The loss function value is plotted versus the number of iterations. Results are obtained for a $n=4$ qubit problem when using a layered hardware efficient ansatz with $L=5$ and $L=35$ layers.}
    \label{fig:auto_loss_vs_iter}
\end{figure}

Consider a bipartite quantum system $AB$ of $n_A$ and $n_B$ qubits, respectively, and let $\ket{\psi_\mu}$ be states from the training set $\SC$. The goal of the quantum autoencoder is to train an encoding parametrized quantum circuit $U(\thv)$ to compress the  states in $\SC$ onto subsystem $A$, so that one can discard the qubits in subsystem $B$ without losing information. A possible loss function here is given by~\cite{romero2017quantum}
\small
\begin{equation}\label{eq:cost-AE}
\LC(\vec{\theta}) =\! \sum_{\ket{\psi_\mu}\in\SC}\!\!\!\left( 1-\frac{1}{|\SC|} \Tr[ U(\vec{\theta}) \dya{\psi_\mu}  \,U^\dagger(\vec{\theta}) O]\right)\,,
\end{equation} 
\normalsize
with $O=\dya{0}^{\otimes n_b}\otimes\id_{A}$, and where $\id_{A}$ denotes the identity on subsystem $A$. 
We note that an alternative local version of this loss function was proposed in~\cite{cerezo2020cost} to avoid barren plateaus issues. However, since we here consider small problem sizes, we use the loss in~\eqref{eq:cost-AE}.

Our results were obtained for a system of $n=4$ qubits (with $n_B$=2)  and for the same hardware efficient ansatz used for unitary compilation. The dataset consisted of four states drawn from the NTangled dataset~\cite{schatzki2021entangled}, a quantum dataset composed of states with different amounts and types of multipartite entanglement. As shown in Fig.~\ref{fig:auto_loss_vs_iter}, we can again see that an overparametrized QNN is able to accurately reach the global optima in few iterations. When computing the rank of the QFIM at random points of the landscape, we found that the rank is always $30$. This is in contrast to the dimension of the DLA, which is $\dim(\liea_\SC)=256$, and thus the latter leads to some hope that overparametrization can be achieved with a number of parameters that is much smaller than $\dim(\liea_\SC)$.

Let us here note a crucial difference between the results obtained for the Hamiltonian variational ansatz and for the hardware efficient ansatz. Namely, for the Hamiltonian variational ansatz the dimension of the DLA scaled polynomially with $n$, whereas for the  hardware efficient ansatz the dimension of the DLA grows exponentially with the system size. Thus, in the first case the system becomes overparametrized at a polynomial number of parameters, while the latter case one can require an exponentially large one. This makes it so that overparametrization can be unachievable in practice for large problem sizes when using an ansatz with an exponentially large DLA. In addition, it has been shown that the ansatzes with exponentially large DLAs can exhibit barren plateaus~\cite{larocca2021diagnosing}, thus further preventing their practical use.

\section{Discussion}


Quantum Machine Learning (QML) is an emerging field that aims to analyze (either classical or quantum) data with significant speedup over classical Machine Learning (ML). However, like classical ML, QML also has trainability issues associated with non-convex landscapes, local minima, and the overall NP-hardness of the optimization. 


Classical ML has benefited from the discovery of the overparameterization phenomenon, whereby increasing the number of parameters beyond some threshold causes many local minima to disappear (e.g., as in Fig.~\ref{fig:schematic}). Similarly, preliminary evidence of overparameterization in QML has been discovered for specific constructions of Quantum Neural Networks (QNNs). However, prior to our work, no general theory existed for the precise properties of QNNs that lead to overparameterization.


In this work, we provide the first general analysis of overparameterization for a broad class of QNNs (i.e., those with periodic structure). We find that the Dynamical Lie Algebra (DLA)  obtained from the set of generators of the QNN plays a crucial role in determining properties of the QNN and the ensuing landscape. To our knowledge, our work is the first \textit{algebraic} theory of overparameterization. This represents an important contribution to Quantum Landscape Theory, i.e., the understanding of QML loss function landscapes and how to engineer them.

We defined overparametrization as the QNN having more than a critical number of parameters  that allow it to explore all independent and relevant directions in the state space. This translates to the Quantum Fisher Information Matrices (QFIMs)  having reached their maximum achievable rank. This definition has direct implications for the loss functions of under- and over-parametrized QNNs. Underparametrized QNNs can exhibit spurious, or false, local minima that disappear when one increases the number of parameters and reaches the overparametrization regime. Since the existence of false local minima negatively affect the QNN's trainability, the overparametrization onset corresponds to a computational phase transition where the QNN parameter optimization improves due to a more favorable landscape.

We found that the critical number of parameters needed to overparametrize the QNN is directly linked to the dimension of the associated DLA $\liea_\SC$. Our theorems showed that the rank of the QFIM (across the whole landscape) and the rank of the Hessian (evaluated at the optima) are upper bounded by $\dim(\liea_\SC)$. Thus, one can potentially reach overparametrization if the QNN has $\dim(\liea_\SC)$ parameters. This result is particularly interesting for QNN constructions where $\dim(\liea_\SC)\in\OC(\poly(n))$. Thus, our results show that there can exists QNNs that are overparametrized for a polynomial number of parameters.

We verified our  theoretical  results by performing numerical simulations of problems where the overparametrization had been heuristically observed~\cite{wiersema2020exploring,kiani2020learning}. Here, our theoretical framework allowed us to shed new light and explain some of the observations in these prior works. 

We note that most ansatzes used for QNNs in the literature are ultimately hardware efficient ansatzes. These  are known to exhibit barren plateaus, and in view of our recent results, they may require an exponential number of parameters to be overparametrized (e.g., see Fig.~\ref{fig:hea_results}). These results indicate that the search for scalable and trainable ansatzes should be a priority for the field. 

In this sense, our results provide additional guidance to develop QNN architectures with extremely favorable landscapes: overparametrization and absence of barren plateaus. In this context, good candidates are architectures with polynomially large DLAs.

\section{Methods}

In this section, we provide additional details and intuition for the results in the main text, as well as a sketch of the proofs for our main theorems. More detailed proofs of our theorems are given in the Supplementary Information. 

\subsection*{Intuition behind the dynamical Lie algebra}
According to Definition~\ref{def:dynamical_lie_algebra}, the DLA $\liea$ is obtained from the nested commutators of the elements in the set of generators. To understand why this is the case, let us consider a single-layered unitary $U(\thv)$ generated by two Hermitian operators, so that $\GC=\{H_1,H_2\}$.   From the Baker–Campbell–Hausdorff formula, we have that 
\begin{equation}
    U(\thv)=e^{i\theta_1 H_1} e^{i \theta_2 H_2} = e^{K_1(\thv)}\,,
\end{equation}
where 
\begin{align}\label{eq:nested}
\begin{split}
    K_1(\thv)=i(&\theta_1 H_1+\theta_2H_2+\frac{i\theta_1\theta_2}{2}[H_1,H_2]\\
    &-\frac{\theta_1^2 \theta_2}{12}[H_1,[H_1,H_2]]+\ldots)\,.
\end{split}
\end{align}
In Eq.~\eqref{eq:nested} we can see that by combining $e^{i\theta_1H_1}$ and $e^{i\theta_2 H_2}$ into a single term, the new evolution is generated by an operator $K_1(\thv)$ that depends on both $\theta_1$ and $\theta_2$, and which contains the nested commutators between $H_1$ and $H_2$. Here, it is also worth noting that the set formed by the operators $\{iH_1,iH_2,i[H_1,[H_1,H_2]],\ldots\}$ will eventually be closed under the commutation operation in the sense that not all elements will be linearly independent, but rather there will be a finite basis. This is precisely what the DLA is. It is the space spanned by the $\dim(\liea)$ operators that form a basis of the nested commutators. 

When the QNN has multiple layers, that is, when $U(\thv)=\prod_{l=1}^L e^{i\theta_{l1}H_1} e^{i\theta_{l2}H_2}$,  one can recursively apply the Baker–Campbell–Hausdorff formula to express the action of the QNN as being generated by a single parametrized operator $K_L(\thv)$. That is, to have $U(\thv)=e^{K_L(\thv)}$.  Evidently,  both $K_1$ and $K_L$ are obtained from the nested commutators of $H_1$ and $H_2$, and thus both operators are elements of $\liea$. However, while $K_1$ depends on only two parameters $\theta_1$ and $\theta_2$, $K_L$ is parametrized by all $2L$ elements in the vector  $\thv=\{\theta_{l1},\theta_{l2}\}_{l=1}^L$. Having these additional parameters allows for a more fine-tuned control of the action of $U(\thv)$. Intuitively, to hope for a locally surjective map between parameter space and $\liea$, we need to place at least $\dim(\liea)$ parameters. Here, there will come a point where further adding parameters does not further increase one's control of the action of  $U(\thv)$.

We finally note that the analysis for a QNN with more than two unitaries in $\GC$ follows readily. 

\subsection*{Motivation for the definition of overparametrization}
Let us here motivate our definition of overparametrization. First, we recall that we are considering the case where the QNN $U(\thv)$ acts on the states of the training set as $U(\thv)\ket{\psi_\mu}$, and that the loss function is estimated via measurement outcomes on such evolved states. 

In Definition~\ref{def:overparametrization}, we defined overparametrization as a property of the QNN (independently of how the loss function is defined). More specifically, we consider a QNN to be overparametrized if the QNN can explore all relevant directions in the state space. This definition is justified from the fact that, irrespective of how the loss function is estimated via measurements on $U(\thv)\ket{\psi_\mu}$, the accessible space in the Hilbert space is ultimately defined by the action of the QNN in the states of the training set.

Here, one could also potentially define  overparametrization  in terms of exploring all relevant directions in the loss landscape. However, this could have some issues. For instance, consider a QML model where one measures the evolved states $U(\thv)\ket{\psi_\mu}$ in the computational basis and evaluates the loss function as $\LC(\thv)=\sum_{\mu,\vec{z}}p(\vec{z}|\psi_\mu)/|\SC|$, where $p(\vec{z}|\psi_\mu)$ is the probability of measuring the bitstring $\vec{z}$ at the output of the QNN when sending the state $\ket{\psi_\mu}$ as input. Evidently, here $\LC(\thv)=1$ for all $\thv$, and independently of how the QNN is defined. Thus, the loss landscape is always flat, and the Hessian is trivially given by the zero matrix. 

The previous example shows that while a QNN can be considered as overparametrized in the state space, this might not be relevant in the loss landscape space. In view of this issue, we have opted to define overparametrization in the state space, as the map leading to states in the Hilbert space (third map in Eq.~\eqref{eq:maps} and Fig.~\ref{fig:surjectivity}) is more fundamental than the map leading to the loss landscape (fourth map). Evidently, we also expect that arguments can be made in favor of defining overparametrization in terms of the loss landscape, or even the unitary space. However, for the setting presently analyzed,  Definition~\ref{def:overparametrization} can be considered as a first step toward better understanding the overparametrization phenomenon.

\subsection*{Sketch of the Proof of Theorem~\ref{theo:1}}

Let us here consider for simplicity the case when the states in the training set do not respect any symmetries in the QNN (i.e., $\liea_\SC=\liea$). From Eq.~\eqref{eq:PSA_ansatz} we have that 
\begin{equation}\label{eq:der_psi}
     \ket{\partial_j\psi_\mu(\thv)} = \partial_j (U(\thv)\ket{\psi_\mu}) = -i U(\thv) \tilde{H_j} \ket{\psi}
\end{equation}
where we defined $\tilde{H}_j = U_1\ad \cdots U_j\ad H_jU_j \cdots U_1$. Note that here the explicit dependence of $\tilde{H}_j$ in the parameters $\thv$ is omitted. Replacing~\eqref{eq:der_psi} in Eq.~\eqref{eq:QFIM-elem} we find that the elements of the QFIM can be written as
\begin{align}\label{}
[F_\mu(\thv)]_{ij}=4\sum_{m\neq\psi_\mu} \Re[ \bra{\psi_\mu} \tilde{H_i} \ket{m}\bra{m} \tilde{H_j} \ket{\psi_\mu} ]\,.
\end{align}
From here, one finds that the  QFIM can be expressed as
\begin{equation}
    \label{eq:F1-MT}
    F_\mu(\thv) \!=\! -2\! \sum_{m\neq \psi} (\vec{R}_{m\psi_\mu} \cdot\vec{R}_{m\psi_\mu}^{\top} + \vec{I}_{m\psi_\mu} \cdot\vec{I}_{m\psi_\mu}^{\top})\,,
\end{equation}
where we have introduced the vectors $\vec{R}_{mn}$ and $\vec{I}_{mn}$ with components
\begin{equation}\label{eq:elements-RI}
R_{mn}(i) =\Re[ \bra{m} \tilde{H_i} \ket{n}],\,\,\, I_{mn}(i)=\Im[ \bra{m} \tilde{H_i} \ket{n} ]\,.
\end{equation}
Eq.~\eqref{eq:F1-MT} allows us to write the QFIM as a sum of $2d-2$ rank-one matrices.

Here we recall that, by definition, $H_j$ are elements in the DLA $\liea$. Then, since the unitaries $U$ are elements of the dynamical Lie group $\mathbb{G}$ generated by $\liea$, conjugating $H_j$ by any unitary $U$ results in another element in $\liea$. That is: $\forall U\in\mathbb{G}$, and $\forall H_i\in \liea$ we have $U H_j U\ad\in\liea$.  Then, by repeating this argument $j$ times, we find that $\tilde{H}_j\in\liea$.

Letting $\{S_\nu \}_{\nu=1}^{\dim(\liea)}$ be a basis of $\liea$, we can express 
\begin{equation}\label{eq:Htilde-MT}
    \tilde{H_j} = \sum_{\nu=1}^{\rm dim(\liea)} a_{\nu}(j) S_{\nu}\,,
\end{equation}
where $a_{\nu}(j)$ are real coefficients. From Eq.~\eqref{eq:Htilde-MT} we can find 
\begin{align}
    \vec{R}_{mn} &= \sum_{\nu=1}^{\dim(\liea_{\SC})} \Re [\bra{m} S_{\nu} \ket{n}] \vec{a}_{\nu}\,,\label{eq:F2-MT-1}\\
    \vec{I}_{mn} &= \sum_{\nu=1}^{\dim(\liea_{\SC})} \Im [\bra{m} S_{\nu} \ket{n}] \vec{a}_{\nu}\,.\label{eq:F2-MT-2}
\end{align}
Equations~\eqref{eq:F2-MT-1}, and~\eqref{eq:F2-MT-2}  show that the vectors $\vec{R}_{mn}$ and $\vec{I}_{mn}$ can be expressed as a linear combination of  $\dim(\liea_{\SC})$ other vectors $\{\vec{a}_{\nu} \}$. Then, while the $\vec{R}_{mn}$ and $\vec{I}_{mn}$ generate the $2d-2$ rank-one matrices in the QFIM, we have that $F_\mu(\thv)$ has a support on a subspace with a basis that has, at most, $\dim(\liea)$ elements. Thus, we find $\rank[F_\mu(\thv)]\leq \dim(\liea)$. The latter hence proves Theorem~\ref{theo:1}.

Here we note that Eqs.~\eqref{eq:F2-MT-1}, and~\eqref{eq:F2-MT-2} do not take into consideration what the state $\ket{\psi_\mu}$  is. However, from~\eqref{eq:F1-MT}, the QFIM  is actually expressed in terms of $\vec{R_{m\psi_\mu}}$ and $\vec{I_{m\psi_\mu}}$. Then, from the definitions in Eq.~\eqref{eq:elements-RI}, one can see that the state plays a role in the terms   $\tilde{H_i} \ket{\psi_\mu}$. From a closer inspection, one can see that $\tilde{H_i} \ket{\psi_\mu}$ is the action of some elements of the Lie algebra over the state  $\ket{\psi_\mu}$. Thus, since $\tilde{H_i}$ are directions in the Lie group, we have that $\tilde{H_i} \ket{\psi_\mu}$ are directions in the state space.

In fact, the expressible states obtained by acting with the QNN on $\ket{\psi_\mu}$ form a submanifold of the Hilbert space known as the state space orbit, which is defined by      $\lieg\ket{\psi_\mu}=\{U\ket{\psi_\mu}, \forall U\in\mathbb{G}\}$~\cite{dalessandro2010introduction}.  The latter has a very important implications. Since the rank of the QFIM quantifies the number of independent directions in the state space that are accessible via arbitrary infinitesimal variations of the parameter vector $\thv$, it cannot be larger than the dimension of the state space orbit. Thus, one can tighten the bound in Theorem~\ref{theo:1} as 
\begin{equation}
    \rank[F_\mu(\thv)]\leq \dim(\lieg \ket{ \psi_{\mu}})\,,
\end{equation}
where we recall that $\dim(\lieg \ket{\psi_{\mu}})$ is upper bounded by $\dim(\liea)$. 

Thus, one can define the overparametrization as
\begin{definition}[Overparametrization]\label{def:overparametrization-2}
    A QNN is said to be overparametrized if the number of parameters $M$ is such that the QFIM has rank equal to the dimension of the orbits given by the action of $\lieg$ on the states in the training set.
\end{definition}
Evidently, a sufficient condition for the QNN to be overparametrized is that $M\geq\max_\mu  \dim(\lieg \ket{ \psi_{\mu}})$.  For example, we have numerically verified that this occurs in the VQE implementation, where the overparametrization onset occurs when $M=\max_\mu  \dim(\lieg \ket{ \psi_{\mu}})$.

\subsection*{Sketch of the Proof of Theorem~\ref{theo:3}}
The proof of Theorem~\ref{theo:3} follows similarly to that of Theorem~\ref{theo:1}. Specifically, one can show that the Hessian evaluated at $\thv_*$ can be expressed as
 \begin{equation}
    \begin{split}
       \nabla^2\LC(\thv_*)&= \!2\!\sum_{m,n = 1}^d \!\kappa_{mn}(\vec{R}'_{mn} \cdot(\vec{R}'_{mn})^{\top} + \vec{I}'_{mn} \cdot(\vec{I}'_{mn})^{\top}) \,,
    \end{split}
    \end{equation}
where $\kappa_{mn}$ are real coefficients, and where now the vectors $\vec{R}'_{mn}$ and $\vec{I}_{mn}$ have components 
\begin{align}\label{eq:elements-RI-2}
R'_{mn}(i) &=\Re[\bra{m} Q \tilde{H_j} Q\ad \ket{n}]\,,\nonumber\\ I'_{mn}(i)&=\Im[ \bra{m} Q \tilde{H_j} Q\ad \ket{n}]\,.\nonumber
\end{align}
Where $Q$ is the matrix that diagonalizes the operator $\sigma=\sum_\mu c_\mu\dya{\psi_\mu}$.

Then following a similar argument as the one previously used for Theorem~\ref{theo:1}, we can again show that the rank of the Hessian  is such that $\rank[\nabla^2\LC(\thv_*)]\leq\min \{ \dim(\liea_\SC)\}$, we leave for the Supplementary Information the rest of the proof, where the quantity  $r=\min\{\rank[\sigma,\rank[O]\}$ comes into play.

In the Supplemental Information we also provide a proof for the following theorem
\begin{theorem}\label{Theorem4}
Consider the loss functions for a unitary compilation task
\begin{equation}
\begin{split}
    \LC_1 (\thv) = 2d - 2 \Re [T(\thv)],\quad \text{and} \quad \LC_2 (\thv) = 1 - \frac{1}{d^2}|T(\thv)|^2, \nonumber
\end{split}
\end{equation}
where $T(\thv)= \Tr[V\ad U(\vec{\theta})]$ for a target unitary $V$.   Then, let $H_1(\thv_*)$ and $H_2(\thv_*)$  be the Hessian for the loss functions  $\LC_1 (\thv)$ and  $\LC_1 (\thv)$, respectively evaluated at their solutions $U(\thv_*)=V$ and $U(\thv_*)=e^{i\phi}V$. Then, the maximal rank of  $\nabla^2\LC_1(\thv_*)$ and $\nabla^2\LC_1(\thv_*)$  is such that $\rank[\nabla^2\LC_1(\thv_*)],\rank[\nabla^2\LC_2(\thv_*)]\leq \dim(\liea_\SC)$ 
\end{theorem}

\subsection*{Sketch of the Proof of Theorem~\ref{theo:2}}
Let us here consider for simplicity the case when the dataset contains a single state $\ket{\psi}$ which does not respect the symmetries of the ansatz ($\liea_\SC=\liea$). The more general proof is presented in the Supplementary Information.

Now, the capacities of Eq.~\eqref{eq:eff_dim_1} and~\eqref{eq:eff_dim_2} are
\begin{equation}\label{eq:eff_dim_12}
    D_1(\thv)=\sum_{i=1}^M\IC(\lambda^{i}(\thv))\,,
\end{equation}
where $\lambda^{i}(\thv)$ are the eigenvalues of the QFIM for the state $\ket{\psi}$, and 
\begin{equation}\label{eq:eff_dim_22}
    D_2=\max_{\thv}\left(\rank\left[\widetilde{F}(\thv)\right]\right)\,,
\end{equation}
for $\widetilde{F}(\thv)$ the classical Fisher information for the input state $\ket{\psi}$.

First, we note that, by definition, $D_1(\thv)=\rank[F(\thv)]$, so that that the inequality $D_1(\thv)\leq\dim(\liea)$ follows readily from Theorem~\ref{theo:1}. Moreover, by the definition of overparametrization in Definition~\ref{def:overparametrization}, the capacity $D_1(\thv)$ is saturated on at least one point of the landscape. 

Now we need to show that $D_2(\thv)\leq\dim(\liea)$. Here we recall that the quantum and classical Fisher information matrices are such that~\cite{meyer2021fisher}
\begin{equation}
    \widetilde{F}(\thv)\leq F(\thv)
\end{equation}
for all $\thv$. Then, using that fact that if $A$ and $B$ are two Hermitian matrices such that $A\leq B$, then $A^q\leq B^q$ for all $q\in[0,1]$~\cite{chan1985hermitian}. Thus, we have that $\widetilde{F}^q(\thv)\leq F^q(\thv)$, and  choosing $q=0$ leads to
\begin{equation}
    {\rm supp}[\widetilde{F}(\thv)]\leq {\rm supp}[F(\thv)]\,,
\end{equation}
where here $\supp(\cdot)$ denotes the support of a matrix. Taking the trace on both sides allows us to obtain
\begin{equation}
    \rank[\widetilde{F}(\thv)]\leq\rank[F(\thv)]\,.
\end{equation}
Finally, combining Theorem~\ref{theo:1} with the definition of overparametrization in Definition~\ref{def:overparametrization} and the definition of the capacity $D_2(\thv)$ in Eq.~\eqref{eq:eff_dim_22}, if follows that $D_2(\thv)\leq\dim(\liea)$.

\subsection*{Implications for Quantum Optimal Control}

In Quantum Optimal Control (QOC) \cite{glaser2015training,acin2018quantum,yang2017optimizing,lu2017enhancing,rembold2020introduction,peterson2020fast,bluvstein2021controlling,ebadi2021quantum,magann2021feedback,larocca2021krylov,brady2021optimal,wittler2021integrated} one is typically interested in controlling the dynamics of a quantum state $\ket{\psi}$ evolving through a functional time-dependent Hamiltonian
\begin{equation}
    H(t,\{\th_k(t)\}) = H_0 + \sum_{k=1}^K \th_k(t) H_k
\end{equation}
that defines the continuous-in-time equation of motion
\begin{equation}
    \label{eq:U(t)}
    \frac{d U(t)}{ d t} = -i H(t,\{\th_k(t)\})U(t), \; \text{with} \;\quad U(0) = \mathbb{I}\,.
\end{equation}
Here, the idea is that the functions $\{\th_k(t)\}$, known as control fields, can be trained to pursue some desired evolution. 

Interestingly, it has been shown some that QOC and the field of variational quantum algorithms can unified into a single framework where the evolution of a quantum system is controlled at the pulse level (QOC), or at the gate level (QNN) ~\cite{magann2021pulses,larocca2021diagnosing}. Most importantly, irregardless of the choice of controls, the unitaries that are expressible by a QOC ansatz $U(t)$ are, like in the QNN case, contained in the group generated by the DLA $\liea$ (see Definition~\ref{def:dynamical_lie_algebra}) that is determined by the set of generators $\GC=\{H_k\}_{k=0}^K$. Since all of the results presented in this manuscript are stated in terms of the DLA of a given QNN, they can be straightforwardly adapted to the QOC setting. For example, the maximum rank achievable by some QFIM associated with an ansatz of the form in Eq.~\eqref{eq:U(t)} will be upper bounded by the dimension of $\liea$ (equivalent to Theorem \ref{theo:1}). That is, a QOC and a QNN ansatz that share the same set of generators $\GC$ can be expected to have the same saturation value for their respective QFIM matrices.

\begin{figure*}[t!]
    \centering
    \includegraphics[width=0.9\linewidth]{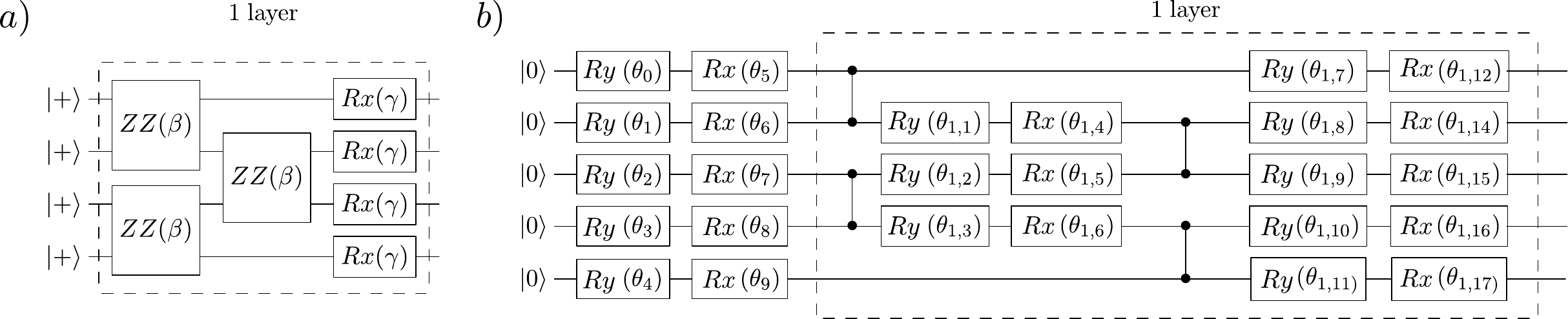}
    \caption{\textbf{QNN ansatzes for the numerical simulations.} a) Hamiltonian variational ansatz for the VQE task. Here we show a single layer of the ansatz for $n=4$ qubits. For closed boundary conditions, there is an extra $ZZ(\beta)$ gate acting on the uppermost and lowermost qubits. A $ZZ(\beta)$ gate on qubits $i,j$ corresponds to the operator $e^{-i \frac{\beta}{2} \sigma^z_i \sigma^z_j}$, and it may be decomposed into two CNOTs and one $Rz$ rotation \cite{smith2019simulating}. The input state was $\ket{+}^{\otimes n}$. b) Hardware efficient ansatz for the unitary compilation and autoencoding tasks. Here we show a single layer of the ansatz for $n=5$ qubits. Notice that there is an extra $Ry$ and $Rx$ rotation on each qubit at the beginning of the circuit.} 
    \label{fig:HVA_circuit}
\end{figure*}

Similarly, in analogy with the results in Theorems \ref{theo:3} and \ref{Theorem4}, the Hessian under a QOC ansatz can be expected to be upper bounded by $\dim(\liea)$ when evaluated at a solution. While the existence of bounds on the rank of the Hessian at solutions is well known in the control literature~\cite{hsieh2009topology,ho2009landscape}, these results analyze the case when the ansatz is controllable (i.e., when $\liea=\mathfrak{s}\mathfrak{u}(d)$) and thus the bounds found are exponentially large. For example, the rank of the Hessian for unitary compilation tasks (see Eq.\eqref{eq:unitary_cost}), has been shown to be upper bounded by $\dim(\liea)=\dim(\mathfrak{s}\mathfrak{u}(d)))=d^2-1$. Hence, the results in this work generalize these previous studies to the case of general $\liea$ (i.e. uncontrollable systems). Let us note that the existence of a fundamental bound on the rank of the Hessian at the global minima is directly connected to another interesting phenomenon: the arisal of continuous submanifolds of degenerate solutions \cite{moore2012exploring,larocca2020exploiting,larocca2020fourier}.

Although historically the quantum control community has mainly focused on controllable systems, the importance of studying uncontrollable ones, in particular those with $\dim(\liea)=\OC(poly(n))$, has been evidenced in~\cite{larocca2021diagnosing}. Here, it has been ascertained that control systems with exponentially large DLAs may encounter scalability issues, like the prescence of barren plateaus in their optimization landscapes. Conversely, systems with polynomially large DLAs can avoid barren plateaus issues and be scalable. Thus, the results in the present manuscript should also be considered as an additional motivation for QOC systems with polynomially sized algebras, as these will achieve overparametrization with $\OC(\poly(n))$ parameters.

Finally, we remark that our results also provide a new insight into the existence of false traps in the control landscape \cite{wu2012singularities,riviello2014searching,rach2015dressing,larocca2018quantum,dalgaard2021predicting}. In QOC, false traps are usually analyzed through the rank of the Jacobian matrix of the map $\{\th_k(t)\} \rightarrow U(t)$. Here, false traps are critical points in the landscape that are not related to local minima of the loss function itself, but to points where this map is not locally not surjective. In our context, this is precisely what a rank-deficient QFIM means: points in parameter space where all possible variations of parameters do not translate into all possible directions in the state space orbit.

\subsection*{Ansatzes for the numerical simulations}

In this section we present the details of the two QNN ansatzes used in our numeric simulations. Let us remark that both ansatzes are of the form in Eq.~\eqref{eqn:hva_cost_function}, i.e. a periodic structured parametrized circuit defined by a given set of generators $\GC$.

Let us first describe the so-called Hamiltonian variational ansatz (HVA) \cite{wecker2015progress,wiersema2020exploring}. Consider a VQE task where one wants to minimize a Hamiltonian of the form
\begin{equation} \label{eqn:hva_Ham}
    H = \sum_{k=1}^N a_k A_k\,,
\end{equation}
where $A_k$ are Hermitian operators and $a_k$ real numbers. The basic idea in the HVA ansatz is to use, as generators, the individual terms in the Hamiltonian that is being minimized, i.e. $\GC=\{ A_k \}_{k=1}^N$. For instance, we show in Fig. \ref{fig:HVA_circuit}(a) the ansatz used to find the ground-state of the  transverse field Ising model with open boundary conditions. Here, the the generators are $A_0=\frac{1}{2}\sum_i \sigma^x_i$ and $A_1= \frac{1}{2}\sum_i \sigma^z_i\sigma^z_{i+1}$.

As a second choice of ansatz, let us introduce the hardware efficient ansatz~\cite{kandala2017hardware} used in the unitary compilation and autoencoding tasks. As shown in Fig.~\ref{fig:HVA_circuit}(b), this ansatz is composed of single qubit rotations followed by CZ gates acting on alternating pairs of qubits. Here we can see that the number of parameters in the ansatz is $M=2n + L(4n-4)$.

\section*{ACKNOWLEDGEMENTS}

We thank Patrick deNiverville, Julia Nakhleh, Stavros Efthymiou, Louis Schatzki and Marco Farinati for useful conversations. NJ and DGM were supported by the U.S. DOE through a quantum computing program sponsored by the Los Alamos National Laboratory (LANL) Information Science \& Technology Institute. DGM acknowledges partial financial support from project QuantumCAT (ref. 001- P-001644), co-funded by the Generalitat de Catalunya and the European Union Regional Development Fund within the ERDF Operational Program of Catalunya, and from the European Union’s Horizon 2020 research and innovation programme under grant agreement No 951911 (AI4Media). PJC and MC were initially supported by Laboratory Directed Research and Development (LDRD) program of LANL under project number 20190065DR. PJC also acknowledges support from the LANL ASC Beyond Moore's Law project. MC also acknowledges support from the Center for Nonlinear Studies at LANL.  This work was supported by the U.S. DOE, Office of Science, Office of Advanced Scientific Computing Research, under the Accelerated Research in Quantum Computing (ARQC) program.

\section*{AUTHOR CONTRIBUTIONS}

The project was conceived by ML, PJC and MC. The manuscript was written by NJ, ML, DGM, PJC, and MC. Theoretical results were proved by  NJ, ML, PJC, and MC. Numerical implementations were performed by DGM. 

\section*{DATA AVAILABILITY}

Data generated and analyzed during current study are available from the corresponding author upon reasonable request.

\section*{COMPETING INTERESTS}
The authors declare no competing interests.

\bibliography{quantum,nathan}

\begin{thebibliography}{135}%
\makeatletter
\providecommand \@ifxundefined [1]{%
 \@ifx{#1\undefined}
}%
\providecommand \@ifnum [1]{%
 \ifnum #1\expandafter \@firstoftwo
 \else \expandafter \@secondoftwo
 \fi
}%
\providecommand \@ifx [1]{%
 \ifx #1\expandafter \@firstoftwo
 \else \expandafter \@secondoftwo
 \fi
}%
\providecommand \natexlab [1]{#1}%
\providecommand \enquote  [1]{``#1''}%
\providecommand \bibnamefont  [1]{#1}%
\providecommand \bibfnamefont [1]{#1}%
\providecommand \citenamefont [1]{#1}%
\providecommand \href@noop [0]{\@secondoftwo}%
\providecommand \href [0]{\begingroup \@sanitize@url \@href}%
\providecommand \@href[1]{\@@startlink{#1}\@@href}%
\providecommand \@@href[1]{\endgroup#1\@@endlink}%
\providecommand \@sanitize@url [0]{\catcode `\\12\catcode `\$12\catcode
  `\&12\catcode `\#12\catcode `\^12\catcode `\_12\catcode `\%12\relax}%
\providecommand \@@startlink[1]{}%
\providecommand \@@endlink[0]{}%
\providecommand \url  [0]{\begingroup\@sanitize@url \@url }%
\providecommand \@url [1]{\endgroup\@href {#1}{\urlprefix }}%
\providecommand \urlprefix  [0]{URL }%
\providecommand \Eprint [0]{\href }%
\providecommand \doibase [0]{https://doi.org/}%
\providecommand \selectlanguage [0]{\@gobble}%
\providecommand \bibinfo  [0]{\@secondoftwo}%
\providecommand \bibfield  [0]{\@secondoftwo}%
\providecommand \translation [1]{[#1]}%
\providecommand \BibitemOpen [0]{}%
\providecommand \bibitemStop [0]{}%
\providecommand \bibitemNoStop [0]{.\EOS\space}%
\providecommand \EOS [0]{\spacefactor3000\relax}%
\providecommand \BibitemShut  [1]{\csname bibitem#1\endcsname}%
\let\auto@bib@innerbib\@empty
\bibitem [{\citenamefont {Mohri}\ \emph {et~al.}(2018)\citenamefont {Mohri},
  \citenamefont {Rostamizadeh},\ and\ \citenamefont
  {Talwalkar}}]{mohri2018foundations}%
  \BibitemOpen
  \bibfield  {author} {\bibinfo {author} {\bibfnamefont {M.}~\bibnamefont
  {Mohri}}, \bibinfo {author} {\bibfnamefont {A.}~\bibnamefont
  {Rostamizadeh}},\ and\ \bibinfo {author} {\bibfnamefont {A.}~\bibnamefont
  {Talwalkar}},\ }\href@noop {} {\emph {\bibinfo {title} {Foundations of
  Machine Learning}}}\ (\bibinfo  {publisher} {MIT Press},\ \bibinfo {year}
  {2018})\BibitemShut {NoStop}%
\bibitem [{\citenamefont {Vamathevan}\ \emph {et~al.}(2019)\citenamefont
  {Vamathevan}, \citenamefont {Clark}, \citenamefont {Czodrowski},
  \citenamefont {Dunham}, \citenamefont {Ferran}, \citenamefont {Lee},
  \citenamefont {Li}, \citenamefont {Madabhushi}, \citenamefont {Shah},
  \citenamefont {Spitzer} \emph {et~al.}}]{vamathevan2019applications}%
  \BibitemOpen
  \bibfield  {author} {\bibinfo {author} {\bibfnamefont {J.}~\bibnamefont
  {Vamathevan}}, \bibinfo {author} {\bibfnamefont {D.}~\bibnamefont {Clark}},
  \bibinfo {author} {\bibfnamefont {P.}~\bibnamefont {Czodrowski}}, \bibinfo
  {author} {\bibfnamefont {I.}~\bibnamefont {Dunham}}, \bibinfo {author}
  {\bibfnamefont {E.}~\bibnamefont {Ferran}}, \bibinfo {author} {\bibfnamefont
  {G.}~\bibnamefont {Lee}}, \bibinfo {author} {\bibfnamefont {B.}~\bibnamefont
  {Li}}, \bibinfo {author} {\bibfnamefont {A.}~\bibnamefont {Madabhushi}},
  \bibinfo {author} {\bibfnamefont {P.}~\bibnamefont {Shah}}, \bibinfo {author}
  {\bibfnamefont {M.}~\bibnamefont {Spitzer}}, \emph {et~al.},\ }\bibfield
  {title} {\bibinfo {title} {Applications of machine learning in drug discovery
  and development},\ }\href {https://doi.org/10.1038/s41573-019-0024-5}
  {\bibfield  {journal} {\bibinfo  {journal} {Nature Reviews Drug Discovery}\
  }\textbf {\bibinfo {volume} {18}},\ \bibinfo {pages} {463} (\bibinfo {year}
  {2019})}\BibitemShut {NoStop}%
\bibitem [{\citenamefont {Schmidt}\ \emph {et~al.}(2019)\citenamefont
  {Schmidt}, \citenamefont {Marques}, \citenamefont {Botti},\ and\
  \citenamefont {Marques}}]{schmidt2019recent}%
  \BibitemOpen
  \bibfield  {author} {\bibinfo {author} {\bibfnamefont {J.}~\bibnamefont
  {Schmidt}}, \bibinfo {author} {\bibfnamefont {M.~R.}\ \bibnamefont
  {Marques}}, \bibinfo {author} {\bibfnamefont {S.}~\bibnamefont {Botti}},\
  and\ \bibinfo {author} {\bibfnamefont {M.~A.}\ \bibnamefont {Marques}},\
  }\bibfield  {title} {\bibinfo {title} {Recent advances and applications of
  machine learning in solid-state materials science},\ }\href
  {https://doi.org/10.1038/s41524-019-0221-0} {\bibfield  {journal} {\bibinfo
  {journal} {npj Computational Materials}\ }\textbf {\bibinfo {volume} {5}},\
  \bibinfo {pages} {1} (\bibinfo {year} {2019})}\BibitemShut {NoStop}%
\bibitem [{\citenamefont {Grigorescu}\ \emph {et~al.}(2020)\citenamefont
  {Grigorescu}, \citenamefont {Trasnea}, \citenamefont {Cocias},\ and\
  \citenamefont {Macesanu}}]{grigorescu2020survey}%
  \BibitemOpen
  \bibfield  {author} {\bibinfo {author} {\bibfnamefont {S.}~\bibnamefont
  {Grigorescu}}, \bibinfo {author} {\bibfnamefont {B.}~\bibnamefont {Trasnea}},
  \bibinfo {author} {\bibfnamefont {T.}~\bibnamefont {Cocias}},\ and\ \bibinfo
  {author} {\bibfnamefont {G.}~\bibnamefont {Macesanu}},\ }\bibfield  {title}
  {\bibinfo {title} {A survey of deep learning techniques for autonomous
  driving},\ }\href {https://doi.org/10.1002/rob.21918} {\bibfield  {journal}
  {\bibinfo  {journal} {Journal of Field Robotics}\ }\textbf {\bibinfo {volume}
  {37}},\ \bibinfo {pages} {362} (\bibinfo {year} {2020})}\BibitemShut
  {NoStop}%
\bibitem [{\citenamefont {Blum}\ and\ \citenamefont
  {Rivest}(1992)}]{blum1992training}%
  \BibitemOpen
  \bibfield  {author} {\bibinfo {author} {\bibfnamefont {A.~L.}\ \bibnamefont
  {Blum}}\ and\ \bibinfo {author} {\bibfnamefont {R.~L.}\ \bibnamefont
  {Rivest}},\ }\bibfield  {title} {\bibinfo {title} {Training a 3-node neural
  network is np-complete},\ }\href
  {https://doi.org/10.1016/S0893-6080(05)80010-3} {\bibfield  {journal}
  {\bibinfo  {journal} {Neural Networks}\ }\textbf {\bibinfo {volume} {5}},\
  \bibinfo {pages} {117} (\bibinfo {year} {1992})}\BibitemShut {NoStop}%
\bibitem [{\citenamefont {Daniely}(2016)}]{daniely2016complexity}%
  \BibitemOpen
  \bibfield  {author} {\bibinfo {author} {\bibfnamefont {A.}~\bibnamefont
  {Daniely}},\ }\bibfield  {title} {\bibinfo {title} {Complexity theoretic
  limitations on learning halfspaces},\ }in\ \href
  {https://doi.org/10.1145/2897518.2897520} {\emph {\bibinfo {booktitle}
  {Proceedings of the forty-eighth annual ACM symposium on Theory of
  Computing}}}\ (\bibinfo {year} {2016})\ pp.\ \bibinfo {pages}
  {105--117}\BibitemShut {NoStop}%
\bibitem [{\citenamefont {Boob}\ \emph {et~al.}(2020)\citenamefont {Boob},
  \citenamefont {Dey},\ and\ \citenamefont {Lan}}]{boob2020complexity}%
  \BibitemOpen
  \bibfield  {author} {\bibinfo {author} {\bibfnamefont {D.}~\bibnamefont
  {Boob}}, \bibinfo {author} {\bibfnamefont {S.~S.}\ \bibnamefont {Dey}},\ and\
  \bibinfo {author} {\bibfnamefont {G.}~\bibnamefont {Lan}},\ }\bibfield
  {title} {\bibinfo {title} {Complexity of training relu neural network},\
  }\href {https://doi.org/10.1016/j.disopt.2020.100620} {\bibfield  {journal}
  {\bibinfo  {journal} {Discrete Optimization}\ ,\ \bibinfo {pages} {100620}}
  (\bibinfo {year} {2020})}\BibitemShut {NoStop}%
\bibitem [{\citenamefont {Neyshabur}\ \emph {et~al.}(2018)\citenamefont
  {Neyshabur}, \citenamefont {Li}, \citenamefont {Bhojanapalli}, \citenamefont
  {LeCun},\ and\ \citenamefont {Srebro}}]{neyshabur2018role}%
  \BibitemOpen
  \bibfield  {author} {\bibinfo {author} {\bibfnamefont {B.}~\bibnamefont
  {Neyshabur}}, \bibinfo {author} {\bibfnamefont {Z.}~\bibnamefont {Li}},
  \bibinfo {author} {\bibfnamefont {S.}~\bibnamefont {Bhojanapalli}}, \bibinfo
  {author} {\bibfnamefont {Y.}~\bibnamefont {LeCun}},\ and\ \bibinfo {author}
  {\bibfnamefont {N.}~\bibnamefont {Srebro}},\ }\bibfield  {title} {\bibinfo
  {title} {The role of over-parametrization in generalization of neural
  networks},\ }in\ \href {https://openreview.net/forum?id=BygfghAcYX} {\emph
  {\bibinfo {booktitle} {International Conference on Learning
  Representations}}}\ (\bibinfo {year} {2018})\BibitemShut {NoStop}%
\bibitem [{\citenamefont {Zhang}\ \emph {et~al.}(2021)\citenamefont {Zhang},
  \citenamefont {Bengio}, \citenamefont {Hardt}, \citenamefont {Recht},\ and\
  \citenamefont {Vinyals}}]{zhang2021understanding}%
  \BibitemOpen
  \bibfield  {author} {\bibinfo {author} {\bibfnamefont {C.}~\bibnamefont
  {Zhang}}, \bibinfo {author} {\bibfnamefont {S.}~\bibnamefont {Bengio}},
  \bibinfo {author} {\bibfnamefont {M.}~\bibnamefont {Hardt}}, \bibinfo
  {author} {\bibfnamefont {B.}~\bibnamefont {Recht}},\ and\ \bibinfo {author}
  {\bibfnamefont {O.}~\bibnamefont {Vinyals}},\ }\bibfield  {title} {\bibinfo
  {title} {Understanding deep learning (still) requires rethinking
  generalization},\ }\href {https://doi.org/10.1145/3446776} {\bibfield
  {journal} {\bibinfo  {journal} {Communications of the ACM}\ }\textbf
  {\bibinfo {volume} {64}},\ \bibinfo {pages} {107} (\bibinfo {year}
  {2021})}\BibitemShut {NoStop}%
\bibitem [{\citenamefont {Allen-Zhu}\ \emph
  {et~al.}(2019{\natexlab{a}})\citenamefont {Allen-Zhu}, \citenamefont {Li},\
  and\ \citenamefont {Song}}]{allen2019convergence}%
  \BibitemOpen
  \bibfield  {author} {\bibinfo {author} {\bibfnamefont {Z.}~\bibnamefont
  {Allen-Zhu}}, \bibinfo {author} {\bibfnamefont {Y.}~\bibnamefont {Li}},\ and\
  \bibinfo {author} {\bibfnamefont {Z.}~\bibnamefont {Song}},\ }\bibfield
  {title} {\bibinfo {title} {A convergence theory for deep learning via
  over-parameterization},\ }in\ \href
  {http://proceedings.mlr.press/v97/allen-zhu19a.html} {\emph {\bibinfo
  {booktitle} {International Conference on Machine Learning}}}\ (\bibinfo
  {organization} {PMLR},\ \bibinfo {year} {2019})\ pp.\ \bibinfo {pages}
  {242--252}\BibitemShut {NoStop}%
\bibitem [{\citenamefont {Allen-Zhu}\ \emph
  {et~al.}(2019{\natexlab{b}})\citenamefont {Allen-Zhu}, \citenamefont {Li},\
  and\ \citenamefont {Liang}}]{allen2019learning}%
  \BibitemOpen
  \bibfield  {author} {\bibinfo {author} {\bibfnamefont {Z.}~\bibnamefont
  {Allen-Zhu}}, \bibinfo {author} {\bibfnamefont {Y.}~\bibnamefont {Li}},\ and\
  \bibinfo {author} {\bibfnamefont {Y.}~\bibnamefont {Liang}},\ }\bibfield
  {title} {\bibinfo {title} {Learning and generalization in overparameterized
  neural networks, going beyond two layers},\ }\href
  {https://proceedings.neurips.cc/paper/2019/file/62dad6e273d32235ae02b7d321578ee8-Paper.pdf}
  {\bibfield  {journal} {\bibinfo  {journal} {Advances in neural information
  processing systems}\ } (\bibinfo {year} {2019}{\natexlab{b}})}\BibitemShut
  {NoStop}%
\bibitem [{\citenamefont {Du}\ \emph {et~al.}(2019{\natexlab{a}})\citenamefont
  {Du}, \citenamefont {Lee}, \citenamefont {Li}, \citenamefont {Wang},\ and\
  \citenamefont {Zhai}}]{du2019gradient}%
  \BibitemOpen
  \bibfield  {author} {\bibinfo {author} {\bibfnamefont {S.}~\bibnamefont
  {Du}}, \bibinfo {author} {\bibfnamefont {J.}~\bibnamefont {Lee}}, \bibinfo
  {author} {\bibfnamefont {H.}~\bibnamefont {Li}}, \bibinfo {author}
  {\bibfnamefont {L.}~\bibnamefont {Wang}},\ and\ \bibinfo {author}
  {\bibfnamefont {X.}~\bibnamefont {Zhai}},\ }\bibfield  {title} {\bibinfo
  {title} {Gradient descent finds global minima of deep neural networks},\ }in\
  \href {http://proceedings.mlr.press/v97/du19c.html} {\emph {\bibinfo
  {booktitle} {International Conference on Machine Learning}}}\ (\bibinfo
  {organization} {PMLR},\ \bibinfo {year} {2019})\ pp.\ \bibinfo {pages}
  {1675--1685}\BibitemShut {NoStop}%
\bibitem [{\citenamefont {Buhai}\ \emph {et~al.}(2020)\citenamefont {Buhai},
  \citenamefont {Halpern}, \citenamefont {Kim}, \citenamefont {Risteski},\ and\
  \citenamefont {Sontag}}]{buhai2020empirical}%
  \BibitemOpen
  \bibfield  {author} {\bibinfo {author} {\bibfnamefont {R.-D.}\ \bibnamefont
  {Buhai}}, \bibinfo {author} {\bibfnamefont {Y.}~\bibnamefont {Halpern}},
  \bibinfo {author} {\bibfnamefont {Y.}~\bibnamefont {Kim}}, \bibinfo {author}
  {\bibfnamefont {A.}~\bibnamefont {Risteski}},\ and\ \bibinfo {author}
  {\bibfnamefont {D.}~\bibnamefont {Sontag}},\ }\bibfield  {title} {\bibinfo
  {title} {Empirical study of the benefits of overparameterization in learning
  latent variable models},\ }in\ \href
  {https://openreview.net/forum?id=rkg0_eHtDr} {\emph {\bibinfo {booktitle}
  {International Conference on Machine Learning}}}\ (\bibinfo {organization}
  {PMLR},\ \bibinfo {year} {2020})\ pp.\ \bibinfo {pages}
  {1211--1219}\BibitemShut {NoStop}%
\bibitem [{\citenamefont {Du}\ \emph {et~al.}(2019{\natexlab{b}})\citenamefont
  {Du}, \citenamefont {Zhai}, \citenamefont {Poczos},\ and\ \citenamefont
  {Singh}}]{du2018gradient}%
  \BibitemOpen
  \bibfield  {author} {\bibinfo {author} {\bibfnamefont {S.~S.}\ \bibnamefont
  {Du}}, \bibinfo {author} {\bibfnamefont {X.}~\bibnamefont {Zhai}}, \bibinfo
  {author} {\bibfnamefont {B.}~\bibnamefont {Poczos}},\ and\ \bibinfo {author}
  {\bibfnamefont {A.}~\bibnamefont {Singh}},\ }\bibfield  {title} {\bibinfo
  {title} {Gradient descent provably optimizes over-parameterized neural
  networks},\ }in\ \href {https://openreview.net/forum?id=S1eK3i09YQ} {\emph
  {\bibinfo {booktitle} {International Conference on Learning
  Representations}}}\ (\bibinfo {year} {2019})\BibitemShut {NoStop}%
\bibitem [{\citenamefont {Brutzkus}\ \emph {et~al.}(2018)\citenamefont
  {Brutzkus}, \citenamefont {Globerson}, \citenamefont {Malach},\ and\
  \citenamefont {Shalev-Shwartz}}]{brutzkus2018sgd}%
  \BibitemOpen
  \bibfield  {author} {\bibinfo {author} {\bibfnamefont {A.}~\bibnamefont
  {Brutzkus}}, \bibinfo {author} {\bibfnamefont {A.}~\bibnamefont {Globerson}},
  \bibinfo {author} {\bibfnamefont {E.}~\bibnamefont {Malach}},\ and\ \bibinfo
  {author} {\bibfnamefont {S.}~\bibnamefont {Shalev-Shwartz}},\ }\bibfield
  {title} {\bibinfo {title} {{SGD} learns over-parameterized networks that
  provably generalize on linearly separable data},\ }in\ \href
  {https://openreview.net/forum?id=rJ33wwxRb} {\emph {\bibinfo {booktitle}
  {International Conference on Learning Representations}}}\ (\bibinfo {year}
  {2018})\BibitemShut {NoStop}%
\bibitem [{\citenamefont {Nielsen}\ and\ \citenamefont
  {Chuang}(2000)}]{nielsen2000quantum}%
  \BibitemOpen
  \bibfield  {author} {\bibinfo {author} {\bibfnamefont {M.~A.}\ \bibnamefont
  {Nielsen}}\ and\ \bibinfo {author} {\bibfnamefont {I.~L.}\ \bibnamefont
  {Chuang}},\ }\href@noop {} {\emph {\bibinfo {title} {Quantum Computation and
  Quantum Information}}}\ (\bibinfo  {publisher} {Cambridge University Press},\
  \bibinfo {year} {2000})\BibitemShut {NoStop}%
\bibitem [{\citenamefont {Preskill}(2018)}]{preskill2018quantum}%
  \BibitemOpen
  \bibfield  {author} {\bibinfo {author} {\bibfnamefont {J.}~\bibnamefont
  {Preskill}},\ }\bibfield  {title} {\bibinfo {title} {Quantum computing in the
  nisq era and beyond},\ }\href
  {https://doi.org/https://doi.org/10.22331/q-2018-08-06-79} {\bibfield
  {journal} {\bibinfo  {journal} {Quantum}\ }\textbf {\bibinfo {volume} {2}},\
  \bibinfo {pages} {79} (\bibinfo {year} {2018})}\BibitemShut {NoStop}%
\bibitem [{\citenamefont {Schuld}\ \emph {et~al.}(2015)\citenamefont {Schuld},
  \citenamefont {Sinayskiy},\ and\ \citenamefont
  {Petruccione}}]{schuld2015introduction}%
  \BibitemOpen
  \bibfield  {author} {\bibinfo {author} {\bibfnamefont {M.}~\bibnamefont
  {Schuld}}, \bibinfo {author} {\bibfnamefont {I.}~\bibnamefont {Sinayskiy}},\
  and\ \bibinfo {author} {\bibfnamefont {F.}~\bibnamefont {Petruccione}},\
  }\bibfield  {title} {\bibinfo {title} {An introduction to quantum machine
  learning},\ }\href {https://doi.org/10.1080/00107514.2014.964942} {\bibfield
  {journal} {\bibinfo  {journal} {Contemporary Physics}\ }\textbf {\bibinfo
  {volume} {56}},\ \bibinfo {pages} {172} (\bibinfo {year} {2015})}\BibitemShut
  {NoStop}%
\bibitem [{\citenamefont {Biamonte}\ \emph {et~al.}(2017)\citenamefont
  {Biamonte}, \citenamefont {Wittek}, \citenamefont {Pancotti}, \citenamefont
  {Rebentrost}, \citenamefont {Wiebe},\ and\ \citenamefont
  {Lloyd}}]{biamonte2017quantum}%
  \BibitemOpen
  \bibfield  {author} {\bibinfo {author} {\bibfnamefont {J.}~\bibnamefont
  {Biamonte}}, \bibinfo {author} {\bibfnamefont {P.}~\bibnamefont {Wittek}},
  \bibinfo {author} {\bibfnamefont {N.}~\bibnamefont {Pancotti}}, \bibinfo
  {author} {\bibfnamefont {P.}~\bibnamefont {Rebentrost}}, \bibinfo {author}
  {\bibfnamefont {N.}~\bibnamefont {Wiebe}},\ and\ \bibinfo {author}
  {\bibfnamefont {S.}~\bibnamefont {Lloyd}},\ }\bibfield  {title} {\bibinfo
  {title} {Quantum machine learning},\ }\href
  {https://www.nature.com/articles/nature23474} {\bibfield  {journal} {\bibinfo
   {journal} {Nature}\ }\textbf {\bibinfo {volume} {549}},\ \bibinfo {pages}
  {195} (\bibinfo {year} {2017})}\BibitemShut {NoStop}%
\bibitem [{\citenamefont {Cerezo}\ \emph
  {et~al.}(2021{\natexlab{a}})\citenamefont {Cerezo}, \citenamefont
  {Arrasmith}, \citenamefont {Babbush}, \citenamefont {Benjamin}, \citenamefont
  {Endo}, \citenamefont {Fujii}, \citenamefont {McClean}, \citenamefont
  {Mitarai}, \citenamefont {Yuan}, \citenamefont {Cincio},\ and\ \citenamefont
  {Coles}}]{cerezo2020variationalreview}%
  \BibitemOpen
  \bibfield  {author} {\bibinfo {author} {\bibfnamefont {M.}~\bibnamefont
  {Cerezo}}, \bibinfo {author} {\bibfnamefont {A.}~\bibnamefont {Arrasmith}},
  \bibinfo {author} {\bibfnamefont {R.}~\bibnamefont {Babbush}}, \bibinfo
  {author} {\bibfnamefont {S.~C.}\ \bibnamefont {Benjamin}}, \bibinfo {author}
  {\bibfnamefont {S.}~\bibnamefont {Endo}}, \bibinfo {author} {\bibfnamefont
  {K.}~\bibnamefont {Fujii}}, \bibinfo {author} {\bibfnamefont {J.~R.}\
  \bibnamefont {McClean}}, \bibinfo {author} {\bibfnamefont {K.}~\bibnamefont
  {Mitarai}}, \bibinfo {author} {\bibfnamefont {X.}~\bibnamefont {Yuan}},
  \bibinfo {author} {\bibfnamefont {L.}~\bibnamefont {Cincio}},\ and\ \bibinfo
  {author} {\bibfnamefont {P.~J.}\ \bibnamefont {Coles}},\ }\bibfield  {title}
  {\bibinfo {title} {Variational quantum algorithms},\ }\href
  {https://doi.org/10.1038/s42254-021-00348-9} {\bibfield  {journal} {\bibinfo
  {journal} {Nature Reviews Physics}\ }\textbf {\bibinfo {volume} {1}},\
  \bibinfo {pages} {19} (\bibinfo {year} {2021}{\natexlab{a}})}\BibitemShut
  {NoStop}%
\bibitem [{\citenamefont {Huang}\ \emph
  {et~al.}(2021{\natexlab{a}})\citenamefont {Huang}, \citenamefont {Kueng},\
  and\ \citenamefont {Preskill}}]{huang2021information}%
  \BibitemOpen
  \bibfield  {author} {\bibinfo {author} {\bibfnamefont {H.-Y.}\ \bibnamefont
  {Huang}}, \bibinfo {author} {\bibfnamefont {R.}~\bibnamefont {Kueng}},\ and\
  \bibinfo {author} {\bibfnamefont {J.}~\bibnamefont {Preskill}},\ }\bibfield
  {title} {\bibinfo {title} {Information-theoretic bounds on quantum advantage
  in machine learning},\ }\href
  {https://doi.org/10.1103/PhysRevLett.126.190505} {\bibfield  {journal}
  {\bibinfo  {journal} {Phys. Rev. Lett.}\ }\textbf {\bibinfo {volume} {126}},\
  \bibinfo {pages} {190505} (\bibinfo {year} {2021}{\natexlab{a}})}\BibitemShut
  {NoStop}%
\bibitem [{\citenamefont {Huang}\ \emph
  {et~al.}(2021{\natexlab{b}})\citenamefont {Huang}, \citenamefont {Broughton},
  \citenamefont {Mohseni}, \citenamefont {Babbush}, \citenamefont {Boixo},
  \citenamefont {Neven},\ and\ \citenamefont {McClean}}]{huang2021power}%
  \BibitemOpen
  \bibfield  {author} {\bibinfo {author} {\bibfnamefont {H.-Y.}\ \bibnamefont
  {Huang}}, \bibinfo {author} {\bibfnamefont {M.}~\bibnamefont {Broughton}},
  \bibinfo {author} {\bibfnamefont {M.}~\bibnamefont {Mohseni}}, \bibinfo
  {author} {\bibfnamefont {R.}~\bibnamefont {Babbush}}, \bibinfo {author}
  {\bibfnamefont {S.}~\bibnamefont {Boixo}}, \bibinfo {author} {\bibfnamefont
  {H.}~\bibnamefont {Neven}},\ and\ \bibinfo {author} {\bibfnamefont {J.~R.}\
  \bibnamefont {McClean}},\ }\bibfield  {title} {\bibinfo {title} {Power of
  data in quantum machine learning},\ }\href
  {https://doi.org/10.1038/s41467-021-22539-9} {\bibfield  {journal} {\bibinfo
  {journal} {Nature Communications}\ }\textbf {\bibinfo {volume} {12}},\
  \bibinfo {pages} {1} (\bibinfo {year} {2021}{\natexlab{b}})}\BibitemShut
  {NoStop}%
\bibitem [{\citenamefont {K{\"u}bler}\ \emph {et~al.}(2021)\citenamefont
  {K{\"u}bler}, \citenamefont {Buchholz},\ and\ \citenamefont
  {Sch{\"o}lkopf}}]{kubler2021inductive}%
  \BibitemOpen
  \bibfield  {author} {\bibinfo {author} {\bibfnamefont {J.~M.}\ \bibnamefont
  {K{\"u}bler}}, \bibinfo {author} {\bibfnamefont {S.}~\bibnamefont
  {Buchholz}},\ and\ \bibinfo {author} {\bibfnamefont {B.}~\bibnamefont
  {Sch{\"o}lkopf}},\ }\bibfield  {title} {\bibinfo {title} {The inductive bias
  of quantum kernels},\ }\href {https://arxiv.org/abs/2106.03747} {\bibfield
  {journal} {\bibinfo  {journal} {arXiv preprint arXiv:2106.03747}\ } (\bibinfo
  {year} {2021})}\BibitemShut {NoStop}%
\bibitem [{\citenamefont {Abbas}\ \emph {et~al.}(2021)\citenamefont {Abbas},
  \citenamefont {Sutter}, \citenamefont {Zoufal}, \citenamefont {Lucchi},
  \citenamefont {Figalli},\ and\ \citenamefont {Woerner}}]{abbas2020power}%
  \BibitemOpen
  \bibfield  {author} {\bibinfo {author} {\bibfnamefont {A.}~\bibnamefont
  {Abbas}}, \bibinfo {author} {\bibfnamefont {D.}~\bibnamefont {Sutter}},
  \bibinfo {author} {\bibfnamefont {C.}~\bibnamefont {Zoufal}}, \bibinfo
  {author} {\bibfnamefont {A.}~\bibnamefont {Lucchi}}, \bibinfo {author}
  {\bibfnamefont {A.}~\bibnamefont {Figalli}},\ and\ \bibinfo {author}
  {\bibfnamefont {S.}~\bibnamefont {Woerner}},\ }\bibfield  {title} {\bibinfo
  {title} {The power of quantum neural networks},\ }\href
  {https://doi.org/10.1038/s43588-021-00084-1} {\bibfield  {journal} {\bibinfo
  {journal} {Nature Computational Science}\ }\textbf {\bibinfo {volume} {1}},\
  \bibinfo {pages} {403} (\bibinfo {year} {2021})}\BibitemShut {NoStop}%
\bibitem [{\citenamefont {Bittel}\ and\ \citenamefont
  {Kliesch}(2021)}]{bittel2021training}%
  \BibitemOpen
  \bibfield  {author} {\bibinfo {author} {\bibfnamefont {L.}~\bibnamefont
  {Bittel}}\ and\ \bibinfo {author} {\bibfnamefont {M.}~\bibnamefont
  {Kliesch}},\ }\bibfield  {title} {\bibinfo {title} {Training variational
  quantum algorithms is np-hard},\ }\href
  {https://doi.org/10.1103/PhysRevLett.127.120502} {\bibfield  {journal}
  {\bibinfo  {journal} {Phys. Rev. Lett.}\ }\textbf {\bibinfo {volume} {127}},\
  \bibinfo {pages} {120502} (\bibinfo {year} {2021})}\BibitemShut {NoStop}%
\bibitem [{\citenamefont {Benedetti}\ \emph {et~al.}(2019)\citenamefont
  {Benedetti}, \citenamefont {Lloyd}, \citenamefont {Sack},\ and\ \citenamefont
  {Fiorentini}}]{benedetti2019parameterized}%
  \BibitemOpen
  \bibfield  {author} {\bibinfo {author} {\bibfnamefont {M.}~\bibnamefont
  {Benedetti}}, \bibinfo {author} {\bibfnamefont {E.}~\bibnamefont {Lloyd}},
  \bibinfo {author} {\bibfnamefont {S.}~\bibnamefont {Sack}},\ and\ \bibinfo
  {author} {\bibfnamefont {M.}~\bibnamefont {Fiorentini}},\ }\bibfield  {title}
  {\bibinfo {title} {Parameterized quantum circuits as machine learning
  models},\ }\href {https://doi.org/10.1088/2058-9565/ab5944} {\bibfield
  {journal} {\bibinfo  {journal} {Quantum Science and Technology}\ }\textbf
  {\bibinfo {volume} {4}},\ \bibinfo {pages} {043001} (\bibinfo {year}
  {2019})}\BibitemShut {NoStop}%
\bibitem [{\citenamefont {Beer}\ \emph {et~al.}(2020)\citenamefont {Beer},
  \citenamefont {Bondarenko}, \citenamefont {Farrelly}, \citenamefont
  {Osborne}, \citenamefont {Salzmann}, \citenamefont {Scheiermann},\ and\
  \citenamefont {Wolf}}]{beer2020training}%
  \BibitemOpen
  \bibfield  {author} {\bibinfo {author} {\bibfnamefont {K.}~\bibnamefont
  {Beer}}, \bibinfo {author} {\bibfnamefont {D.}~\bibnamefont {Bondarenko}},
  \bibinfo {author} {\bibfnamefont {T.}~\bibnamefont {Farrelly}}, \bibinfo
  {author} {\bibfnamefont {T.~J.}\ \bibnamefont {Osborne}}, \bibinfo {author}
  {\bibfnamefont {R.}~\bibnamefont {Salzmann}}, \bibinfo {author}
  {\bibfnamefont {D.}~\bibnamefont {Scheiermann}},\ and\ \bibinfo {author}
  {\bibfnamefont {R.}~\bibnamefont {Wolf}},\ }\bibfield  {title} {\bibinfo
  {title} {Training deep quantum neural networks},\ }\href
  {https://doi.org/10.1038/s41467-020-14454-2} {\bibfield  {journal} {\bibinfo
  {journal} {Nature Communications}\ }\textbf {\bibinfo {volume} {11}},\
  \bibinfo {pages} {808} (\bibinfo {year} {2020})}\BibitemShut {NoStop}%
\bibitem [{\citenamefont {Cong}\ \emph {et~al.}(2019)\citenamefont {Cong},
  \citenamefont {Choi},\ and\ \citenamefont {Lukin}}]{cong2019quantum}%
  \BibitemOpen
  \bibfield  {author} {\bibinfo {author} {\bibfnamefont {I.}~\bibnamefont
  {Cong}}, \bibinfo {author} {\bibfnamefont {S.}~\bibnamefont {Choi}},\ and\
  \bibinfo {author} {\bibfnamefont {M.~D.}\ \bibnamefont {Lukin}},\ }\bibfield
  {title} {\bibinfo {title} {Quantum convolutional neural networks},\ }\href
  {https://doi.org/10.1038/s41567-019-0648-8} {\bibfield  {journal} {\bibinfo
  {journal} {Nature Physics}\ }\textbf {\bibinfo {volume} {15}},\ \bibinfo
  {pages} {1273} (\bibinfo {year} {2019})}\BibitemShut {NoStop}%
\bibitem [{\citenamefont {Farhi}\ and\ \citenamefont
  {Neven}(2018)}]{farhi2018classification}%
  \BibitemOpen
  \bibfield  {author} {\bibinfo {author} {\bibfnamefont {E.}~\bibnamefont
  {Farhi}}\ and\ \bibinfo {author} {\bibfnamefont {H.}~\bibnamefont {Neven}},\
  }\bibfield  {title} {\bibinfo {title} {Classification with quantum neural
  networks on near term processors},\ }\href {https://arxiv.org/abs/1802.06002}
  {\bibfield  {journal} {\bibinfo  {journal} {arXiv preprint arXiv:1802.06002}\
  } (\bibinfo {year} {2018})}\BibitemShut {NoStop}%
\bibitem [{\citenamefont {Arrasmith}\ \emph {et~al.}(2021)\citenamefont
  {Arrasmith}, \citenamefont {Holmes}, \citenamefont {Cerezo},\ and\
  \citenamefont {Coles}}]{arrasmith2021equivalence}%
  \BibitemOpen
  \bibfield  {author} {\bibinfo {author} {\bibfnamefont {A.}~\bibnamefont
  {Arrasmith}}, \bibinfo {author} {\bibfnamefont {Z.}~\bibnamefont {Holmes}},
  \bibinfo {author} {\bibfnamefont {M.}~\bibnamefont {Cerezo}},\ and\ \bibinfo
  {author} {\bibfnamefont {P.~J.}\ \bibnamefont {Coles}},\ }\bibfield  {title}
  {\bibinfo {title} {Equivalence of quantum barren plateaus to cost
  concentration and narrow gorges},\ }\href {https://arxiv.org/abs/2104.05868}
  {\bibfield  {journal} {\bibinfo  {journal} {arXiv preprint arXiv:2104.05868}\
  } (\bibinfo {year} {2021})}\BibitemShut {NoStop}%
\bibitem [{\citenamefont {Rivera-Dean}\ \emph {et~al.}(2021)\citenamefont
  {Rivera-Dean}, \citenamefont {Huembeli}, \citenamefont {Ac{\'\i}n},\ and\
  \citenamefont {Bowles}}]{rivera2021avoiding}%
  \BibitemOpen
  \bibfield  {author} {\bibinfo {author} {\bibfnamefont {J.}~\bibnamefont
  {Rivera-Dean}}, \bibinfo {author} {\bibfnamefont {P.}~\bibnamefont
  {Huembeli}}, \bibinfo {author} {\bibfnamefont {A.}~\bibnamefont
  {Ac{\'\i}n}},\ and\ \bibinfo {author} {\bibfnamefont {J.}~\bibnamefont
  {Bowles}},\ }\bibfield  {title} {\bibinfo {title} {Avoiding local minima in
  variational quantum algorithms with neural networks},\ }\href
  {https://arxiv.org/abs/2104.02955} {\bibfield  {journal} {\bibinfo  {journal}
  {arXiv preprint arXiv:2104.02955}\ } (\bibinfo {year} {2021})}\BibitemShut
  {NoStop}%
\bibitem [{\citenamefont {Wierichs}\ \emph {et~al.}(2020)\citenamefont
  {Wierichs}, \citenamefont {Gogolin},\ and\ \citenamefont
  {Kastoryano}}]{wierichs2020avoiding}%
  \BibitemOpen
  \bibfield  {author} {\bibinfo {author} {\bibfnamefont {D.}~\bibnamefont
  {Wierichs}}, \bibinfo {author} {\bibfnamefont {C.}~\bibnamefont {Gogolin}},\
  and\ \bibinfo {author} {\bibfnamefont {M.}~\bibnamefont {Kastoryano}},\
  }\bibfield  {title} {\bibinfo {title} {Avoiding local minima in variational
  quantum eigensolvers with the natural gradient optimizer},\ }\href
  {https://doi.org/10.1103/PhysRevA.104.032401} {\bibfield  {journal} {\bibinfo
   {journal} {Physical Review Research}\ }\textbf {\bibinfo {volume} {2}},\
  \bibinfo {pages} {043246} (\bibinfo {year} {2020})}\BibitemShut {NoStop}%
\bibitem [{\citenamefont {McClean}\ \emph {et~al.}(2018)\citenamefont
  {McClean}, \citenamefont {Boixo}, \citenamefont {Smelyanskiy}, \citenamefont
  {Babbush},\ and\ \citenamefont {Neven}}]{mcclean2018barren}%
  \BibitemOpen
  \bibfield  {author} {\bibinfo {author} {\bibfnamefont {J.~R.}\ \bibnamefont
  {McClean}}, \bibinfo {author} {\bibfnamefont {S.}~\bibnamefont {Boixo}},
  \bibinfo {author} {\bibfnamefont {V.~N.}\ \bibnamefont {Smelyanskiy}},
  \bibinfo {author} {\bibfnamefont {R.}~\bibnamefont {Babbush}},\ and\ \bibinfo
  {author} {\bibfnamefont {H.}~\bibnamefont {Neven}},\ }\bibfield  {title}
  {\bibinfo {title} {Barren plateaus in quantum neural network training
  landscapes},\ }\href {https://doi.org/10.1038/s41467-018-07090-4} {\bibfield
  {journal} {\bibinfo  {journal} {Nature communications}\ }\textbf {\bibinfo
  {volume} {9}},\ \bibinfo {pages} {1} (\bibinfo {year} {2018})}\BibitemShut
  {NoStop}%
\bibitem [{\citenamefont {Cerezo}\ \emph
  {et~al.}(2021{\natexlab{b}})\citenamefont {Cerezo}, \citenamefont {Sone},
  \citenamefont {Volkoff}, \citenamefont {Cincio},\ and\ \citenamefont
  {Coles}}]{cerezo2020cost}%
  \BibitemOpen
  \bibfield  {author} {\bibinfo {author} {\bibfnamefont {M.}~\bibnamefont
  {Cerezo}}, \bibinfo {author} {\bibfnamefont {A.}~\bibnamefont {Sone}},
  \bibinfo {author} {\bibfnamefont {T.}~\bibnamefont {Volkoff}}, \bibinfo
  {author} {\bibfnamefont {L.}~\bibnamefont {Cincio}},\ and\ \bibinfo {author}
  {\bibfnamefont {P.~J.}\ \bibnamefont {Coles}},\ }\bibfield  {title} {\bibinfo
  {title} {Cost function dependent barren plateaus in shallow parametrized
  quantum circuits},\ }\href {https://doi.org/10.1038/s41467-021-21728-w}
  {\bibfield  {journal} {\bibinfo  {journal} {Nature Communications}\ }\textbf
  {\bibinfo {volume} {12}},\ \bibinfo {pages} {1791} (\bibinfo {year}
  {2021}{\natexlab{b}})}\BibitemShut {NoStop}%
\bibitem [{\citenamefont {Marrero}\ \emph {et~al.}(2020)\citenamefont
  {Marrero}, \citenamefont {Kieferová},\ and\ \citenamefont
  {Wiebe}}]{marrero2020entanglement}%
  \BibitemOpen
  \bibfield  {author} {\bibinfo {author} {\bibfnamefont {C.~O.}\ \bibnamefont
  {Marrero}}, \bibinfo {author} {\bibfnamefont {M.}~\bibnamefont
  {Kieferová}},\ and\ \bibinfo {author} {\bibfnamefont {N.}~\bibnamefont
  {Wiebe}},\ }\bibfield  {title} {\bibinfo {title} {Entanglement induced barren
  plateaus},\ }\href {https://arxiv.org/abs/2010.15968} {\bibfield  {journal}
  {\bibinfo  {journal} {arXiv preprint arXiv:2010.15968}\ } (\bibinfo {year}
  {2020})}\BibitemShut {NoStop}%
\bibitem [{\citenamefont {Patti}\ \emph {et~al.}(2021)\citenamefont {Patti},
  \citenamefont {Najafi}, \citenamefont {Gao},\ and\ \citenamefont
  {Yelin}}]{patti2020entanglement}%
  \BibitemOpen
  \bibfield  {author} {\bibinfo {author} {\bibfnamefont {T.~L.}\ \bibnamefont
  {Patti}}, \bibinfo {author} {\bibfnamefont {K.}~\bibnamefont {Najafi}},
  \bibinfo {author} {\bibfnamefont {X.}~\bibnamefont {Gao}},\ and\ \bibinfo
  {author} {\bibfnamefont {S.~F.}\ \bibnamefont {Yelin}},\ }\bibfield  {title}
  {\bibinfo {title} {Entanglement devised barren plateau mitigation},\ }\href
  {https://doi.org/https://doi.org/10.1103/PhysRevResearch.3.033090} {\bibfield
   {journal} {\bibinfo  {journal} {Physical Review Research}\ }\textbf
  {\bibinfo {volume} {3}},\ \bibinfo {pages} {033090} (\bibinfo {year}
  {2021})}\BibitemShut {NoStop}%
\bibitem [{\citenamefont {Larocca}\ \emph {et~al.}(2021)\citenamefont
  {Larocca}, \citenamefont {Czarnik}, \citenamefont {Sharma}, \citenamefont
  {Muraleedharan}, \citenamefont {Coles},\ and\ \citenamefont
  {Cerezo}}]{larocca2021diagnosing}%
  \BibitemOpen
  \bibfield  {author} {\bibinfo {author} {\bibfnamefont {M.}~\bibnamefont
  {Larocca}}, \bibinfo {author} {\bibfnamefont {P.}~\bibnamefont {Czarnik}},
  \bibinfo {author} {\bibfnamefont {K.}~\bibnamefont {Sharma}}, \bibinfo
  {author} {\bibfnamefont {G.}~\bibnamefont {Muraleedharan}}, \bibinfo {author}
  {\bibfnamefont {P.~J.}\ \bibnamefont {Coles}},\ and\ \bibinfo {author}
  {\bibfnamefont {M.}~\bibnamefont {Cerezo}},\ }\bibfield  {title} {\bibinfo
  {title} {Diagnosing barren plateaus with tools from quantum optimal
  control},\ }\href {https://arxiv.org/abs/2105.14377} {\bibfield  {journal}
  {\bibinfo  {journal} {arXiv preprint arXiv:2105.14377}\ } (\bibinfo {year}
  {2021})}\BibitemShut {NoStop}%
\bibitem [{\citenamefont {Holmes}\ \emph
  {et~al.}(2021{\natexlab{a}})\citenamefont {Holmes}, \citenamefont {Sharma},
  \citenamefont {Cerezo},\ and\ \citenamefont {Coles}}]{holmes2021connecting}%
  \BibitemOpen
  \bibfield  {author} {\bibinfo {author} {\bibfnamefont {Z.}~\bibnamefont
  {Holmes}}, \bibinfo {author} {\bibfnamefont {K.}~\bibnamefont {Sharma}},
  \bibinfo {author} {\bibfnamefont {M.}~\bibnamefont {Cerezo}},\ and\ \bibinfo
  {author} {\bibfnamefont {P.~J.}\ \bibnamefont {Coles}},\ }\bibfield  {title}
  {\bibinfo {title} {Connecting ansatz expressibility to gradient magnitudes
  and barren plateaus},\ }\href {https://arxiv.org/abs/2101.02138} {\bibfield
  {journal} {\bibinfo  {journal} {arXiv preprint arXiv:2101.02138}\ } (\bibinfo
  {year} {2021}{\natexlab{a}})}\BibitemShut {NoStop}%
\bibitem [{\citenamefont {Cerezo}\ and\ \citenamefont
  {Coles}(2021)}]{cerezo2020impact}%
  \BibitemOpen
  \bibfield  {author} {\bibinfo {author} {\bibfnamefont {M.}~\bibnamefont
  {Cerezo}}\ and\ \bibinfo {author} {\bibfnamefont {P.~J.}\ \bibnamefont
  {Coles}},\ }\bibfield  {title} {\bibinfo {title} {Higher order derivatives of
  quantum neural networks with barren plateaus},\ }\href
  {https://doi.org/10.1088/2058-9565/abf51a} {\bibfield  {journal} {\bibinfo
  {journal} {Quantum Science and Technology}\ }\textbf {\bibinfo {volume}
  {6}},\ \bibinfo {pages} {035006} (\bibinfo {year} {2021})}\BibitemShut
  {NoStop}%
\bibitem [{\citenamefont {Holmes}\ \emph
  {et~al.}(2021{\natexlab{b}})\citenamefont {Holmes}, \citenamefont
  {Arrasmith}, \citenamefont {Yan}, \citenamefont {Coles}, \citenamefont
  {Albrecht},\ and\ \citenamefont {Sornborger}}]{holmes2021barren}%
  \BibitemOpen
  \bibfield  {author} {\bibinfo {author} {\bibfnamefont {Z.}~\bibnamefont
  {Holmes}}, \bibinfo {author} {\bibfnamefont {A.}~\bibnamefont {Arrasmith}},
  \bibinfo {author} {\bibfnamefont {B.}~\bibnamefont {Yan}}, \bibinfo {author}
  {\bibfnamefont {P.~J.}\ \bibnamefont {Coles}}, \bibinfo {author}
  {\bibfnamefont {A.}~\bibnamefont {Albrecht}},\ and\ \bibinfo {author}
  {\bibfnamefont {A.~T.}\ \bibnamefont {Sornborger}},\ }\bibfield  {title}
  {\bibinfo {title} {Barren plateaus preclude learning scramblers},\ }\href
  {https://doi.org/10.1103/PhysRevLett.126.190501} {\bibfield  {journal}
  {\bibinfo  {journal} {Physical Review Letters}\ }\textbf {\bibinfo {volume}
  {126}},\ \bibinfo {pages} {190501} (\bibinfo {year}
  {2021}{\natexlab{b}})}\BibitemShut {NoStop}%
\bibitem [{\citenamefont {Arrasmith}\ \emph {et~al.}(2020)\citenamefont
  {Arrasmith}, \citenamefont {Cerezo}, \citenamefont {Czarnik}, \citenamefont
  {Cincio},\ and\ \citenamefont {Coles}}]{arrasmith2020effect}%
  \BibitemOpen
  \bibfield  {author} {\bibinfo {author} {\bibfnamefont {A.}~\bibnamefont
  {Arrasmith}}, \bibinfo {author} {\bibfnamefont {M.}~\bibnamefont {Cerezo}},
  \bibinfo {author} {\bibfnamefont {P.}~\bibnamefont {Czarnik}}, \bibinfo
  {author} {\bibfnamefont {L.}~\bibnamefont {Cincio}},\ and\ \bibinfo {author}
  {\bibfnamefont {P.~J.}\ \bibnamefont {Coles}},\ }\bibfield  {title} {\bibinfo
  {title} {Effect of barren plateaus on gradient-free optimization},\ }\href
  {https://arxiv.org/abs/2011.12245} {\bibfield  {journal} {\bibinfo  {journal}
  {arXiv preprint arXiv:2011.12245}\ } (\bibinfo {year} {2020})}\BibitemShut
  {NoStop}%
\bibitem [{\citenamefont {Huembeli}\ and\ \citenamefont
  {Dauphin}(2021)}]{huembeli2021characterizing}%
  \BibitemOpen
  \bibfield  {author} {\bibinfo {author} {\bibfnamefont {P.}~\bibnamefont
  {Huembeli}}\ and\ \bibinfo {author} {\bibfnamefont {A.}~\bibnamefont
  {Dauphin}},\ }\bibfield  {title} {\bibinfo {title} {Characterizing the loss
  landscape of variational quantum circuits},\ }\href
  {https://doi.org/10.1088/2058-9565/abdbc9} {\bibfield  {journal} {\bibinfo
  {journal} {Quantum Science and Technology}\ }\textbf {\bibinfo {volume}
  {6}},\ \bibinfo {pages} {025011} (\bibinfo {year} {2021})}\BibitemShut
  {NoStop}%
\bibitem [{\citenamefont {Pesah}\ \emph {et~al.}(2020)\citenamefont {Pesah},
  \citenamefont {Cerezo}, \citenamefont {Wang}, \citenamefont {Volkoff},
  \citenamefont {Sornborger},\ and\ \citenamefont {Coles}}]{pesah2020absence}%
  \BibitemOpen
  \bibfield  {author} {\bibinfo {author} {\bibfnamefont {A.}~\bibnamefont
  {Pesah}}, \bibinfo {author} {\bibfnamefont {M.}~\bibnamefont {Cerezo}},
  \bibinfo {author} {\bibfnamefont {S.}~\bibnamefont {Wang}}, \bibinfo {author}
  {\bibfnamefont {T.}~\bibnamefont {Volkoff}}, \bibinfo {author} {\bibfnamefont
  {A.~T.}\ \bibnamefont {Sornborger}},\ and\ \bibinfo {author} {\bibfnamefont
  {P.~J.}\ \bibnamefont {Coles}},\ }\bibfield  {title} {\bibinfo {title}
  {Absence of barren plateaus in quantum convolutional neural networks},\
  }\href {https://arxiv.org/abs/2011.02966} {\bibfield  {journal} {\bibinfo
  {journal} {arXiv preprint arXiv:2011.02966}\ } (\bibinfo {year}
  {2020})}\BibitemShut {NoStop}%
\bibitem [{\citenamefont {Sharma}\ \emph
  {et~al.}(2020{\natexlab{a}})\citenamefont {Sharma}, \citenamefont {Cerezo},
  \citenamefont {Cincio},\ and\ \citenamefont
  {Coles}}]{sharma2020trainability}%
  \BibitemOpen
  \bibfield  {author} {\bibinfo {author} {\bibfnamefont {K.}~\bibnamefont
  {Sharma}}, \bibinfo {author} {\bibfnamefont {M.}~\bibnamefont {Cerezo}},
  \bibinfo {author} {\bibfnamefont {L.}~\bibnamefont {Cincio}},\ and\ \bibinfo
  {author} {\bibfnamefont {P.~J.}\ \bibnamefont {Coles}},\ }\bibfield  {title}
  {\bibinfo {title} {Trainability of dissipative perceptron-based quantum
  neural networks},\ }\href {https://arxiv.org/abs/2005.12458} {\bibfield
  {journal} {\bibinfo  {journal} {arXiv preprint arXiv:2005.12458}\ } (\bibinfo
  {year} {2020}{\natexlab{a}})}\BibitemShut {NoStop}%
\bibitem [{\citenamefont {Wang}\ \emph {et~al.}(2020)\citenamefont {Wang},
  \citenamefont {Fontana}, \citenamefont {Cerezo}, \citenamefont {Sharma},
  \citenamefont {Sone}, \citenamefont {Cincio},\ and\ \citenamefont
  {Coles}}]{wang2020noise}%
  \BibitemOpen
  \bibfield  {author} {\bibinfo {author} {\bibfnamefont {S.}~\bibnamefont
  {Wang}}, \bibinfo {author} {\bibfnamefont {E.}~\bibnamefont {Fontana}},
  \bibinfo {author} {\bibfnamefont {M.}~\bibnamefont {Cerezo}}, \bibinfo
  {author} {\bibfnamefont {K.}~\bibnamefont {Sharma}}, \bibinfo {author}
  {\bibfnamefont {A.}~\bibnamefont {Sone}}, \bibinfo {author} {\bibfnamefont
  {L.}~\bibnamefont {Cincio}},\ and\ \bibinfo {author} {\bibfnamefont {P.~J.}\
  \bibnamefont {Coles}},\ }\bibfield  {title} {\bibinfo {title} {Noise-induced
  barren plateaus in variational quantum algorithms},\ }\href
  {https://arxiv.org/abs/2007.14384} {\bibfield  {journal} {\bibinfo  {journal}
  {arXiv preprint arXiv:2007.14384}\ } (\bibinfo {year} {2020})}\BibitemShut
  {NoStop}%
\bibitem [{\citenamefont {Fontana}\ \emph {et~al.}(2020)\citenamefont
  {Fontana}, \citenamefont {Cerezo}, \citenamefont {Arrasmith}, \citenamefont
  {Rungger},\ and\ \citenamefont {Coles}}]{fontana2020optimizing}%
  \BibitemOpen
  \bibfield  {author} {\bibinfo {author} {\bibfnamefont {E.}~\bibnamefont
  {Fontana}}, \bibinfo {author} {\bibfnamefont {M.}~\bibnamefont {Cerezo}},
  \bibinfo {author} {\bibfnamefont {A.}~\bibnamefont {Arrasmith}}, \bibinfo
  {author} {\bibfnamefont {I.}~\bibnamefont {Rungger}},\ and\ \bibinfo {author}
  {\bibfnamefont {P.~J.}\ \bibnamefont {Coles}},\ }\bibfield  {title} {\bibinfo
  {title} {Optimizing parametrized quantum circuits via noise-induced breaking
  of symmetries},\ }\href {https://arxiv.org/abs/2011.08763} {\bibfield
  {journal} {\bibinfo  {journal} {arXiv preprint arXiv:2011.08763}\ } (\bibinfo
  {year} {2020})}\BibitemShut {NoStop}%
\bibitem [{\citenamefont {Wang}\ \emph {et~al.}(2021)\citenamefont {Wang},
  \citenamefont {Czarnik}, \citenamefont {Arrasmith}, \citenamefont {Cerezo},
  \citenamefont {Cincio},\ and\ \citenamefont {Coles}}]{wang2021can}%
  \BibitemOpen
  \bibfield  {author} {\bibinfo {author} {\bibfnamefont {S.}~\bibnamefont
  {Wang}}, \bibinfo {author} {\bibfnamefont {P.}~\bibnamefont {Czarnik}},
  \bibinfo {author} {\bibfnamefont {A.}~\bibnamefont {Arrasmith}}, \bibinfo
  {author} {\bibfnamefont {M.}~\bibnamefont {Cerezo}}, \bibinfo {author}
  {\bibfnamefont {L.}~\bibnamefont {Cincio}},\ and\ \bibinfo {author}
  {\bibfnamefont {P.~J.}\ \bibnamefont {Coles}},\ }\bibfield  {title} {\bibinfo
  {title} {Can error mitigation improve trainability of noisy variational
  quantum algorithms?},\ }\href {https://arxiv.org/abs/2109.01051} {\bibfield
  {journal} {\bibinfo  {journal} {arXiv preprint arXiv:2109.01051}\ } (\bibinfo
  {year} {2021})}\BibitemShut {NoStop}%
\bibitem [{\citenamefont {Campos}\ \emph {et~al.}(2021)\citenamefont {Campos},
  \citenamefont {Rabinovich}, \citenamefont {Akshay},\ and\ \citenamefont
  {Biamonte}}]{campos2021training}%
  \BibitemOpen
  \bibfield  {author} {\bibinfo {author} {\bibfnamefont {E.}~\bibnamefont
  {Campos}}, \bibinfo {author} {\bibfnamefont {D.}~\bibnamefont {Rabinovich}},
  \bibinfo {author} {\bibfnamefont {V.}~\bibnamefont {Akshay}},\ and\ \bibinfo
  {author} {\bibfnamefont {J.}~\bibnamefont {Biamonte}},\ }\bibfield  {title}
  {\bibinfo {title} {Training saturation in layerwise quantum approximate
  optimization},\ }\href {https://doi.org/10.1103/PhysRevA.104.L030401}
  {\bibfield  {journal} {\bibinfo  {journal} {Phys. Rev. A}\ }\textbf {\bibinfo
  {volume} {104}},\ \bibinfo {pages} {L030401} (\bibinfo {year}
  {2021})}\BibitemShut {NoStop}%
\bibitem [{\citenamefont {Franca}\ and\ \citenamefont
  {Garcia-Patron}(2020)}]{franca2020limitations}%
  \BibitemOpen
  \bibfield  {author} {\bibinfo {author} {\bibfnamefont {D.~S.}\ \bibnamefont
  {Franca}}\ and\ \bibinfo {author} {\bibfnamefont {R.}~\bibnamefont
  {Garcia-Patron}},\ }\bibfield  {title} {\bibinfo {title} {Limitations of
  optimization algorithms on noisy quantum devices},\ }\href
  {https://arxiv.org/abs/2009.05532} {\bibfield  {journal} {\bibinfo  {journal}
  {arXiv preprint arXiv:2009.05532}\ } (\bibinfo {year} {2020})}\BibitemShut
  {NoStop}%
\bibitem [{\citenamefont {Wiersema}\ \emph {et~al.}(2020)\citenamefont
  {Wiersema}, \citenamefont {Zhou}, \citenamefont {de~Sereville}, \citenamefont
  {Carrasquilla}, \citenamefont {Kim},\ and\ \citenamefont
  {Yuen}}]{wiersema2020exploring}%
  \BibitemOpen
  \bibfield  {author} {\bibinfo {author} {\bibfnamefont {R.}~\bibnamefont
  {Wiersema}}, \bibinfo {author} {\bibfnamefont {C.}~\bibnamefont {Zhou}},
  \bibinfo {author} {\bibfnamefont {Y.}~\bibnamefont {de~Sereville}}, \bibinfo
  {author} {\bibfnamefont {J.~F.}\ \bibnamefont {Carrasquilla}}, \bibinfo
  {author} {\bibfnamefont {Y.~B.}\ \bibnamefont {Kim}},\ and\ \bibinfo {author}
  {\bibfnamefont {H.}~\bibnamefont {Yuen}},\ }\bibfield  {title} {\bibinfo
  {title} {Exploring entanglement and optimization within the hamiltonian
  variational ansatz},\ }\href {https://doi.org/10.1103/PRXQuantum.1.020319}
  {\bibfield  {journal} {\bibinfo  {journal} {PRX Quantum}\ }\textbf {\bibinfo
  {volume} {1}},\ \bibinfo {pages} {020319} (\bibinfo {year}
  {2020})}\BibitemShut {NoStop}%
\bibitem [{\citenamefont {Zhang}\ and\ \citenamefont
  {Cui}(2020)}]{zhang2020overparametrization}%
  \BibitemOpen
  \bibfield  {author} {\bibinfo {author} {\bibfnamefont {S.}~\bibnamefont
  {Zhang}}\ and\ \bibinfo {author} {\bibfnamefont {W.}~\bibnamefont {Cui}},\
  }\bibfield  {title} {\bibinfo {title} {Overparametrization in qaoa},\ }\href
  {http://www.henryyuen.net/fall2020/projects/overparameterization.pdf}
  {\bibfield  {journal} {\bibinfo  {journal} {Written Report}\ } (\bibinfo
  {year} {2020})}\BibitemShut {NoStop}%
\bibitem [{\citenamefont {Kiani}\ \emph {et~al.}(2020)\citenamefont {Kiani},
  \citenamefont {Lloyd},\ and\ \citenamefont {Maity}}]{kiani2020learning}%
  \BibitemOpen
  \bibfield  {author} {\bibinfo {author} {\bibfnamefont {B.~T.}\ \bibnamefont
  {Kiani}}, \bibinfo {author} {\bibfnamefont {S.}~\bibnamefont {Lloyd}},\ and\
  \bibinfo {author} {\bibfnamefont {R.}~\bibnamefont {Maity}},\ }\bibfield
  {title} {\bibinfo {title} {Learning unitaries by gradient descent},\ }\href
  {https://arxiv.org/abs/2001.11897} {\bibfield  {journal} {\bibinfo  {journal}
  {arXiv preprint arXiv:2001.11897}\ } (\bibinfo {year} {2020})}\BibitemShut
  {NoStop}%
\bibitem [{\citenamefont {Funcke}\ \emph {et~al.}(2021)\citenamefont {Funcke},
  \citenamefont {Hartung}, \citenamefont {Jansen}, \citenamefont {K{\"u}hn},
  \citenamefont {Schneider},\ and\ \citenamefont {Stornati}}]{funcke2021best}%
  \BibitemOpen
  \bibfield  {author} {\bibinfo {author} {\bibfnamefont {L.}~\bibnamefont
  {Funcke}}, \bibinfo {author} {\bibfnamefont {T.}~\bibnamefont {Hartung}},
  \bibinfo {author} {\bibfnamefont {K.}~\bibnamefont {Jansen}}, \bibinfo
  {author} {\bibfnamefont {S.}~\bibnamefont {K{\"u}hn}}, \bibinfo {author}
  {\bibfnamefont {M.}~\bibnamefont {Schneider}},\ and\ \bibinfo {author}
  {\bibfnamefont {P.}~\bibnamefont {Stornati}},\ }\bibfield  {title} {\bibinfo
  {title} {Best-approximation error for parametric quantum circuits},\ }\href
  {https://arxiv.org/abs/2107.07378} {\bibfield  {journal} {\bibinfo  {journal}
  {arXiv preprint arXiv:2107.07378}\ } (\bibinfo {year} {2021})}\BibitemShut
  {NoStop}%
\bibitem [{\citenamefont {Lee}\ \emph {et~al.}(2021)\citenamefont {Lee},
  \citenamefont {Magann}, \citenamefont {Rabitz},\ and\ \citenamefont
  {Arenz}}]{lee2021towards}%
  \BibitemOpen
  \bibfield  {author} {\bibinfo {author} {\bibfnamefont {J.}~\bibnamefont
  {Lee}}, \bibinfo {author} {\bibfnamefont {A.~B.}\ \bibnamefont {Magann}},
  \bibinfo {author} {\bibfnamefont {H.~A.}\ \bibnamefont {Rabitz}},\ and\
  \bibinfo {author} {\bibfnamefont {C.}~\bibnamefont {Arenz}},\ }\bibfield
  {title} {\bibinfo {title} {Progress toward favorable landscapes in quantum
  combinatorial optimization},\ }\href
  {https://doi.org/10.1103/PhysRevA.104.032401} {\bibfield  {journal} {\bibinfo
   {journal} {Physical Review A}\ }\textbf {\bibinfo {volume} {104}},\ \bibinfo
  {pages} {032401} (\bibinfo {year} {2021})}\BibitemShut {NoStop}%
\bibitem [{\citenamefont {Anschuetz}(2021)}]{anschuetz2021critical}%
  \BibitemOpen
  \bibfield  {author} {\bibinfo {author} {\bibfnamefont {E.~R.}\ \bibnamefont
  {Anschuetz}},\ }\bibfield  {title} {\bibinfo {title} {Critical points in
  hamiltonian agnostic variational quantum algorithms},\ }\href
  {https://arxiv.org/abs/2109.06957} {\bibfield  {journal} {\bibinfo  {journal}
  {arXiv preprint arXiv:2109.06957}\ } (\bibinfo {year} {2021})}\BibitemShut
  {NoStop}%
\bibitem [{\citenamefont {D'Alessandro}(2007)}]{dalessandro2010introduction}%
  \BibitemOpen
  \bibfield  {author} {\bibinfo {author} {\bibfnamefont {D.}~\bibnamefont
  {D'Alessandro}},\ }\href {https://books.google.sm/books?id=HbMYmAEACAAJ}
  {\emph {\bibinfo {title} {Introduction to Quantum Control and Dynamics}}},\
  Chapman \& Hall/CRC Applied Mathematics \& Nonlinear Science\ (\bibinfo
  {publisher} {Taylor \& Francis},\ \bibinfo {year} {2007})\BibitemShut
  {NoStop}%
\bibitem [{\citenamefont {Zeier}\ and\ \citenamefont
  {Schulte-Herbr{\"u}ggen}(2011)}]{zeier2011symmetry}%
  \BibitemOpen
  \bibfield  {author} {\bibinfo {author} {\bibfnamefont {R.}~\bibnamefont
  {Zeier}}\ and\ \bibinfo {author} {\bibfnamefont {T.}~\bibnamefont
  {Schulte-Herbr{\"u}ggen}},\ }\bibfield  {title} {\bibinfo {title} {Symmetry
  principles in quantum systems theory},\ }\href
  {https://doi.org/https://doi.org/10.1063/1.3657939} {\bibfield  {journal}
  {\bibinfo  {journal} {Journal of mathematical physics}\ }\textbf {\bibinfo
  {volume} {52}},\ \bibinfo {pages} {113510} (\bibinfo {year}
  {2011})}\BibitemShut {NoStop}%
\bibitem [{\citenamefont {Haug}\ \emph {et~al.}(2021)\citenamefont {Haug},
  \citenamefont {Bharti},\ and\ \citenamefont {Kim}}]{haug2021capacity}%
  \BibitemOpen
  \bibfield  {author} {\bibinfo {author} {\bibfnamefont {T.}~\bibnamefont
  {Haug}}, \bibinfo {author} {\bibfnamefont {K.}~\bibnamefont {Bharti}},\ and\
  \bibinfo {author} {\bibfnamefont {M.}~\bibnamefont {Kim}},\ }\bibfield
  {title} {\bibinfo {title} {Capacity and quantum geometry of parametrized
  quantum circuits},\ }\href {https://arxiv.org/abs/2102.01659} {\bibfield
  {journal} {\bibinfo  {journal} {arXiv preprint arXiv:2102.01659}\ } (\bibinfo
  {year} {2021})}\BibitemShut {NoStop}%
\bibitem [{\citenamefont {Kim}\ \emph {et~al.}(2021)\citenamefont {Kim},
  \citenamefont {Kim},\ and\ \citenamefont {Rosa}}]{kim2021universal}%
  \BibitemOpen
  \bibfield  {author} {\bibinfo {author} {\bibfnamefont {J.}~\bibnamefont
  {Kim}}, \bibinfo {author} {\bibfnamefont {J.}~\bibnamefont {Kim}},\ and\
  \bibinfo {author} {\bibfnamefont {D.}~\bibnamefont {Rosa}},\ }\bibfield
  {title} {\bibinfo {title} {Universal effectiveness of high-depth circuits in
  variational eigenproblems},\ }\href
  {https://doi.org/10.1103/physrevresearch.3.023203} {\bibfield  {journal}
  {\bibinfo  {journal} {Physical Review Research}\ }\textbf {\bibinfo {volume}
  {3}},\ \bibinfo {pages} {023203} (\bibinfo {year} {2021})}\BibitemShut
  {NoStop}%
\bibitem [{\citenamefont {d'Alessandro}(2007)}]{d2007introduction}%
  \BibitemOpen
  \bibfield  {author} {\bibinfo {author} {\bibfnamefont {D.}~\bibnamefont
  {d'Alessandro}},\ }\href@noop {} {\emph {\bibinfo {title} {Introduction to
  quantum control and dynamics}}}\ (\bibinfo  {publisher} {CRC press},\
  \bibinfo {year} {2007})\BibitemShut {NoStop}%
\bibitem [{\citenamefont {Chakrabarti}\ and\ \citenamefont
  {Rabitz}(2007)}]{chakrabarti2007quantum}%
  \BibitemOpen
  \bibfield  {author} {\bibinfo {author} {\bibfnamefont {R.}~\bibnamefont
  {Chakrabarti}}\ and\ \bibinfo {author} {\bibfnamefont {H.}~\bibnamefont
  {Rabitz}},\ }\bibfield  {title} {\bibinfo {title} {Quantum control
  landscapes},\ }\href {https://doi.org/10.1080/01442350701633300} {\bibfield
  {journal} {\bibinfo  {journal} {International Reviews in Physical Chemistry}\
  }\textbf {\bibinfo {volume} {26}},\ \bibinfo {pages} {671} (\bibinfo {year}
  {2007})}\BibitemShut {NoStop}%
\bibitem [{\citenamefont {Larocca}\ \emph
  {et~al.}(2020{\natexlab{a}})\citenamefont {Larocca}, \citenamefont
  {Calzetta},\ and\ \citenamefont {Wisniacki}}]{larocca2020exploiting}%
  \BibitemOpen
  \bibfield  {author} {\bibinfo {author} {\bibfnamefont {M.}~\bibnamefont
  {Larocca}}, \bibinfo {author} {\bibfnamefont {E.}~\bibnamefont {Calzetta}},\
  and\ \bibinfo {author} {\bibfnamefont {D.~A.}\ \bibnamefont {Wisniacki}},\
  }\bibfield  {title} {\bibinfo {title} {Exploiting landscape geometry to
  enhance quantum optimal control},\ }\href
  {https://doi.org/10.1103/PhysRevA.101.023410} {\bibfield  {journal} {\bibinfo
   {journal} {Phys. Rev. A}\ }\textbf {\bibinfo {volume} {101}},\ \bibinfo
  {pages} {023410} (\bibinfo {year} {2020}{\natexlab{a}})}\BibitemShut
  {NoStop}%
\bibitem [{\citenamefont {Larocca}\ \emph
  {et~al.}(2020{\natexlab{b}})\citenamefont {Larocca}, \citenamefont
  {Calzetta},\ and\ \citenamefont {Wisniacki}}]{larocca2020fourier}%
  \BibitemOpen
  \bibfield  {author} {\bibinfo {author} {\bibfnamefont {M.}~\bibnamefont
  {Larocca}}, \bibinfo {author} {\bibfnamefont {E.}~\bibnamefont {Calzetta}},\
  and\ \bibinfo {author} {\bibfnamefont {D.}~\bibnamefont {Wisniacki}},\
  }\bibfield  {title} {\bibinfo {title} {Fourier compression: A customization
  method for quantum control protocols},\ }\href
  {https://doi.org/10.1103/PhysRevA.102.033108} {\bibfield  {journal} {\bibinfo
   {journal} {Phys. Rev. A}\ }\textbf {\bibinfo {volume} {102}},\ \bibinfo
  {pages} {033108} (\bibinfo {year} {2020}{\natexlab{b}})}\BibitemShut
  {NoStop}%
\bibitem [{\citenamefont {Boixo}\ \emph {et~al.}(2018)\citenamefont {Boixo},
  \citenamefont {Isakov}, \citenamefont {Smelyanskiy}, \citenamefont {Babbush},
  \citenamefont {Ding}, \citenamefont {Jiang}, \citenamefont {Bremner},
  \citenamefont {Martinis},\ and\ \citenamefont
  {Neven}}]{boixo2018characterizing}%
  \BibitemOpen
  \bibfield  {author} {\bibinfo {author} {\bibfnamefont {S.}~\bibnamefont
  {Boixo}}, \bibinfo {author} {\bibfnamefont {S.~V.}\ \bibnamefont {Isakov}},
  \bibinfo {author} {\bibfnamefont {V.~N.}\ \bibnamefont {Smelyanskiy}},
  \bibinfo {author} {\bibfnamefont {R.}~\bibnamefont {Babbush}}, \bibinfo
  {author} {\bibfnamefont {N.}~\bibnamefont {Ding}}, \bibinfo {author}
  {\bibfnamefont {Z.}~\bibnamefont {Jiang}}, \bibinfo {author} {\bibfnamefont
  {M.~J.}\ \bibnamefont {Bremner}}, \bibinfo {author} {\bibfnamefont {J.~M.}\
  \bibnamefont {Martinis}},\ and\ \bibinfo {author} {\bibfnamefont
  {H.}~\bibnamefont {Neven}},\ }\bibfield  {title} {\bibinfo {title}
  {Characterizing quantum supremacy in near-term devices},\ }\href
  {https://doi.org/0.1038/s41567-018-0124-x} {\bibfield  {journal} {\bibinfo
  {journal} {Nature Physics}\ }\textbf {\bibinfo {volume} {14}},\ \bibinfo
  {pages} {595} (\bibinfo {year} {2018})}\BibitemShut {NoStop}%
\bibitem [{\citenamefont {Arute}\ \emph {et~al.}(2019)\citenamefont {Arute},
  \citenamefont {Arya}, \citenamefont {Babbush}, \citenamefont {Bacon} \emph
  {et~al.}}]{google2019supremacy}%
  \BibitemOpen
  \bibfield  {author} {\bibinfo {author} {\bibfnamefont {F.}~\bibnamefont
  {Arute}}, \bibinfo {author} {\bibfnamefont {K.}~\bibnamefont {Arya}},
  \bibinfo {author} {\bibfnamefont {R.}~\bibnamefont {Babbush}}, \bibinfo
  {author} {\bibfnamefont {D.}~\bibnamefont {Bacon}}, \emph {et~al.},\
  }\bibfield  {title} {\bibinfo {title} {Quantum supremacy using a programmable
  superconducting processor},\ }\href
  {https://doi.org/10.1038/s41586-019-1666-5} {\bibfield  {journal} {\bibinfo
  {journal} {Nature}\ }\textbf {\bibinfo {volume} {574}},\ \bibinfo {pages}
  {505} (\bibinfo {year} {2019})}\BibitemShut {NoStop}%
\bibitem [{\citenamefont {Kandala}\ \emph {et~al.}(2017)\citenamefont
  {Kandala}, \citenamefont {Mezzacapo}, \citenamefont {Temme}, \citenamefont
  {Takita}, \citenamefont {Brink}, \citenamefont {Chow},\ and\ \citenamefont
  {Gambetta}}]{kandala2017hardware}%
  \BibitemOpen
  \bibfield  {author} {\bibinfo {author} {\bibfnamefont {A.}~\bibnamefont
  {Kandala}}, \bibinfo {author} {\bibfnamefont {A.}~\bibnamefont {Mezzacapo}},
  \bibinfo {author} {\bibfnamefont {K.}~\bibnamefont {Temme}}, \bibinfo
  {author} {\bibfnamefont {M.}~\bibnamefont {Takita}}, \bibinfo {author}
  {\bibfnamefont {M.}~\bibnamefont {Brink}}, \bibinfo {author} {\bibfnamefont
  {J.~M.}\ \bibnamefont {Chow}},\ and\ \bibinfo {author} {\bibfnamefont
  {J.~M.}\ \bibnamefont {Gambetta}},\ }\bibfield  {title} {\bibinfo {title}
  {Hardware-efficient variational quantum eigensolver for small molecules and
  quantum magnets},\ }\href {https://doi.org/10.1038/nature23879} {\bibfield
  {journal} {\bibinfo  {journal} {Nature}\ }\textbf {\bibinfo {volume} {549}},\
  \bibinfo {pages} {242} (\bibinfo {year} {2017})}\BibitemShut {NoStop}%
\bibitem [{\citenamefont {Farhi}\ \emph {et~al.}(2014)\citenamefont {Farhi},
  \citenamefont {Goldstone},\ and\ \citenamefont {Gutmann}}]{farhi2014quantum}%
  \BibitemOpen
  \bibfield  {author} {\bibinfo {author} {\bibfnamefont {E.}~\bibnamefont
  {Farhi}}, \bibinfo {author} {\bibfnamefont {J.}~\bibnamefont {Goldstone}},\
  and\ \bibinfo {author} {\bibfnamefont {S.}~\bibnamefont {Gutmann}},\
  }\bibfield  {title} {\bibinfo {title} {A quantum approximate optimization
  algorithm},\ }\href {https://arxiv.org/abs/1411.4028} {\bibfield  {journal}
  {\bibinfo  {journal} {arXiv preprint arXiv:1411.4028}\ } (\bibinfo {year}
  {2014})}\BibitemShut {NoStop}%
\bibitem [{\citenamefont {Hadfield}\ \emph {et~al.}(2019)\citenamefont
  {Hadfield}, \citenamefont {Wang}, \citenamefont {O'Gorman}, \citenamefont
  {Rieffel}, \citenamefont {Venturelli},\ and\ \citenamefont
  {Biswas}}]{hadfield2019quantum}%
  \BibitemOpen
  \bibfield  {author} {\bibinfo {author} {\bibfnamefont {S.}~\bibnamefont
  {Hadfield}}, \bibinfo {author} {\bibfnamefont {Z.}~\bibnamefont {Wang}},
  \bibinfo {author} {\bibfnamefont {B.}~\bibnamefont {O'Gorman}}, \bibinfo
  {author} {\bibfnamefont {E.~G.}\ \bibnamefont {Rieffel}}, \bibinfo {author}
  {\bibfnamefont {D.}~\bibnamefont {Venturelli}},\ and\ \bibinfo {author}
  {\bibfnamefont {R.}~\bibnamefont {Biswas}},\ }\bibfield  {title} {\bibinfo
  {title} {From the quantum approximate optimization algorithm to a quantum
  alternating operator ansatz},\ }\href
  {https://www.mdpi.com/1999-4893/12/2/34} {\bibfield  {journal} {\bibinfo
  {journal} {Algorithms}\ }\textbf {\bibinfo {volume} {12}},\ \bibinfo {pages}
  {34} (\bibinfo {year} {2019})}\BibitemShut {NoStop}%
\bibitem [{\citenamefont {Zhu}\ \emph {et~al.}(2020)\citenamefont {Zhu},
  \citenamefont {Tang}, \citenamefont {Barron}, \citenamefont {Mayhall},
  \citenamefont {Barnes},\ and\ \citenamefont {Economou}}]{zhu2020adaptive}%
  \BibitemOpen
  \bibfield  {author} {\bibinfo {author} {\bibfnamefont {L.}~\bibnamefont
  {Zhu}}, \bibinfo {author} {\bibfnamefont {H.~L.}\ \bibnamefont {Tang}},
  \bibinfo {author} {\bibfnamefont {G.~S.}\ \bibnamefont {Barron}}, \bibinfo
  {author} {\bibfnamefont {N.~J.}\ \bibnamefont {Mayhall}}, \bibinfo {author}
  {\bibfnamefont {E.}~\bibnamefont {Barnes}},\ and\ \bibinfo {author}
  {\bibfnamefont {S.~E.}\ \bibnamefont {Economou}},\ }\bibfield  {title}
  {\bibinfo {title} {An adaptive quantum approximate optimization algorithm for
  solving combinatorial problems on a quantum computer},\ }\href
  {https://arxiv.org/abs/2005.10258} {\bibfield  {journal} {\bibinfo  {journal}
  {arXiv preprint arXiv:2005.10258}\ } (\bibinfo {year} {2020})}\BibitemShut
  {NoStop}%
\bibitem [{\citenamefont {Wecker}\ \emph {et~al.}(2015)\citenamefont {Wecker},
  \citenamefont {Hastings},\ and\ \citenamefont {Troyer}}]{wecker2015progress}%
  \BibitemOpen
  \bibfield  {author} {\bibinfo {author} {\bibfnamefont {D.}~\bibnamefont
  {Wecker}}, \bibinfo {author} {\bibfnamefont {M.~B.}\ \bibnamefont
  {Hastings}},\ and\ \bibinfo {author} {\bibfnamefont {M.}~\bibnamefont
  {Troyer}},\ }\bibfield  {title} {\bibinfo {title} {Progress towards practical
  quantum variational algorithms},\ }\href
  {https://doi.org/10.1103/PhysRevA.92.042303} {\bibfield  {journal} {\bibinfo
  {journal} {Phys. Rev. A}\ }\textbf {\bibinfo {volume} {92}},\ \bibinfo
  {pages} {042303} (\bibinfo {year} {2015})}\BibitemShut {NoStop}%
\bibitem [{\citenamefont {Choquette}\ \emph {et~al.}(2021)\citenamefont
  {Choquette}, \citenamefont {Di~Paolo}, \citenamefont {Barkoutsos},
  \citenamefont {S{\'e}n{\'e}chal}, \citenamefont {Tavernelli},\ and\
  \citenamefont {Blais}}]{choquette2020quantum}%
  \BibitemOpen
  \bibfield  {author} {\bibinfo {author} {\bibfnamefont {A.}~\bibnamefont
  {Choquette}}, \bibinfo {author} {\bibfnamefont {A.}~\bibnamefont {Di~Paolo}},
  \bibinfo {author} {\bibfnamefont {P.~K.}\ \bibnamefont {Barkoutsos}},
  \bibinfo {author} {\bibfnamefont {D.}~\bibnamefont {S{\'e}n{\'e}chal}},
  \bibinfo {author} {\bibfnamefont {I.}~\bibnamefont {Tavernelli}},\ and\
  \bibinfo {author} {\bibfnamefont {A.}~\bibnamefont {Blais}},\ }\bibfield
  {title} {\bibinfo {title} {Quantum-optimal-control-inspired ansatz for
  variational quantum algorithms},\ }\href
  {https://doi.org/10.1103/PhysRevResearch.3.023092} {\bibfield  {journal}
  {\bibinfo  {journal} {Physical Review Research}\ }\textbf {\bibinfo {volume}
  {3}},\ \bibinfo {pages} {023092} (\bibinfo {year} {2021})}\BibitemShut
  {NoStop}%
\bibitem [{\citenamefont {Morales}\ \emph {et~al.}(2020)\citenamefont
  {Morales}, \citenamefont {Biamonte},\ and\ \citenamefont
  {Zimbor{\'a}s}}]{morales2020universality}%
  \BibitemOpen
  \bibfield  {author} {\bibinfo {author} {\bibfnamefont {M.~E.}\ \bibnamefont
  {Morales}}, \bibinfo {author} {\bibfnamefont {J.}~\bibnamefont {Biamonte}},\
  and\ \bibinfo {author} {\bibfnamefont {Z.}~\bibnamefont {Zimbor{\'a}s}},\
  }\bibfield  {title} {\bibinfo {title} {On the universality of the quantum
  approximate optimization algorithm},\ }\href
  {https://doi.org/10.1007/s11128-020-02748-9} {\bibfield  {journal} {\bibinfo
  {journal} {Quantum Information Processing}\ }\textbf {\bibinfo {volume}
  {19}},\ \bibinfo {pages} {1} (\bibinfo {year} {2020})}\BibitemShut {NoStop}%
\bibitem [{\citenamefont {Sim}\ \emph {et~al.}(2019)\citenamefont {Sim},
  \citenamefont {Johnson},\ and\ \citenamefont
  {Aspuru-Guzik}}]{sim2019expressibility}%
  \BibitemOpen
  \bibfield  {author} {\bibinfo {author} {\bibfnamefont {S.}~\bibnamefont
  {Sim}}, \bibinfo {author} {\bibfnamefont {P.~D.}\ \bibnamefont {Johnson}},\
  and\ \bibinfo {author} {\bibfnamefont {A.}~\bibnamefont {Aspuru-Guzik}},\
  }\bibfield  {title} {\bibinfo {title} {Expressibility and entangling
  capability of parameterized quantum circuits for hybrid quantum-classical
  algorithms},\ }\href
  {https://onlinelibrary.wiley.com/doi/full/10.1002/qute.201900070} {\bibfield
  {journal} {\bibinfo  {journal} {Advanced Quantum Technologies}\ }\textbf
  {\bibinfo {volume} {2}},\ \bibinfo {pages} {1900070} (\bibinfo {year}
  {2019})}\BibitemShut {NoStop}%
\bibitem [{\citenamefont {Cheng}(2010)}]{cheng2010quantum}%
  \BibitemOpen
  \bibfield  {author} {\bibinfo {author} {\bibfnamefont {R.}~\bibnamefont
  {Cheng}},\ }\bibfield  {title} {\bibinfo {title} {Quantum geometric tensor
  (fubini-study metric) in simple quantum system: A pedagogical introduction},\
  }\href {https://arxiv.org/abs/1012.1337} {\bibfield  {journal} {\bibinfo
  {journal} {arXiv preprint arXiv:1012.1337}\ } (\bibinfo {year}
  {2010})}\BibitemShut {NoStop}%
\bibitem [{\citenamefont {Meyer}(2021{\natexlab{a}})}]{meyer2021fisher}%
  \BibitemOpen
  \bibfield  {author} {\bibinfo {author} {\bibfnamefont {J.~J.}\ \bibnamefont
  {Meyer}},\ }\bibfield  {title} {\bibinfo {title} {Fisher {I}nformation in
  {N}oisy {I}ntermediate-{S}cale {Q}uantum {A}pplications},\ }\href
  {https://doi.org/10.22331/q-2021-09-09-539} {\bibfield  {journal} {\bibinfo
  {journal} {{Quantum}}\ }\textbf {\bibinfo {volume} {5}},\ \bibinfo {pages}
  {539} (\bibinfo {year} {2021}{\natexlab{a}})}\BibitemShut {NoStop}%
\bibitem [{\citenamefont {Liu}\ \emph {et~al.}(2019)\citenamefont {Liu},
  \citenamefont {Yuan}, \citenamefont {Lu},\ and\ \citenamefont
  {Wang}}]{liu2019quantum}%
  \BibitemOpen
  \bibfield  {author} {\bibinfo {author} {\bibfnamefont {J.}~\bibnamefont
  {Liu}}, \bibinfo {author} {\bibfnamefont {H.}~\bibnamefont {Yuan}}, \bibinfo
  {author} {\bibfnamefont {X.-M.}\ \bibnamefont {Lu}},\ and\ \bibinfo {author}
  {\bibfnamefont {X.}~\bibnamefont {Wang}},\ }\bibfield  {title} {\bibinfo
  {title} {Quantum fisher information matrix and multiparameter estimation},\
  }\href {https://doi.org/10.1088/1751-8121/ab5d4d} {\bibfield  {journal}
  {\bibinfo  {journal} {Journal of Physics A: Mathematical and Theoretical}\
  }\textbf {\bibinfo {volume} {53}},\ \bibinfo {pages} {023001} (\bibinfo
  {year} {2019})}\BibitemShut {NoStop}%
\bibitem [{\citenamefont {McArdle}\ \emph {et~al.}(2019)\citenamefont
  {McArdle}, \citenamefont {Jones}, \citenamefont {Endo}, \citenamefont {Li},
  \citenamefont {Benjamin},\ and\ \citenamefont
  {Yuan}}]{mcardle2019variational}%
  \BibitemOpen
  \bibfield  {author} {\bibinfo {author} {\bibfnamefont {S.}~\bibnamefont
  {McArdle}}, \bibinfo {author} {\bibfnamefont {T.}~\bibnamefont {Jones}},
  \bibinfo {author} {\bibfnamefont {S.}~\bibnamefont {Endo}}, \bibinfo {author}
  {\bibfnamefont {Y.}~\bibnamefont {Li}}, \bibinfo {author} {\bibfnamefont
  {S.~C.}\ \bibnamefont {Benjamin}},\ and\ \bibinfo {author} {\bibfnamefont
  {X.}~\bibnamefont {Yuan}},\ }\bibfield  {title} {\bibinfo {title}
  {Variational ansatz-based quantum simulation of imaginary time evolution},\
  }\href {https://doi.org/10.1038/s41534-019-0187-2} {\bibfield  {journal}
  {\bibinfo  {journal} {npj Quantum Information}\ }\textbf {\bibinfo {volume}
  {5}},\ \bibinfo {pages} {1} (\bibinfo {year} {2019})}\BibitemShut {NoStop}%
\bibitem [{\citenamefont {Stokes}\ \emph {et~al.}(2020)\citenamefont {Stokes},
  \citenamefont {Izaac}, \citenamefont {Killoran},\ and\ \citenamefont
  {Carleo}}]{stokes2020quantum}%
  \BibitemOpen
  \bibfield  {author} {\bibinfo {author} {\bibfnamefont {J.}~\bibnamefont
  {Stokes}}, \bibinfo {author} {\bibfnamefont {J.}~\bibnamefont {Izaac}},
  \bibinfo {author} {\bibfnamefont {N.}~\bibnamefont {Killoran}},\ and\
  \bibinfo {author} {\bibfnamefont {G.}~\bibnamefont {Carleo}},\ }\bibfield
  {title} {\bibinfo {title} {Quantum natural gradient},\ }\href
  {https://doi.org/10.22331/q-2020-05-25-269} {\bibfield  {journal} {\bibinfo
  {journal} {Quantum}\ }\textbf {\bibinfo {volume} {4}},\ \bibinfo {pages}
  {269} (\bibinfo {year} {2020})}\BibitemShut {NoStop}%
\bibitem [{\citenamefont {Koczor}\ and\ \citenamefont
  {Benjamin}(2019)}]{koczor2019quantum}%
  \BibitemOpen
  \bibfield  {author} {\bibinfo {author} {\bibfnamefont {B.}~\bibnamefont
  {Koczor}}\ and\ \bibinfo {author} {\bibfnamefont {S.~C.}\ \bibnamefont
  {Benjamin}},\ }\bibfield  {title} {\bibinfo {title} {Quantum natural gradient
  generalised to non-unitary circuits},\ }\href
  {https://arxiv.org/abs/1912.08660} {\bibfield  {journal} {\bibinfo  {journal}
  {arXiv preprint arXiv:1912.08660}\ } (\bibinfo {year} {2019})}\BibitemShut
  {NoStop}%
\bibitem [{\citenamefont {Gacon}\ \emph {et~al.}(2021)\citenamefont {Gacon},
  \citenamefont {Zoufal}, \citenamefont {Carleo},\ and\ \citenamefont
  {Woerner}}]{gacon2021simultaneous}%
  \BibitemOpen
  \bibfield  {author} {\bibinfo {author} {\bibfnamefont {J.}~\bibnamefont
  {Gacon}}, \bibinfo {author} {\bibfnamefont {C.}~\bibnamefont {Zoufal}},
  \bibinfo {author} {\bibfnamefont {G.}~\bibnamefont {Carleo}},\ and\ \bibinfo
  {author} {\bibfnamefont {S.}~\bibnamefont {Woerner}},\ }\bibfield  {title}
  {\bibinfo {title} {Simultaneous perturbation stochastic approximation of the
  quantum fisher information},\ }\href {https://arxiv.org/abs/2103.09232}
  {\bibfield  {journal} {\bibinfo  {journal} {arXiv preprint arXiv:2103.09232}\
  } (\bibinfo {year} {2021})}\BibitemShut {NoStop}%
\bibitem [{\citenamefont {Haug}\ and\ \citenamefont
  {Kim}(2021)}]{haug2021natural}%
  \BibitemOpen
  \bibfield  {author} {\bibinfo {author} {\bibfnamefont {T.}~\bibnamefont
  {Haug}}\ and\ \bibinfo {author} {\bibfnamefont {M.}~\bibnamefont {Kim}},\
  }\bibfield  {title} {\bibinfo {title} {Natural parameterized quantum
  circuit},\ }\href {https://arxiv.org/abs/2107.14063} {\bibfield  {journal}
  {\bibinfo  {journal} {arXiv preprint arXiv:2107.14063}\ } (\bibinfo {year}
  {2021})}\BibitemShut {NoStop}%
\bibitem [{\citenamefont {Kim}\ and\ \citenamefont
  {Oz}(2021)}]{kim2021quantum}%
  \BibitemOpen
  \bibfield  {author} {\bibinfo {author} {\bibfnamefont {J.}~\bibnamefont
  {Kim}}\ and\ \bibinfo {author} {\bibfnamefont {Y.}~\bibnamefont {Oz}},\
  }\bibfield  {title} {\bibinfo {title} {Quantum energy landscape and vqa
  optimization},\ }\href {https://arxiv.org/abs/2107.10166} {\bibfield
  {journal} {\bibinfo  {journal} {arXiv preprint arXiv:2107.10166}\ } (\bibinfo
  {year} {2021})}\BibitemShut {NoStop}%
\bibitem [{\citenamefont {Dalgaard}\ \emph {et~al.}(2021)\citenamefont
  {Dalgaard}, \citenamefont {Sherson},\ and\ \citenamefont
  {Motzoi}}]{dalgaard2021predicting}%
  \BibitemOpen
  \bibfield  {author} {\bibinfo {author} {\bibfnamefont {M.}~\bibnamefont
  {Dalgaard}}, \bibinfo {author} {\bibfnamefont {J.}~\bibnamefont {Sherson}},\
  and\ \bibinfo {author} {\bibfnamefont {F.}~\bibnamefont {Motzoi}},\
  }\bibfield  {title} {\bibinfo {title} {Predicting quantum dynamical cost
  landscapes with deep learning},\ }\href {https://arxiv.org/abs/2107.00008}
  {\bibfield  {journal} {\bibinfo  {journal} {arXiv preprint arXiv:2107.00008}\
  } (\bibinfo {year} {2021})}\BibitemShut {NoStop}%
\bibitem [{\citenamefont {Dalgaard}\ \emph {et~al.}(2020)\citenamefont
  {Dalgaard}, \citenamefont {Motzoi}, \citenamefont {Jensen},\ and\
  \citenamefont {Sherson}}]{dalgaard2020hessian}%
  \BibitemOpen
  \bibfield  {author} {\bibinfo {author} {\bibfnamefont {M.}~\bibnamefont
  {Dalgaard}}, \bibinfo {author} {\bibfnamefont {F.}~\bibnamefont {Motzoi}},
  \bibinfo {author} {\bibfnamefont {J.~H.~M.}\ \bibnamefont {Jensen}},\ and\
  \bibinfo {author} {\bibfnamefont {J.}~\bibnamefont {Sherson}},\ }\bibfield
  {title} {\bibinfo {title} {Hessian-based optimization of constrained quantum
  control},\ }\href
  {https://doi.org/https://doi.org/10.1103/PhysRevA.102.042612} {\bibfield
  {journal} {\bibinfo  {journal} {Physical Review A}\ }\textbf {\bibinfo
  {volume} {102}},\ \bibinfo {pages} {042612} (\bibinfo {year}
  {2020})}\BibitemShut {NoStop}%
\bibitem [{\citenamefont {Moore}\ and\ \citenamefont
  {Rabitz}(2012)}]{moore2012exploring}%
  \BibitemOpen
  \bibfield  {author} {\bibinfo {author} {\bibfnamefont {K.~W.}\ \bibnamefont
  {Moore}}\ and\ \bibinfo {author} {\bibfnamefont {H.}~\bibnamefont {Rabitz}},\
  }\bibfield  {title} {\bibinfo {title} {Exploring constrained quantum control
  landscapes},\ }\href {https://doi.org/https://doi.org/10.1063/1.4757133}
  {\bibfield  {journal} {\bibinfo  {journal} {The Journal of chemical physics}\
  }\textbf {\bibinfo {volume} {137}},\ \bibinfo {pages} {134113} (\bibinfo
  {year} {2012})}\BibitemShut {NoStop}%
\bibitem [{\citenamefont {Wu}\ \emph {et~al.}(2012)\citenamefont {Wu},
  \citenamefont {Long}, \citenamefont {Dominy}, \citenamefont {Ho},\ and\
  \citenamefont {Rabitz}}]{wu2012singularities}%
  \BibitemOpen
  \bibfield  {author} {\bibinfo {author} {\bibfnamefont {R.-B.}\ \bibnamefont
  {Wu}}, \bibinfo {author} {\bibfnamefont {R.}~\bibnamefont {Long}}, \bibinfo
  {author} {\bibfnamefont {J.}~\bibnamefont {Dominy}}, \bibinfo {author}
  {\bibfnamefont {T.-S.}\ \bibnamefont {Ho}},\ and\ \bibinfo {author}
  {\bibfnamefont {H.}~\bibnamefont {Rabitz}},\ }\bibfield  {title} {\bibinfo
  {title} {Singularities of quantum control landscapes},\ }\href
  {https://doi.org/https://doi.org/10.1103/PhysRevA.86.013405} {\bibfield
  {journal} {\bibinfo  {journal} {Physical Review A}\ }\textbf {\bibinfo
  {volume} {86}},\ \bibinfo {pages} {013405} (\bibinfo {year}
  {2012})}\BibitemShut {NoStop}%
\bibitem [{\citenamefont {Riviello}\ \emph {et~al.}(2014)\citenamefont
  {Riviello}, \citenamefont {Brif}, \citenamefont {Long}, \citenamefont {Wu},
  \citenamefont {Tibbetts}, \citenamefont {Ho},\ and\ \citenamefont
  {Rabitz}}]{riviello2014searching}%
  \BibitemOpen
  \bibfield  {author} {\bibinfo {author} {\bibfnamefont {G.}~\bibnamefont
  {Riviello}}, \bibinfo {author} {\bibfnamefont {C.}~\bibnamefont {Brif}},
  \bibinfo {author} {\bibfnamefont {R.}~\bibnamefont {Long}}, \bibinfo {author}
  {\bibfnamefont {R.-B.}\ \bibnamefont {Wu}}, \bibinfo {author} {\bibfnamefont
  {K.~M.}\ \bibnamefont {Tibbetts}}, \bibinfo {author} {\bibfnamefont {T.-S.}\
  \bibnamefont {Ho}},\ and\ \bibinfo {author} {\bibfnamefont {H.}~\bibnamefont
  {Rabitz}},\ }\bibfield  {title} {\bibinfo {title} {Searching for quantum
  optimal control fields in the presence of singular critical points},\ }\href
  {https://doi.org/https://doi.org/10.1103/PhysRevA.90.013404} {\bibfield
  {journal} {\bibinfo  {journal} {Physical Review A}\ }\textbf {\bibinfo
  {volume} {90}},\ \bibinfo {pages} {013404} (\bibinfo {year}
  {2014})}\BibitemShut {NoStop}%
\bibitem [{\citenamefont {Rach}\ \emph {et~al.}(2015)\citenamefont {Rach},
  \citenamefont {M{\"u}ller}, \citenamefont {Calarco},\ and\ \citenamefont
  {Montangero}}]{rach2015dressing}%
  \BibitemOpen
  \bibfield  {author} {\bibinfo {author} {\bibfnamefont {N.}~\bibnamefont
  {Rach}}, \bibinfo {author} {\bibfnamefont {M.~M.}\ \bibnamefont
  {M{\"u}ller}}, \bibinfo {author} {\bibfnamefont {T.}~\bibnamefont
  {Calarco}},\ and\ \bibinfo {author} {\bibfnamefont {S.}~\bibnamefont
  {Montangero}},\ }\bibfield  {title} {\bibinfo {title} {Dressing the
  chopped-random-basis optimization: A bandwidth-limited access to the
  trap-free landscape},\ }\href
  {https://doi.org/https://doi.org/10.1103/PhysRevA.92.062343} {\bibfield
  {journal} {\bibinfo  {journal} {Physical Review A}\ }\textbf {\bibinfo
  {volume} {92}},\ \bibinfo {pages} {062343} (\bibinfo {year}
  {2015})}\BibitemShut {NoStop}%
\bibitem [{\citenamefont {Larocca}\ \emph {et~al.}(2018)\citenamefont
  {Larocca}, \citenamefont {Poggi},\ and\ \citenamefont
  {Wisniacki}}]{larocca2018quantum}%
  \BibitemOpen
  \bibfield  {author} {\bibinfo {author} {\bibfnamefont {M.}~\bibnamefont
  {Larocca}}, \bibinfo {author} {\bibfnamefont {P.~M.}\ \bibnamefont {Poggi}},\
  and\ \bibinfo {author} {\bibfnamefont {D.~A.}\ \bibnamefont {Wisniacki}},\
  }\bibfield  {title} {\bibinfo {title} {Quantum control landscape for a
  two-level system near the quantum speed limit},\ }\href
  {https://doi.org/10.1088/1751-8121/aad657} {\bibfield  {journal} {\bibinfo
  {journal} {Journal of Physics A: Mathematical and Theoretical}\ }\textbf
  {\bibinfo {volume} {51}},\ \bibinfo {pages} {385305} (\bibinfo {year}
  {2018})}\BibitemShut {NoStop}%
\bibitem [{\citenamefont {Coles}(2021)}]{coles2021seeking}%
  \BibitemOpen
  \bibfield  {author} {\bibinfo {author} {\bibfnamefont {P.~J.}\ \bibnamefont
  {Coles}},\ }\bibfield  {title} {\bibinfo {title} {Seeking quantum advantage
  for neural networks},\ }\href {https://doi.org//10.1038/s43588-021-00088-x}
  {\bibfield  {journal} {\bibinfo  {journal} {Nature Computational Science}\
  }\textbf {\bibinfo {volume} {1}},\ \bibinfo {pages} {389} (\bibinfo {year}
  {2021})}\BibitemShut {NoStop}%
\bibitem [{\citenamefont {Havl{\'\i}{\v{c}}ek}\ \emph
  {et~al.}(2019)\citenamefont {Havl{\'\i}{\v{c}}ek}, \citenamefont
  {C{\'o}rcoles}, \citenamefont {Temme}, \citenamefont {Harrow}, \citenamefont
  {Kandala}, \citenamefont {Chow},\ and\ \citenamefont
  {Gambetta}}]{havlivcek2019supervised}%
  \BibitemOpen
  \bibfield  {author} {\bibinfo {author} {\bibfnamefont {V.}~\bibnamefont
  {Havl{\'\i}{\v{c}}ek}}, \bibinfo {author} {\bibfnamefont {A.~D.}\
  \bibnamefont {C{\'o}rcoles}}, \bibinfo {author} {\bibfnamefont
  {K.}~\bibnamefont {Temme}}, \bibinfo {author} {\bibfnamefont {A.~W.}\
  \bibnamefont {Harrow}}, \bibinfo {author} {\bibfnamefont {A.}~\bibnamefont
  {Kandala}}, \bibinfo {author} {\bibfnamefont {J.~M.}\ \bibnamefont {Chow}},\
  and\ \bibinfo {author} {\bibfnamefont {J.~M.}\ \bibnamefont {Gambetta}},\
  }\bibfield  {title} {\bibinfo {title} {Supervised learning with
  quantum-enhanced feature spaces},\ }\href
  {https://doi.org/10.1038/s41586-019-0980-2} {\bibfield  {journal} {\bibinfo
  {journal} {Nature}\ }\textbf {\bibinfo {volume} {567}},\ \bibinfo {pages}
  {209} (\bibinfo {year} {2019})}\BibitemShut {NoStop}%
\bibitem [{\citenamefont {Romero}\ \emph {et~al.}(2017)\citenamefont {Romero},
  \citenamefont {Olson},\ and\ \citenamefont
  {Aspuru-Guzik}}]{romero2017quantum}%
  \BibitemOpen
  \bibfield  {author} {\bibinfo {author} {\bibfnamefont {J.}~\bibnamefont
  {Romero}}, \bibinfo {author} {\bibfnamefont {J.~P.}\ \bibnamefont {Olson}},\
  and\ \bibinfo {author} {\bibfnamefont {A.}~\bibnamefont {Aspuru-Guzik}},\
  }\bibfield  {title} {\bibinfo {title} {Quantum autoencoders for efficient
  compression of quantum data},\ }\href
  {https://doi.org/10.1088/2058-9565/aa8072} {\bibfield  {journal} {\bibinfo
  {journal} {Quantum Science and Technology}\ }\textbf {\bibinfo {volume}
  {2}},\ \bibinfo {pages} {045001} (\bibinfo {year} {2017})}\BibitemShut
  {NoStop}%
\bibitem [{\citenamefont {LaRose}\ \emph {et~al.}(2019)\citenamefont {LaRose},
  \citenamefont {Tikku}, \citenamefont {O'Neel-Judy}, \citenamefont {Cincio},\
  and\ \citenamefont {Coles}}]{larose2019variational}%
  \BibitemOpen
  \bibfield  {author} {\bibinfo {author} {\bibfnamefont {R.}~\bibnamefont
  {LaRose}}, \bibinfo {author} {\bibfnamefont {A.}~\bibnamefont {Tikku}},
  \bibinfo {author} {\bibfnamefont {{\'E}.}~\bibnamefont {O'Neel-Judy}},
  \bibinfo {author} {\bibfnamefont {L.}~\bibnamefont {Cincio}},\ and\ \bibinfo
  {author} {\bibfnamefont {P.~J.}\ \bibnamefont {Coles}},\ }\bibfield  {title}
  {\bibinfo {title} {Variational quantum state diagonalization},\ }\href
  {https://www.nature.com/articles/s41534-019-0167-6} {\bibfield  {journal}
  {\bibinfo  {journal} {npj Quantum Information}\ }\textbf {\bibinfo {volume}
  {5}},\ \bibinfo {pages} {1} (\bibinfo {year} {2019})}\BibitemShut {NoStop}%
\bibitem [{\citenamefont {Bravo-Prieto}\ \emph
  {et~al.}(2020{\natexlab{a}})\citenamefont {Bravo-Prieto}, \citenamefont
  {Garc\'{\i}a-Mart\'{\i}n},\ and\ \citenamefont {Latorre}}]{bravo2020quantum}%
  \BibitemOpen
  \bibfield  {author} {\bibinfo {author} {\bibfnamefont {C.}~\bibnamefont
  {Bravo-Prieto}}, \bibinfo {author} {\bibfnamefont {D.}~\bibnamefont
  {Garc\'{\i}a-Mart\'{\i}n}},\ and\ \bibinfo {author} {\bibfnamefont {J.~I.}\
  \bibnamefont {Latorre}},\ }\bibfield  {title} {\bibinfo {title} {Quantum
  singular value decomposer},\ }\href
  {https://doi.org/10.1103/PhysRevA.101.062310} {\bibfield  {journal} {\bibinfo
   {journal} {Phys. Rev. A}\ }\textbf {\bibinfo {volume} {101}},\ \bibinfo
  {pages} {062310} (\bibinfo {year} {2020}{\natexlab{a}})}\BibitemShut
  {NoStop}%
\bibitem [{\citenamefont {Cerezo}\ \emph {et~al.}(2020)\citenamefont {Cerezo},
  \citenamefont {Sharma}, \citenamefont {Arrasmith},\ and\ \citenamefont
  {Coles}}]{cerezo2020variational}%
  \BibitemOpen
  \bibfield  {author} {\bibinfo {author} {\bibfnamefont {M.}~\bibnamefont
  {Cerezo}}, \bibinfo {author} {\bibfnamefont {K.}~\bibnamefont {Sharma}},
  \bibinfo {author} {\bibfnamefont {A.}~\bibnamefont {Arrasmith}},\ and\
  \bibinfo {author} {\bibfnamefont {P.~J.}\ \bibnamefont {Coles}},\ }\bibfield
  {title} {\bibinfo {title} {Variational quantum state eigensolver},\ }\href
  {https://arxiv.org/abs/2004.01372} {\bibfield  {journal} {\bibinfo  {journal}
  {arXiv preprint arXiv:2004.01372}\ } (\bibinfo {year} {2020})}\BibitemShut
  {NoStop}%
\bibitem [{\citenamefont {Yuan}\ \emph {et~al.}(2019)\citenamefont {Yuan},
  \citenamefont {Endo}, \citenamefont {Zhao}, \citenamefont {Li},\ and\
  \citenamefont {Benjamin}}]{yuan2019theory}%
  \BibitemOpen
  \bibfield  {author} {\bibinfo {author} {\bibfnamefont {X.}~\bibnamefont
  {Yuan}}, \bibinfo {author} {\bibfnamefont {S.}~\bibnamefont {Endo}}, \bibinfo
  {author} {\bibfnamefont {Q.}~\bibnamefont {Zhao}}, \bibinfo {author}
  {\bibfnamefont {Y.}~\bibnamefont {Li}},\ and\ \bibinfo {author}
  {\bibfnamefont {S.~C.}\ \bibnamefont {Benjamin}},\ }\bibfield  {title}
  {\bibinfo {title} {Theory of variational quantum simulation},\ }\href
  {https://doi.org/10.22331/q-2019-10-07-191} {\bibfield  {journal} {\bibinfo
  {journal} {Quantum}\ }\textbf {\bibinfo {volume} {3}},\ \bibinfo {pages}
  {191} (\bibinfo {year} {2019})}\BibitemShut {NoStop}%
\bibitem [{\citenamefont {Cirstoiu}\ \emph {et~al.}(2020)\citenamefont
  {Cirstoiu}, \citenamefont {Holmes}, \citenamefont {Iosue}, \citenamefont
  {Cincio}, \citenamefont {Coles},\ and\ \citenamefont
  {Sornborger}}]{cirstoiu2020variational}%
  \BibitemOpen
  \bibfield  {author} {\bibinfo {author} {\bibfnamefont {C.}~\bibnamefont
  {Cirstoiu}}, \bibinfo {author} {\bibfnamefont {Z.}~\bibnamefont {Holmes}},
  \bibinfo {author} {\bibfnamefont {J.}~\bibnamefont {Iosue}}, \bibinfo
  {author} {\bibfnamefont {L.}~\bibnamefont {Cincio}}, \bibinfo {author}
  {\bibfnamefont {P.~J.}\ \bibnamefont {Coles}},\ and\ \bibinfo {author}
  {\bibfnamefont {A.}~\bibnamefont {Sornborger}},\ }\bibfield  {title}
  {\bibinfo {title} {Variational fast forwarding for quantum simulation beyond
  the coherence time},\ }\href
  {https://www.nature.com/articles/s41534-020-00302-0} {\bibfield  {journal}
  {\bibinfo  {journal} {npj Quantum Information}\ }\textbf {\bibinfo {volume}
  {6}},\ \bibinfo {pages} {1} (\bibinfo {year} {2020})}\BibitemShut {NoStop}%
\bibitem [{\citenamefont {Commeau}\ \emph {et~al.}(2020)\citenamefont
  {Commeau}, \citenamefont {Cerezo}, \citenamefont {Holmes}, \citenamefont
  {Cincio}, \citenamefont {Coles},\ and\ \citenamefont
  {Sornborger}}]{commeau2020variational}%
  \BibitemOpen
  \bibfield  {author} {\bibinfo {author} {\bibfnamefont {B.}~\bibnamefont
  {Commeau}}, \bibinfo {author} {\bibfnamefont {M.}~\bibnamefont {Cerezo}},
  \bibinfo {author} {\bibfnamefont {Z.}~\bibnamefont {Holmes}}, \bibinfo
  {author} {\bibfnamefont {L.}~\bibnamefont {Cincio}}, \bibinfo {author}
  {\bibfnamefont {P.~J.}\ \bibnamefont {Coles}},\ and\ \bibinfo {author}
  {\bibfnamefont {A.}~\bibnamefont {Sornborger}},\ }\bibfield  {title}
  {\bibinfo {title} {Variational hamiltonian diagonalization for dynamical
  quantum simulation},\ }\href {https://arxiv.org/abs/2009.02559} {\bibfield
  {journal} {\bibinfo  {journal} {arXiv preprint arXiv:2009.02559}\ } (\bibinfo
  {year} {2020})}\BibitemShut {NoStop}%
\bibitem [{\citenamefont {Gibbs}\ \emph {et~al.}(2021)\citenamefont {Gibbs},
  \citenamefont {Gili}, \citenamefont {Holmes}, \citenamefont {Commeau},
  \citenamefont {Arrasmith}, \citenamefont {Cincio}, \citenamefont {Coles},\
  and\ \citenamefont {Sornborger}}]{gibbs2021long}%
  \BibitemOpen
  \bibfield  {author} {\bibinfo {author} {\bibfnamefont {J.}~\bibnamefont
  {Gibbs}}, \bibinfo {author} {\bibfnamefont {K.}~\bibnamefont {Gili}},
  \bibinfo {author} {\bibfnamefont {Z.}~\bibnamefont {Holmes}}, \bibinfo
  {author} {\bibfnamefont {B.}~\bibnamefont {Commeau}}, \bibinfo {author}
  {\bibfnamefont {A.}~\bibnamefont {Arrasmith}}, \bibinfo {author}
  {\bibfnamefont {L.}~\bibnamefont {Cincio}}, \bibinfo {author} {\bibfnamefont
  {P.~J.}\ \bibnamefont {Coles}},\ and\ \bibinfo {author} {\bibfnamefont
  {A.}~\bibnamefont {Sornborger}},\ }\bibfield  {title} {\bibinfo {title}
  {Long-time simulations with high fidelity on quantum hardware},\ }\href
  {https://arxiv.org/abs/2102.04313} {\bibfield  {journal} {\bibinfo  {journal}
  {arXiv preprint arXiv:2102.04313}\ } (\bibinfo {year} {2021})}\BibitemShut
  {NoStop}%
\bibitem [{\citenamefont {Bharti}\ \emph {et~al.}(2021)\citenamefont {Bharti},
  \citenamefont {Cervera-Lierta}, \citenamefont {Kyaw}, \citenamefont {Haug},
  \citenamefont {Alperin-Lea}, \citenamefont {Anand}, \citenamefont {Degroote},
  \citenamefont {Heimonen}, \citenamefont {Kottmann}, \citenamefont {Menke},
  \citenamefont {Mok}, \citenamefont {Sim}, \citenamefont {Kwek},\ and\
  \citenamefont {Aspuru-Guzik}}]{bharti2021noisy}%
  \BibitemOpen
  \bibfield  {author} {\bibinfo {author} {\bibfnamefont {K.}~\bibnamefont
  {Bharti}}, \bibinfo {author} {\bibfnamefont {A.}~\bibnamefont
  {Cervera-Lierta}}, \bibinfo {author} {\bibfnamefont {T.~H.}\ \bibnamefont
  {Kyaw}}, \bibinfo {author} {\bibfnamefont {T.}~\bibnamefont {Haug}}, \bibinfo
  {author} {\bibfnamefont {S.}~\bibnamefont {Alperin-Lea}}, \bibinfo {author}
  {\bibfnamefont {A.}~\bibnamefont {Anand}}, \bibinfo {author} {\bibfnamefont
  {M.}~\bibnamefont {Degroote}}, \bibinfo {author} {\bibfnamefont
  {H.}~\bibnamefont {Heimonen}}, \bibinfo {author} {\bibfnamefont {J.~S.}\
  \bibnamefont {Kottmann}}, \bibinfo {author} {\bibfnamefont {T.}~\bibnamefont
  {Menke}}, \bibinfo {author} {\bibfnamefont {W.-K.}\ \bibnamefont {Mok}},
  \bibinfo {author} {\bibfnamefont {S.}~\bibnamefont {Sim}}, \bibinfo {author}
  {\bibfnamefont {L.-C.}\ \bibnamefont {Kwek}},\ and\ \bibinfo {author}
  {\bibfnamefont {A.}~\bibnamefont {Aspuru-Guzik}},\ }\bibfield  {title}
  {\bibinfo {title} {Noisy intermediate-scale quantum (nisq) algorithms},\
  }\href {https://arxiv.org/abs/2101.08448} {\bibfield  {journal} {\bibinfo
  {journal} {arXiv preprint arXiv:2101.08448}\ } (\bibinfo {year}
  {2021})}\BibitemShut {NoStop}%
\bibitem [{\citenamefont {Efthymiou}\ \emph {et~al.}(2020)\citenamefont
  {Efthymiou}, \citenamefont {Ramos-Calderer}, \citenamefont {Bravo-Prieto},
  \citenamefont {Pérez-Salinas}, \citenamefont {García-Martín},
  \citenamefont {Garcia-Saez}, \citenamefont {Latorre},\ and\ \citenamefont
  {Carrazza}}]{efthymiou2020qibo}%
  \BibitemOpen
  \bibfield  {author} {\bibinfo {author} {\bibfnamefont {S.}~\bibnamefont
  {Efthymiou}}, \bibinfo {author} {\bibfnamefont {S.}~\bibnamefont
  {Ramos-Calderer}}, \bibinfo {author} {\bibfnamefont {C.}~\bibnamefont
  {Bravo-Prieto}}, \bibinfo {author} {\bibfnamefont {A.}~\bibnamefont
  {Pérez-Salinas}}, \bibinfo {author} {\bibfnamefont {D.}~\bibnamefont
  {García-Martín}}, \bibinfo {author} {\bibfnamefont {A.}~\bibnamefont
  {Garcia-Saez}}, \bibinfo {author} {\bibfnamefont {J.~I.}\ \bibnamefont
  {Latorre}},\ and\ \bibinfo {author} {\bibfnamefont {S.}~\bibnamefont
  {Carrazza}},\ }\bibfield  {title} {\bibinfo {title} {Qibo: a framework for
  quantum simulation with hardware acceleration},\ }\href
  {https://arxiv.org/abs/2009.01845} {\bibfield  {journal} {\bibinfo  {journal}
  {arXiv preprint arXiv:2009.01845}\ } (\bibinfo {year} {2020})}\BibitemShut
  {NoStop}%
\bibitem [{\citenamefont {Efthymiou}\ \emph {et~al.}(2021)\citenamefont
  {Efthymiou}, \citenamefont {Carrazza}, \citenamefont {Ramos}, \citenamefont
  {bpcarlos}, \citenamefont {AdrianPerezSalinas}, \citenamefont
  {García-Martín}, \citenamefont {Paul}, \citenamefont {Serrano},\ and\
  \citenamefont {atomicprinter}}]{efthymiou2021qibo_zenodo}%
  \BibitemOpen
  \bibfield  {author} {\bibinfo {author} {\bibfnamefont {S.}~\bibnamefont
  {Efthymiou}}, \bibinfo {author} {\bibfnamefont {S.}~\bibnamefont {Carrazza}},
  \bibinfo {author} {\bibfnamefont {S.}~\bibnamefont {Ramos}}, \bibinfo
  {author} {\bibnamefont {bpcarlos}}, \bibinfo {author} {\bibnamefont
  {AdrianPerezSalinas}}, \bibinfo {author} {\bibfnamefont {D.}~\bibnamefont
  {García-Martín}}, \bibinfo {author} {\bibnamefont {Paul}}, \bibinfo
  {author} {\bibfnamefont {J.}~\bibnamefont {Serrano}},\ and\ \bibinfo {author}
  {\bibnamefont {atomicprinter}},\ }\href
  {https://doi.org/10.5281/zenodo.5088103} {\bibinfo {title} {qiboteam/qibo:
  Qibo 0.1.6-rc1}} (\bibinfo {year} {2021})\BibitemShut {NoStop}%
\bibitem [{\citenamefont {Peruzzo}\ \emph {et~al.}(2014)\citenamefont
  {Peruzzo}, \citenamefont {McClean}, \citenamefont {Shadbolt}, \citenamefont
  {Yung}, \citenamefont {Zhou}, \citenamefont {Love}, \citenamefont
  {Aspuru-Guzik},\ and\ \citenamefont {O’brien}}]{peruzzo2014variational}%
  \BibitemOpen
  \bibfield  {author} {\bibinfo {author} {\bibfnamefont {A.}~\bibnamefont
  {Peruzzo}}, \bibinfo {author} {\bibfnamefont {J.}~\bibnamefont {McClean}},
  \bibinfo {author} {\bibfnamefont {P.}~\bibnamefont {Shadbolt}}, \bibinfo
  {author} {\bibfnamefont {M.-H.}\ \bibnamefont {Yung}}, \bibinfo {author}
  {\bibfnamefont {X.-Q.}\ \bibnamefont {Zhou}}, \bibinfo {author}
  {\bibfnamefont {P.~J.}\ \bibnamefont {Love}}, \bibinfo {author}
  {\bibfnamefont {A.}~\bibnamefont {Aspuru-Guzik}},\ and\ \bibinfo {author}
  {\bibfnamefont {J.~L.}\ \bibnamefont {O’brien}},\ }\bibfield  {title}
  {\bibinfo {title} {A variational eigenvalue solver on a photonic quantum
  processor},\ }\href {https://doi.org/doi.org/10.1038/ncomms5213} {\bibfield
  {journal} {\bibinfo  {journal} {Nature communications}\ }\textbf {\bibinfo
  {volume} {5}},\ \bibinfo {pages} {1} (\bibinfo {year} {2014})}\BibitemShut
  {NoStop}%
\bibitem [{\citenamefont {Bravo-Prieto}\ \emph
  {et~al.}(2020{\natexlab{b}})\citenamefont {Bravo-Prieto}, \citenamefont
  {Lumbreras-Zarapico}, \citenamefont {Tagliacozzo},\ and\ \citenamefont
  {Latorre}}]{bravo2020scaling}%
  \BibitemOpen
  \bibfield  {author} {\bibinfo {author} {\bibfnamefont {C.}~\bibnamefont
  {Bravo-Prieto}}, \bibinfo {author} {\bibfnamefont {J.}~\bibnamefont
  {Lumbreras-Zarapico}}, \bibinfo {author} {\bibfnamefont {L.}~\bibnamefont
  {Tagliacozzo}},\ and\ \bibinfo {author} {\bibfnamefont {J.~I.}\ \bibnamefont
  {Latorre}},\ }\bibfield  {title} {\bibinfo {title} {Scaling of variational
  quantum circuit depth for condensed matter systems},\ }\href
  {https://doi.org/10.22331/q-2020-05-28-272} {\bibfield  {journal} {\bibinfo
  {journal} {{Quantum}}\ }\textbf {\bibinfo {volume} {4}},\ \bibinfo {pages}
  {272} (\bibinfo {year} {2020}{\natexlab{b}})}\BibitemShut {NoStop}%
\bibitem [{\citenamefont {Consiglio}\ \emph {et~al.}(2021)\citenamefont
  {Consiglio}, \citenamefont {Chetcuti}, \citenamefont {Bravo-Prieto},
  \citenamefont {Ramos-Calderer}, \citenamefont {Minguzzi}, \citenamefont
  {Latorre}, \citenamefont {Amico},\ and\ \citenamefont
  {Apollaro}}]{consiglio2021variational}%
  \BibitemOpen
  \bibfield  {author} {\bibinfo {author} {\bibfnamefont {M.}~\bibnamefont
  {Consiglio}}, \bibinfo {author} {\bibfnamefont {W.~J.}\ \bibnamefont
  {Chetcuti}}, \bibinfo {author} {\bibfnamefont {C.}~\bibnamefont
  {Bravo-Prieto}}, \bibinfo {author} {\bibfnamefont {S.}~\bibnamefont
  {Ramos-Calderer}}, \bibinfo {author} {\bibfnamefont {A.}~\bibnamefont
  {Minguzzi}}, \bibinfo {author} {\bibfnamefont {J.~I.}\ \bibnamefont
  {Latorre}}, \bibinfo {author} {\bibfnamefont {L.}~\bibnamefont {Amico}},\
  and\ \bibinfo {author} {\bibfnamefont {T.~J.~G.}\ \bibnamefont {Apollaro}},\
  }\bibfield  {title} {\bibinfo {title} {Variational quantum eigensolver for
  su($n$) fermions},\ }\href {https://arxiv.org/abs/2106.15552} {\bibfield
  {journal} {\bibinfo  {journal} {arXiv preprint arXiv:2106.15552}\ } (\bibinfo
  {year} {2021})}\BibitemShut {NoStop}%
\bibitem [{SI-()}]{SI-overparam}%
  \BibitemOpen
  \href@noop {} {}\bibinfo {note} {Equations~\eqref{eqn:HVA_DLAdim} are
  actually upper bounds for $\dim(\liea_\SC)$, however, the bounds get quickly
  saturated as $n$ increases.}\BibitemShut {Stop}%
\bibitem [{\citenamefont {Chong}\ \emph {et~al.}(2017)\citenamefont {Chong},
  \citenamefont {Franklin},\ and\ \citenamefont
  {Martonosi}}]{chong2017programming}%
  \BibitemOpen
  \bibfield  {author} {\bibinfo {author} {\bibfnamefont {F.~T.}\ \bibnamefont
  {Chong}}, \bibinfo {author} {\bibfnamefont {D.}~\bibnamefont {Franklin}},\
  and\ \bibinfo {author} {\bibfnamefont {M.}~\bibnamefont {Martonosi}},\
  }\bibfield  {title} {\bibinfo {title} {Programming languages and compiler
  design for realistic quantum hardware},\ }\href
  {https://doi.org/10.1038/nature23459} {\bibfield  {journal} {\bibinfo
  {journal} {Nature}\ }\textbf {\bibinfo {volume} {549}},\ \bibinfo {pages}
  {180} (\bibinfo {year} {2017})}\BibitemShut {NoStop}%
\bibitem [{\citenamefont {H{\"a}ner}\ \emph {et~al.}(2018)\citenamefont
  {H{\"a}ner}, \citenamefont {Steiger}, \citenamefont {Svore},\ and\
  \citenamefont {Troyer}}]{haner2018software}%
  \BibitemOpen
  \bibfield  {author} {\bibinfo {author} {\bibfnamefont {T.}~\bibnamefont
  {H{\"a}ner}}, \bibinfo {author} {\bibfnamefont {D.~S.}\ \bibnamefont
  {Steiger}}, \bibinfo {author} {\bibfnamefont {K.}~\bibnamefont {Svore}},\
  and\ \bibinfo {author} {\bibfnamefont {M.}~\bibnamefont {Troyer}},\
  }\bibfield  {title} {\bibinfo {title} {A software methodology for compiling
  quantum programs},\ }\href {https://doi.org/10.1088/2058-9565/aaa5cc}
  {\bibfield  {journal} {\bibinfo  {journal} {Quantum Science and Technology}\
  }\textbf {\bibinfo {volume} {3}},\ \bibinfo {pages} {020501} (\bibinfo {year}
  {2018})}\BibitemShut {NoStop}%
\bibitem [{\citenamefont {Venturelli}\ \emph {et~al.}(2018)\citenamefont
  {Venturelli}, \citenamefont {Do}, \citenamefont {Rieffel},\ and\
  \citenamefont {Frank}}]{venturelli2018compiling}%
  \BibitemOpen
  \bibfield  {author} {\bibinfo {author} {\bibfnamefont {D.}~\bibnamefont
  {Venturelli}}, \bibinfo {author} {\bibfnamefont {M.}~\bibnamefont {Do}},
  \bibinfo {author} {\bibfnamefont {E.}~\bibnamefont {Rieffel}},\ and\ \bibinfo
  {author} {\bibfnamefont {J.}~\bibnamefont {Frank}},\ }\bibfield  {title}
  {\bibinfo {title} {Compiling quantum circuits to realistic hardware
  architectures using temporal planners},\ }\href
  {https://doi.org/10.1088/2058-9565/aaa331} {\bibfield  {journal} {\bibinfo
  {journal} {Quantum Science and Technology}\ }\textbf {\bibinfo {volume}
  {3}},\ \bibinfo {pages} {025004} (\bibinfo {year} {2018})}\BibitemShut
  {NoStop}%
\bibitem [{\citenamefont {Khatri}\ \emph {et~al.}(2019)\citenamefont {Khatri},
  \citenamefont {LaRose}, \citenamefont {Poremba}, \citenamefont {Cincio},
  \citenamefont {Sornborger},\ and\ \citenamefont {Coles}}]{khatri2019quantum}%
  \BibitemOpen
  \bibfield  {author} {\bibinfo {author} {\bibfnamefont {S.}~\bibnamefont
  {Khatri}}, \bibinfo {author} {\bibfnamefont {R.}~\bibnamefont {LaRose}},
  \bibinfo {author} {\bibfnamefont {A.}~\bibnamefont {Poremba}}, \bibinfo
  {author} {\bibfnamefont {L.}~\bibnamefont {Cincio}}, \bibinfo {author}
  {\bibfnamefont {A.~T.}\ \bibnamefont {Sornborger}},\ and\ \bibinfo {author}
  {\bibfnamefont {P.~J.}\ \bibnamefont {Coles}},\ }\bibfield  {title} {\bibinfo
  {title} {Quantum-assisted quantum compiling},\ }\href
  {https://quantum-journal.org/papers/q-2019-05-13-140/} {\bibfield  {journal}
  {\bibinfo  {journal} {Quantum}\ }\textbf {\bibinfo {volume} {3}},\ \bibinfo
  {pages} {140} (\bibinfo {year} {2019})}\BibitemShut {NoStop}%
\bibitem [{\citenamefont {Sharma}\ \emph
  {et~al.}(2020{\natexlab{b}})\citenamefont {Sharma}, \citenamefont {Khatri},
  \citenamefont {Cerezo},\ and\ \citenamefont {Coles}}]{sharma2019noise}%
  \BibitemOpen
  \bibfield  {author} {\bibinfo {author} {\bibfnamefont {K.}~\bibnamefont
  {Sharma}}, \bibinfo {author} {\bibfnamefont {S.}~\bibnamefont {Khatri}},
  \bibinfo {author} {\bibfnamefont {M.}~\bibnamefont {Cerezo}},\ and\ \bibinfo
  {author} {\bibfnamefont {P.~J.}\ \bibnamefont {Coles}},\ }\bibfield  {title}
  {\bibinfo {title} {Noise resilience of variational quantum compiling},\
  }\href {https://iopscience.iop.org/article/10.1088/1367-2630/ab784c}
  {\bibfield  {journal} {\bibinfo  {journal} {New Journal of Physics}\ }\textbf
  {\bibinfo {volume} {22}},\ \bibinfo {pages} {043006} (\bibinfo {year}
  {2020}{\natexlab{b}})}\BibitemShut {NoStop}%
\bibitem [{\citenamefont {Ma}\ \emph {et~al.}(2020)\citenamefont {Ma},
  \citenamefont {Huang}, \citenamefont {Chen}, \citenamefont {Dong},
  \citenamefont {Wang}, \citenamefont {Wu},\ and\ \citenamefont
  {Xiang}}]{ma2020compression}%
  \BibitemOpen
  \bibfield  {author} {\bibinfo {author} {\bibfnamefont {H.}~\bibnamefont
  {Ma}}, \bibinfo {author} {\bibfnamefont {C.-J.}\ \bibnamefont {Huang}},
  \bibinfo {author} {\bibfnamefont {C.}~\bibnamefont {Chen}}, \bibinfo {author}
  {\bibfnamefont {D.}~\bibnamefont {Dong}}, \bibinfo {author} {\bibfnamefont
  {Y.}~\bibnamefont {Wang}}, \bibinfo {author} {\bibfnamefont {R.-B.}\
  \bibnamefont {Wu}},\ and\ \bibinfo {author} {\bibfnamefont {G.-Y.}\
  \bibnamefont {Xiang}},\ }\bibfield  {title} {\bibinfo {title} {On compression
  rate of quantum autoencoders: Control design, numerical and experimental
  realization},\ }\href {https://arxiv.org/abs/2005.11149} {\bibfield
  {journal} {\bibinfo  {journal} {arXiv preprint arXiv:2005.11149}\ } (\bibinfo
  {year} {2020})}\BibitemShut {NoStop}%
\bibitem [{\citenamefont {Schatzki}\ \emph {et~al.}(2021)\citenamefont
  {Schatzki}, \citenamefont {Arrasmith}, \citenamefont {Coles},\ and\
  \citenamefont {Cerezo}}]{schatzki2021entangled}%
  \BibitemOpen
  \bibfield  {author} {\bibinfo {author} {\bibfnamefont {L.}~\bibnamefont
  {Schatzki}}, \bibinfo {author} {\bibfnamefont {A.}~\bibnamefont {Arrasmith}},
  \bibinfo {author} {\bibfnamefont {P.~J.}\ \bibnamefont {Coles}},\ and\
  \bibinfo {author} {\bibfnamefont {M.}~\bibnamefont {Cerezo}},\ }\bibfield
  {title} {\bibinfo {title} {Entangled datasets for quantum machine learning},\
  }\href {https://arxiv.org/abs/2109.03400} {\bibfield  {journal} {\bibinfo
  {journal} {arXiv preprint arXiv:2109.03400}\ } (\bibinfo {year}
  {2021})}\BibitemShut {NoStop}%
\bibitem [{\citenamefont {Chan}\ and\ \citenamefont
  {Kwong}(1985)}]{chan1985hermitian}%
  \BibitemOpen
  \bibfield  {author} {\bibinfo {author} {\bibfnamefont {N.}~\bibnamefont
  {Chan}}\ and\ \bibinfo {author} {\bibfnamefont {M.~K.}\ \bibnamefont
  {Kwong}},\ }\bibfield  {title} {\bibinfo {title} {Hermitian matrix
  inequalities and a conjecture},\ }\href {https://doi.org/10.2307/2323157}
  {\bibfield  {journal} {\bibinfo  {journal} {The American Mathematical
  Monthly}\ }\textbf {\bibinfo {volume} {92}} (\bibinfo {year}
  {1985})}\BibitemShut {NoStop}%
\bibitem [{\citenamefont {Glaser}\ \emph {et~al.}(2015)\citenamefont {Glaser},
  \citenamefont {Boscain}, \citenamefont {Calarco}, \citenamefont {Koch},
  \citenamefont {K{\"o}ckenberger}, \citenamefont {Kosloff}, \citenamefont
  {Kuprov}, \citenamefont {Luy}, \citenamefont {Schirmer}, \citenamefont
  {Schulte-Herbr{\"u}ggen} \emph {et~al.}}]{glaser2015training}%
  \BibitemOpen
  \bibfield  {author} {\bibinfo {author} {\bibfnamefont {S.~J.}\ \bibnamefont
  {Glaser}}, \bibinfo {author} {\bibfnamefont {U.}~\bibnamefont {Boscain}},
  \bibinfo {author} {\bibfnamefont {T.}~\bibnamefont {Calarco}}, \bibinfo
  {author} {\bibfnamefont {C.~P.}\ \bibnamefont {Koch}}, \bibinfo {author}
  {\bibfnamefont {W.}~\bibnamefont {K{\"o}ckenberger}}, \bibinfo {author}
  {\bibfnamefont {R.}~\bibnamefont {Kosloff}}, \bibinfo {author} {\bibfnamefont
  {I.}~\bibnamefont {Kuprov}}, \bibinfo {author} {\bibfnamefont
  {B.}~\bibnamefont {Luy}}, \bibinfo {author} {\bibfnamefont {S.}~\bibnamefont
  {Schirmer}}, \bibinfo {author} {\bibfnamefont {T.}~\bibnamefont
  {Schulte-Herbr{\"u}ggen}}, \emph {et~al.},\ }\bibfield  {title} {\bibinfo
  {title} {Training schr{\"o}dinger’s cat: quantum optimal control},\ }\href
  {https://doi.org/https://doi.org/10.1140/epjd/e2015-60464-1} {\bibfield
  {journal} {\bibinfo  {journal} {The European Physical Journal D}\ }\textbf
  {\bibinfo {volume} {69}},\ \bibinfo {pages} {1} (\bibinfo {year}
  {2015})}\BibitemShut {NoStop}%
\bibitem [{\citenamefont {Ac{\'\i}n}\ \emph {et~al.}(2018)\citenamefont
  {Ac{\'\i}n}, \citenamefont {Bloch}, \citenamefont {Buhrman}, \citenamefont
  {Calarco}, \citenamefont {Eichler}, \citenamefont {Eisert}, \citenamefont
  {Esteve}, \citenamefont {Gisin}, \citenamefont {Glaser}, \citenamefont
  {Jelezko} \emph {et~al.}}]{acin2018quantum}%
  \BibitemOpen
  \bibfield  {author} {\bibinfo {author} {\bibfnamefont {A.}~\bibnamefont
  {Ac{\'\i}n}}, \bibinfo {author} {\bibfnamefont {I.}~\bibnamefont {Bloch}},
  \bibinfo {author} {\bibfnamefont {H.}~\bibnamefont {Buhrman}}, \bibinfo
  {author} {\bibfnamefont {T.}~\bibnamefont {Calarco}}, \bibinfo {author}
  {\bibfnamefont {C.}~\bibnamefont {Eichler}}, \bibinfo {author} {\bibfnamefont
  {J.}~\bibnamefont {Eisert}}, \bibinfo {author} {\bibfnamefont
  {D.}~\bibnamefont {Esteve}}, \bibinfo {author} {\bibfnamefont
  {N.}~\bibnamefont {Gisin}}, \bibinfo {author} {\bibfnamefont {S.~J.}\
  \bibnamefont {Glaser}}, \bibinfo {author} {\bibfnamefont {F.}~\bibnamefont
  {Jelezko}}, \emph {et~al.},\ }\bibfield  {title} {\bibinfo {title} {The
  quantum technologies roadmap: a european community view},\ }\href
  {https://doi.org/https://doi.org/10.1088/1367-2630/aad1ea} {\bibfield
  {journal} {\bibinfo  {journal} {New Journal of Physics}\ }\textbf {\bibinfo
  {volume} {20}},\ \bibinfo {pages} {080201} (\bibinfo {year}
  {2018})}\BibitemShut {NoStop}%
\bibitem [{\citenamefont {Yang}\ \emph {et~al.}(2017)\citenamefont {Yang},
  \citenamefont {Rahmani}, \citenamefont {Shabani}, \citenamefont {Neven},\
  and\ \citenamefont {Chamon}}]{yang2017optimizing}%
  \BibitemOpen
  \bibfield  {author} {\bibinfo {author} {\bibfnamefont {Z.-C.}\ \bibnamefont
  {Yang}}, \bibinfo {author} {\bibfnamefont {A.}~\bibnamefont {Rahmani}},
  \bibinfo {author} {\bibfnamefont {A.}~\bibnamefont {Shabani}}, \bibinfo
  {author} {\bibfnamefont {H.}~\bibnamefont {Neven}},\ and\ \bibinfo {author}
  {\bibfnamefont {C.}~\bibnamefont {Chamon}},\ }\bibfield  {title} {\bibinfo
  {title} {Optimizing variational quantum algorithms using pontryagin’s
  minimum principle},\ }\href {https://doi.org/10.1103/PhysRevX.7.021027}
  {\bibfield  {journal} {\bibinfo  {journal} {Physical Review X}\ }\textbf
  {\bibinfo {volume} {7}},\ \bibinfo {pages} {021027} (\bibinfo {year}
  {2017})}\BibitemShut {NoStop}%
\bibitem [{\citenamefont {Lu}\ \emph {et~al.}(2017)\citenamefont {Lu},
  \citenamefont {Li}, \citenamefont {Li}, \citenamefont {Katiyar},
  \citenamefont {Park}, \citenamefont {Feng}, \citenamefont {Xin},
  \citenamefont {Li}, \citenamefont {Long}, \citenamefont {Brodutch} \emph
  {et~al.}}]{lu2017enhancing}%
  \BibitemOpen
  \bibfield  {author} {\bibinfo {author} {\bibfnamefont {D.}~\bibnamefont
  {Lu}}, \bibinfo {author} {\bibfnamefont {K.}~\bibnamefont {Li}}, \bibinfo
  {author} {\bibfnamefont {J.}~\bibnamefont {Li}}, \bibinfo {author}
  {\bibfnamefont {H.}~\bibnamefont {Katiyar}}, \bibinfo {author} {\bibfnamefont
  {A.~J.}\ \bibnamefont {Park}}, \bibinfo {author} {\bibfnamefont
  {G.}~\bibnamefont {Feng}}, \bibinfo {author} {\bibfnamefont {T.}~\bibnamefont
  {Xin}}, \bibinfo {author} {\bibfnamefont {H.}~\bibnamefont {Li}}, \bibinfo
  {author} {\bibfnamefont {G.}~\bibnamefont {Long}}, \bibinfo {author}
  {\bibfnamefont {A.}~\bibnamefont {Brodutch}}, \emph {et~al.},\ }\bibfield
  {title} {\bibinfo {title} {Enhancing quantum control by bootstrapping a
  quantum processor of 12 qubits},\ }\href
  {https://doi.org/10.1038/s41534-017-0045-z} {\bibfield  {journal} {\bibinfo
  {journal} {npj Quantum Information}\ }\textbf {\bibinfo {volume} {3}},\
  \bibinfo {pages} {1} (\bibinfo {year} {2017})}\BibitemShut {NoStop}%
\bibitem [{\citenamefont {Rembold}\ \emph {et~al.}(2020)\citenamefont
  {Rembold}, \citenamefont {Oshnik}, \citenamefont {M{\"u}ller}, \citenamefont
  {Montangero}, \citenamefont {Calarco},\ and\ \citenamefont
  {Neu}}]{rembold2020introduction}%
  \BibitemOpen
  \bibfield  {author} {\bibinfo {author} {\bibfnamefont {P.}~\bibnamefont
  {Rembold}}, \bibinfo {author} {\bibfnamefont {N.}~\bibnamefont {Oshnik}},
  \bibinfo {author} {\bibfnamefont {M.~M.}\ \bibnamefont {M{\"u}ller}},
  \bibinfo {author} {\bibfnamefont {S.}~\bibnamefont {Montangero}}, \bibinfo
  {author} {\bibfnamefont {T.}~\bibnamefont {Calarco}},\ and\ \bibinfo {author}
  {\bibfnamefont {E.}~\bibnamefont {Neu}},\ }\bibfield  {title} {\bibinfo
  {title} {Introduction to quantum optimal control for quantum sensing with
  nitrogen-vacancy centers in diamond},\ }\href
  {https://doi.org/https://doi.org/10.1116/5.0006785} {\bibfield  {journal}
  {\bibinfo  {journal} {AVS Quantum Science}\ }\textbf {\bibinfo {volume}
  {2}},\ \bibinfo {pages} {024701} (\bibinfo {year} {2020})}\BibitemShut
  {NoStop}%
\bibitem [{\citenamefont {Peterson}\ \emph {et~al.}(2020)\citenamefont
  {Peterson}, \citenamefont {Katiyar},\ and\ \citenamefont
  {Laflamme}}]{peterson2020fast}%
  \BibitemOpen
  \bibfield  {author} {\bibinfo {author} {\bibfnamefont {J.~P.}\ \bibnamefont
  {Peterson}}, \bibinfo {author} {\bibfnamefont {H.}~\bibnamefont {Katiyar}},\
  and\ \bibinfo {author} {\bibfnamefont {R.}~\bibnamefont {Laflamme}},\
  }\bibfield  {title} {\bibinfo {title} {Fast simulation of magnetic field
  gradients for optimization of pulse sequences},\ }\href
  {https://arxiv.org/abs/2006.10133} {\bibfield  {journal} {\bibinfo  {journal}
  {arXiv preprint arXiv:2006.10133}\ } (\bibinfo {year} {2020})}\BibitemShut
  {NoStop}%
\bibitem [{\citenamefont {Bluvstein}\ \emph {et~al.}(2021)\citenamefont
  {Bluvstein}, \citenamefont {Omran}, \citenamefont {Levine}, \citenamefont
  {Keesling}, \citenamefont {Semeghini}, \citenamefont {Ebadi}, \citenamefont
  {Wang}, \citenamefont {Michailidis}, \citenamefont {Maskara}, \citenamefont
  {Ho} \emph {et~al.}}]{bluvstein2021controlling}%
  \BibitemOpen
  \bibfield  {author} {\bibinfo {author} {\bibfnamefont {D.}~\bibnamefont
  {Bluvstein}}, \bibinfo {author} {\bibfnamefont {A.}~\bibnamefont {Omran}},
  \bibinfo {author} {\bibfnamefont {H.}~\bibnamefont {Levine}}, \bibinfo
  {author} {\bibfnamefont {A.}~\bibnamefont {Keesling}}, \bibinfo {author}
  {\bibfnamefont {G.}~\bibnamefont {Semeghini}}, \bibinfo {author}
  {\bibfnamefont {S.}~\bibnamefont {Ebadi}}, \bibinfo {author} {\bibfnamefont
  {T.~T.}\ \bibnamefont {Wang}}, \bibinfo {author} {\bibfnamefont {A.~A.}\
  \bibnamefont {Michailidis}}, \bibinfo {author} {\bibfnamefont
  {N.}~\bibnamefont {Maskara}}, \bibinfo {author} {\bibfnamefont {W.~W.}\
  \bibnamefont {Ho}}, \emph {et~al.},\ }\bibfield  {title} {\bibinfo {title}
  {Controlling quantum many-body dynamics in driven rydberg atom arrays},\
  }\href {https://doi.org/DOI: 10.1126/science.abg2530} {\bibfield  {journal}
  {\bibinfo  {journal} {Science}\ }\textbf {\bibinfo {volume} {371}},\ \bibinfo
  {pages} {1355} (\bibinfo {year} {2021})}\BibitemShut {NoStop}%
\bibitem [{\citenamefont {Ebadi}\ \emph {et~al.}(2021)\citenamefont {Ebadi},
  \citenamefont {Wang}, \citenamefont {Levine}, \citenamefont {Keesling},
  \citenamefont {Semeghini}, \citenamefont {Omran}, \citenamefont {Bluvstein},
  \citenamefont {Samajdar}, \citenamefont {Pichler}, \citenamefont {Ho} \emph
  {et~al.}}]{ebadi2021quantum}%
  \BibitemOpen
  \bibfield  {author} {\bibinfo {author} {\bibfnamefont {S.}~\bibnamefont
  {Ebadi}}, \bibinfo {author} {\bibfnamefont {T.~T.}\ \bibnamefont {Wang}},
  \bibinfo {author} {\bibfnamefont {H.}~\bibnamefont {Levine}}, \bibinfo
  {author} {\bibfnamefont {A.}~\bibnamefont {Keesling}}, \bibinfo {author}
  {\bibfnamefont {G.}~\bibnamefont {Semeghini}}, \bibinfo {author}
  {\bibfnamefont {A.}~\bibnamefont {Omran}}, \bibinfo {author} {\bibfnamefont
  {D.}~\bibnamefont {Bluvstein}}, \bibinfo {author} {\bibfnamefont
  {R.}~\bibnamefont {Samajdar}}, \bibinfo {author} {\bibfnamefont
  {H.}~\bibnamefont {Pichler}}, \bibinfo {author} {\bibfnamefont {W.~W.}\
  \bibnamefont {Ho}}, \emph {et~al.},\ }\bibfield  {title} {\bibinfo {title}
  {Quantum phases of matter on a 256-atom programmable quantum simulator},\
  }\href {https://doi.org/https://doi.org/10.1038/s41586-021-03582-4}
  {\bibfield  {journal} {\bibinfo  {journal} {Nature}\ }\textbf {\bibinfo
  {volume} {595}},\ \bibinfo {pages} {227} (\bibinfo {year}
  {2021})}\BibitemShut {NoStop}%
\bibitem [{\citenamefont {Magann}\ \emph
  {et~al.}(2021{\natexlab{a}})\citenamefont {Magann}, \citenamefont {Rudinger},
  \citenamefont {Grace},\ and\ \citenamefont {Sarovar}}]{magann2021feedback}%
  \BibitemOpen
  \bibfield  {author} {\bibinfo {author} {\bibfnamefont {A.~B.}\ \bibnamefont
  {Magann}}, \bibinfo {author} {\bibfnamefont {K.~M.}\ \bibnamefont
  {Rudinger}}, \bibinfo {author} {\bibfnamefont {M.~D.}\ \bibnamefont
  {Grace}},\ and\ \bibinfo {author} {\bibfnamefont {M.}~\bibnamefont
  {Sarovar}},\ }\bibfield  {title} {\bibinfo {title} {Feedback-based quantum
  optimization},\ }\href {https://arxiv.org/abs/2103.08619} {\bibfield
  {journal} {\bibinfo  {journal} {arXiv preprint arXiv:2103.08619}\ } (\bibinfo
  {year} {2021}{\natexlab{a}})}\BibitemShut {NoStop}%
\bibitem [{\citenamefont {Larocca}\ and\ \citenamefont
  {Wisniacki}(2021)}]{larocca2021krylov}%
  \BibitemOpen
  \bibfield  {author} {\bibinfo {author} {\bibfnamefont {M.}~\bibnamefont
  {Larocca}}\ and\ \bibinfo {author} {\bibfnamefont {D.}~\bibnamefont
  {Wisniacki}},\ }\bibfield  {title} {\bibinfo {title} {Krylov-subspace
  approach for the efficient control of quantum many-body dynamics},\ }\href
  {https://doi.org/https://doi.org/10.1103/PhysRevA.103.023107} {\bibfield
  {journal} {\bibinfo  {journal} {Physical Review A}\ }\textbf {\bibinfo
  {volume} {103}},\ \bibinfo {pages} {023107} (\bibinfo {year}
  {2021})}\BibitemShut {NoStop}%
\bibitem [{\citenamefont {Brady}\ \emph {et~al.}(2021)\citenamefont {Brady},
  \citenamefont {Baldwin}, \citenamefont {Bapat}, \citenamefont {Kharkov},\
  and\ \citenamefont {Gorshkov}}]{brady2021optimal}%
  \BibitemOpen
  \bibfield  {author} {\bibinfo {author} {\bibfnamefont {L.~T.}\ \bibnamefont
  {Brady}}, \bibinfo {author} {\bibfnamefont {C.~L.}\ \bibnamefont {Baldwin}},
  \bibinfo {author} {\bibfnamefont {A.}~\bibnamefont {Bapat}}, \bibinfo
  {author} {\bibfnamefont {Y.}~\bibnamefont {Kharkov}},\ and\ \bibinfo {author}
  {\bibfnamefont {A.~V.}\ \bibnamefont {Gorshkov}},\ }\bibfield  {title}
  {\bibinfo {title} {Optimal protocols in quantum annealing and quantum
  approximate optimization algorithm problems},\ }\href
  {https://doi.org/https://doi.org/10.1103/PhysRevLett.126.070505} {\bibfield
  {journal} {\bibinfo  {journal} {Physical Review Letters}\ }\textbf {\bibinfo
  {volume} {126}},\ \bibinfo {pages} {070505} (\bibinfo {year}
  {2021})}\BibitemShut {NoStop}%
\bibitem [{\citenamefont {Wittler}\ \emph {et~al.}(2021)\citenamefont
  {Wittler}, \citenamefont {Roy}, \citenamefont {Pack}, \citenamefont
  {Werninghaus}, \citenamefont {Roy}, \citenamefont {Egger}, \citenamefont
  {Filipp}, \citenamefont {Wilhelm},\ and\ \citenamefont
  {Machnes}}]{wittler2021integrated}%
  \BibitemOpen
  \bibfield  {author} {\bibinfo {author} {\bibfnamefont {N.}~\bibnamefont
  {Wittler}}, \bibinfo {author} {\bibfnamefont {F.}~\bibnamefont {Roy}},
  \bibinfo {author} {\bibfnamefont {K.}~\bibnamefont {Pack}}, \bibinfo {author}
  {\bibfnamefont {M.}~\bibnamefont {Werninghaus}}, \bibinfo {author}
  {\bibfnamefont {A.~S.}\ \bibnamefont {Roy}}, \bibinfo {author} {\bibfnamefont
  {D.~J.}\ \bibnamefont {Egger}}, \bibinfo {author} {\bibfnamefont
  {S.}~\bibnamefont {Filipp}}, \bibinfo {author} {\bibfnamefont {F.~K.}\
  \bibnamefont {Wilhelm}},\ and\ \bibinfo {author} {\bibfnamefont
  {S.}~\bibnamefont {Machnes}},\ }\bibfield  {title} {\bibinfo {title}
  {Integrated tool set for control, calibration, and characterization of
  quantum devices applied to superconducting qubits},\ }\href
  {https://doi.org/https://doi.org/10.1103/PhysRevApplied.15.034080} {\bibfield
   {journal} {\bibinfo  {journal} {Physical Review Applied}\ }\textbf {\bibinfo
  {volume} {15}},\ \bibinfo {pages} {034080} (\bibinfo {year}
  {2021})}\BibitemShut {NoStop}%
\bibitem [{\citenamefont {Magann}\ \emph
  {et~al.}(2021{\natexlab{b}})\citenamefont {Magann}, \citenamefont {Arenz},
  \citenamefont {Grace}, \citenamefont {Ho}, \citenamefont {Kosut},
  \citenamefont {McClean}, \citenamefont {Rabitz},\ and\ \citenamefont
  {Sarovar}}]{magann2021pulses}%
  \BibitemOpen
  \bibfield  {author} {\bibinfo {author} {\bibfnamefont {A.~B.}\ \bibnamefont
  {Magann}}, \bibinfo {author} {\bibfnamefont {C.}~\bibnamefont {Arenz}},
  \bibinfo {author} {\bibfnamefont {M.~D.}\ \bibnamefont {Grace}}, \bibinfo
  {author} {\bibfnamefont {T.-S.}\ \bibnamefont {Ho}}, \bibinfo {author}
  {\bibfnamefont {R.~L.}\ \bibnamefont {Kosut}}, \bibinfo {author}
  {\bibfnamefont {J.~R.}\ \bibnamefont {McClean}}, \bibinfo {author}
  {\bibfnamefont {H.~A.}\ \bibnamefont {Rabitz}},\ and\ \bibinfo {author}
  {\bibfnamefont {M.}~\bibnamefont {Sarovar}},\ }\bibfield  {title} {\bibinfo
  {title} {From pulses to circuits and back again: A quantum optimal control
  perspective on variational quantum algorithms},\ }\href
  {https://doi.org/https://doi.org/10.1103/PRXQuantum.2.010101} {\bibfield
  {journal} {\bibinfo  {journal} {PRX Quantum}\ }\textbf {\bibinfo {volume}
  {2}},\ \bibinfo {pages} {010101} (\bibinfo {year}
  {2021}{\natexlab{b}})}\BibitemShut {NoStop}%
\bibitem [{\citenamefont {Smith}\ \emph {et~al.}(2019)\citenamefont {Smith},
  \citenamefont {Kim}, \citenamefont {Pollmann},\ and\ \citenamefont
  {Knolle}}]{smith2019simulating}%
  \BibitemOpen
  \bibfield  {author} {\bibinfo {author} {\bibfnamefont {A.}~\bibnamefont
  {Smith}}, \bibinfo {author} {\bibfnamefont {M.}~\bibnamefont {Kim}}, \bibinfo
  {author} {\bibfnamefont {F.}~\bibnamefont {Pollmann}},\ and\ \bibinfo
  {author} {\bibfnamefont {J.}~\bibnamefont {Knolle}},\ }\bibfield  {title}
  {\bibinfo {title} {Simulating quantum many-body dynamics on a current digital
  quantum computer},\ }\href {https://doi.org/10.1038/s41534-019-0217-0}
  {\bibfield  {journal} {\bibinfo  {journal} {npj Quantum Information}\
  }\textbf {\bibinfo {volume} {5}},\ \bibinfo {pages} {1} (\bibinfo {year}
  {2019})}\BibitemShut {NoStop}%
\bibitem [{\citenamefont {Hsieh}\ \emph {et~al.}(2009)\citenamefont {Hsieh},
  \citenamefont {Wu},\ and\ \citenamefont {Rabitz}}]{hsieh2009topology}%
  \BibitemOpen
  \bibfield  {author} {\bibinfo {author} {\bibfnamefont {M.}~\bibnamefont
  {Hsieh}}, \bibinfo {author} {\bibfnamefont {R.}~\bibnamefont {Wu}},\ and\
  \bibinfo {author} {\bibfnamefont {H.}~\bibnamefont {Rabitz}},\ }\bibfield
  {title} {\bibinfo {title} {Topology of the quantum control landscape for
  observables},\ }\href {https://doi.org/https://doi.org/10.1063/1.2981796}
  {\bibfield  {journal} {\bibinfo  {journal} {The Journal of chemical physics}\
  }\textbf {\bibinfo {volume} {130}},\ \bibinfo {pages} {104109} (\bibinfo
  {year} {2009})}\BibitemShut {NoStop}%
\bibitem [{\citenamefont {Ho}\ \emph {et~al.}(2009)\citenamefont {Ho},
  \citenamefont {Dominy},\ and\ \citenamefont {Rabitz}}]{ho2009landscape}%
  \BibitemOpen
  \bibfield  {author} {\bibinfo {author} {\bibfnamefont {T.-S.}\ \bibnamefont
  {Ho}}, \bibinfo {author} {\bibfnamefont {J.}~\bibnamefont {Dominy}},\ and\
  \bibinfo {author} {\bibfnamefont {H.}~\bibnamefont {Rabitz}},\ }\bibfield
  {title} {\bibinfo {title} {Landscape of unitary transformations in controlled
  quantum dynamics},\ }\href
  {https://doi.org/https://doi.org/10.1103/PhysRevA.79.013422} {\bibfield
  {journal} {\bibinfo  {journal} {Physical Review A}\ }\textbf {\bibinfo
  {volume} {79}},\ \bibinfo {pages} {013422} (\bibinfo {year}
  {2009})}\BibitemShut {NoStop}%
\bibitem [{\citenamefont {Meyer}(2021{\natexlab{b}})}]{meyer_2021}%
  \BibitemOpen
  \bibfield  {author} {\bibinfo {author} {\bibfnamefont {J.~J.}\ \bibnamefont
  {Meyer}},\ }\bibfield  {title} {\bibinfo {title} {Fisher {I}nformation in
  {N}oisy {I}ntermediate-{S}cale {Q}uantum {A}pplications},\ }\href
  {https://doi.org/10.22331/q-2021-09-09-539} {\bibfield  {journal} {\bibinfo
  {journal} {{Quantum}}\ }\textbf {\bibinfo {volume} {5}},\ \bibinfo {pages}
  {539} (\bibinfo {year} {2021}{\natexlab{b}})}\BibitemShut {NoStop}%
\bibitem [{\citenamefont {Kingma}\ and\ \citenamefont
  {Ba}(2015)}]{kingma2015adam}%
  \BibitemOpen
  \bibfield  {author} {\bibinfo {author} {\bibfnamefont {D.~P.}\ \bibnamefont
  {Kingma}}\ and\ \bibinfo {author} {\bibfnamefont {J.}~\bibnamefont {Ba}},\
  }\bibfield  {title} {\bibinfo {title} {Adam: {A} method for stochastic
  optimization},\ }in\ \href {http://arxiv.org/abs/1412.6980} {\emph {\bibinfo
  {booktitle} {Proceedings of the 3rd International Conference on Learning
  Representations (ICLR)}}}\ (\bibinfo {year} {2015})\BibitemShut {NoStop}%
\bibitem [{\citenamefont {Mitarai}\ \emph {et~al.}(2018)\citenamefont
  {Mitarai}, \citenamefont {Negoro}, \citenamefont {Kitagawa},\ and\
  \citenamefont {Fujii}}]{mitarai2018quantum}%
  \BibitemOpen
  \bibfield  {author} {\bibinfo {author} {\bibfnamefont {K.}~\bibnamefont
  {Mitarai}}, \bibinfo {author} {\bibfnamefont {M.}~\bibnamefont {Negoro}},
  \bibinfo {author} {\bibfnamefont {M.}~\bibnamefont {Kitagawa}},\ and\
  \bibinfo {author} {\bibfnamefont {K.}~\bibnamefont {Fujii}},\ }\bibfield
  {title} {\bibinfo {title} {Quantum circuit learning},\ }\href
  {https://doi.org/10.1103/PhysRevA.98.032309} {\bibfield  {journal} {\bibinfo
  {journal} {Physical Review A}\ }\textbf {\bibinfo {volume} {98}},\ \bibinfo
  {pages} {032309} (\bibinfo {year} {2018})}\BibitemShut {NoStop}%
\bibitem [{\citenamefont {Schuld}\ \emph {et~al.}(2019)\citenamefont {Schuld},
  \citenamefont {Bergholm}, \citenamefont {Gogolin}, \citenamefont {Izaac},\
  and\ \citenamefont {Killoran}}]{schuld2019evaluating}%
  \BibitemOpen
  \bibfield  {author} {\bibinfo {author} {\bibfnamefont {M.}~\bibnamefont
  {Schuld}}, \bibinfo {author} {\bibfnamefont {V.}~\bibnamefont {Bergholm}},
  \bibinfo {author} {\bibfnamefont {C.}~\bibnamefont {Gogolin}}, \bibinfo
  {author} {\bibfnamefont {J.}~\bibnamefont {Izaac}},\ and\ \bibinfo {author}
  {\bibfnamefont {N.}~\bibnamefont {Killoran}},\ }\bibfield  {title} {\bibinfo
  {title} {Evaluating analytic gradients on quantum hardware},\ }\href
  {https://doi.org/10.1103/PhysRevA.99.032331} {\bibfield  {journal} {\bibinfo
  {journal} {Physical Review A}\ }\textbf {\bibinfo {volume} {99}},\ \bibinfo
  {pages} {032331} (\bibinfo {year} {2019})}\BibitemShut {NoStop}%
\bibitem [{\citenamefont {Mari}\ \emph {et~al.}(2021)\citenamefont {Mari},
  \citenamefont {Bromley},\ and\ \citenamefont
  {Killoran}}]{mari2021estimating}%
  \BibitemOpen
  \bibfield  {author} {\bibinfo {author} {\bibfnamefont {A.}~\bibnamefont
  {Mari}}, \bibinfo {author} {\bibfnamefont {T.~R.}\ \bibnamefont {Bromley}},\
  and\ \bibinfo {author} {\bibfnamefont {N.}~\bibnamefont {Killoran}},\
  }\bibfield  {title} {\bibinfo {title} {Estimating the gradient and
  higher-order derivatives on quantum hardware},\ }\href
  {https://doi.org/10.1103/PhysRevA.103.012405} {\bibfield  {journal} {\bibinfo
   {journal} {Phys. Rev. A}\ }\textbf {\bibinfo {volume} {103}},\ \bibinfo
  {pages} {012405} (\bibinfo {year} {2021})}\BibitemShut {NoStop}%
\end{thebibliography}%

\clearpage
\newpage

\onecolumngrid

\setcounter{section}{0}
\setcounter{proposition}{0}
\setcounter{figure}{0}
\setcounter{theorem}{0}
\setcounter{definition}{0}
\setcounter{corollary}{0}
\renewcommand{\figurename}{SUP FIG.}

\section*{\normalsize{Supplementary Information for ``Theory of Overparametrization in quantum neural networks'' }}

In this Supplementary Information, we present detailed proofs of the theorems, and corollaries presented in the manuscript ``\textit{Theory of overparametrization in quantum neural networks}''.  In addition, here we provide additional details and results for the numerical simulations.

\section{Preliminaries}\label{SP_prelim}

Let us start by recalling that we consider the case when the QNN $U(\thv)$ is a parametrized quantum circuit  with an $L$-layered periodic structure  of the form
\begin{equation}\label{eq:PSA_ansatz-SM}
    U(\thv)=\prod_{l=1}^LU_l(\thv_l)\,, \quad U_l(\thv_l)=\prod_{k=1}^K e^{-i \theta_{lk}H_k}\,,    
\end{equation}
where the index $l$ indicates the layer,  and the index $k$ spans the traceless Hermitian operators $H_k$ that generate the unitaries in the ansatz. Here,  $\thv_{l}=(\theta_{l1},\ldots\theta_{lK})$ are the parameters in a single layer, and $\thv=\{\thv_{1},\ldots,\thv_L\}$ denotes the set of $M=K\cdot L$ trainable parameters in the QNN. In this Supplementary Information, we make use of the following relabelling of the parameters $\theta_{lk}$ and operators $H_k$:

\begin{equation}
    U(\thv) \equiv \prod_{j=1}^{LK} e^{-i \theta_j H_j} \,.
\end{equation}

For convenience we also recall the following definitions:
\begin{definition}[Set of generators $\GC$]\label{def:generators-SM}
Consider a parametrized quantum circuit of the form \eqref{eq:PSA_ansatz-SM}. The set of generators $\GC=\{H_k\}_{k=1}^K$ is defined as the set (of size $|\GC|=K$) of the Hermitian operators that generate the unitaries in a single layer of $U(\thv)$.
\end{definition}
And, the definition for the dynamical Lie Algebra:
\begin{definition}[Dynamical Lie Algebra (DLA)]\label{def:dynamical_lie_algebra-SM}
Consider a set of generators $\GC$ according to Definition~\ref{def:generators-SM}. The DLA $\liea$ is generated by repeated nested commutators of the operators in $\GC$. That is, 
\begin{equation}
\liea={\rm span}\left\langle iH_1, \ldots, iH_K \right\rangle_{Lie}\,,
\end{equation}
where $\left\langle S\right\rangle_{Lie}$ denotes the Lie closure, i.e., the set obtained by repeatedly taking the commutator of the elements in $S$. 
\end{definition}

\textbf{Invariant subspaces.} Consider now the case when the elements in the DLA share a symmetry (for simplicity we assume only one symmetry, although generalization to multiple symmetries is straightforward). That is, there exists a Hermitian operator $\Sigma$ such that $[\Sigma,g]=0$ for all $g\in\liea $. If $\Sigma$ has $N$ distinct eigenvalues, then the DLA has the form $\liea=\bigoplus_{m=1}^N \liea_m$. This imposes a partition of Hilbert space $\HC=\bigoplus_{m=1}^N \HC_m$ where each subspace $\HC_m$ of dimension $d_m$ is invariant under $\liea$.  

Let us introduce some notation. Consider the $d\times d_m$ matrix that results from horizontally stacking the eigenvectors of $\Sigma$ associated with the $m$-th eigenvalue (of degeneracy $g_m$)
\begin{equation}
Q_m\ad = \begin{bmatrix} \vdots &  \vdots  && \vdots \\ \ket{v_1}, & \ket{v_2},&  &, \ket{v_{g_m}} \\ \vdots &  \vdots  && \vdots \end{bmatrix}, 
\end{equation}
such that $Q_m$ maps vectors from $\H$ to $\H_m$. These satisfy 
\begin{equation}
Q_m Q_n\ad = \id_{d_m} \delta_{mn},\quad  Q_m\ad Q_m  = \mathbb{P}_m\,,
\end{equation}\label{eq:proj2}
where $\mathbb{P}_m$ are projectors onto the $m$-th eigenspace, such that $\sum_{m=1}^N \mathbb{P}_m=\id$. Let us now use the notation
\begin{equation}\label{Eq_red}
    \ket{\psi}^{(m)} = Q_m \ket{\psi}, \quad   A^{(m)} = Q_m A Q_m\ad\,,
\end{equation}
to denote the $d_m$-dimensional reduced states and operators, respectively. Recall that, since any unitary $U\in\lieg$ produced by such a system is block diagonal, we can write $U= \sum_m \P_m U \P_m$. Also, let us note that if $A$ is Hermitian, then $A\k$ is also Hermitian.

\section{Proof of Theorem 1}
In the following we provide a proof for Theorem~\ref{theo:1_SM}. Let us first recall the definition of the Quantum Fisher Information Matrix (QFIM).
The QFIM entries are given by
\begin{equation}
    [F(\thv)]_{jk} = 4\Re[ \bra{\partial_j \psi(\thv)} \ket{\partial_k \psi(\thv)} -\bra{\partial_j \psi(\thv)}  \ket{\psi(\thv)} \bra{\psi(\thv)} \ket{\partial_k \psi(\thv)}]\,,
\label{QFIM_jk}
\end{equation}
where $\ket{\psi(\thv)} = U(\thv) \ket{\psi}$. Here we also denote where $\ket{\partial_i\psi(\thv)}=\partial \ket{\psi(\thv)}/\partial\theta_i=\partial_i\ket{\psi(\thv)}$ for $\theta_i\in\thv$. 

We now restate Theorem~\ref{theo:1_SM} for convenience.
\begin{theorem}\label{theo:1_SM}
    For each state $\ket{\psi_\mu}$ in the training set $\SC$, the maximum rank $R_\mu$ of its associated QFIM (defined in Eq.~\eqref{QFIM_jk}) is upper bounded  as
    \begin{equation}
        R_\mu\leq\dim(\liea_\SC)\,.
    \end{equation}
\end{theorem}
\begin{proof}
Let us first note that the partial derivatives of the parametrized state are
\begin{equation}
     \ket{\partial_j\psi(\thv)} = \partial_j (U(\thv)\ket{\psi}) = -i U(\thv) \tilde{H_j} \ket{\psi} ,\quad 1\leq j\leq M\,,
\end{equation}
where 
\begin{equation}\label{SM_Hj}
\tilde{H}_j= U_{1\vsa j}\ad H_j U_{1\vsa j},\quad\text{and where}\quad U_{1\vsa j}= U_j \cdots U_1\,.
\end{equation}
Thus, $ U_{1\vsa j}$ is the propagator up to the j-th layer in the circuit and we are labeling ${H}_j$ modulo $|\GC|$, e.g. ${H}_{|\GC|}={H}_1$, i.e. the first generator.

Next, let us consider the case then the DLA has a symmetry (see the Preliminaries section above) and that all states in the training set belong to the $m$-th invariant subspace of the symmetry. We denote by $\liea_\SC$ the DLA associated with said symmetry respected by the training set, by $\HC_\SC$ the corresponding Hilbert space, and by $\ket{\psi}^{(m)} = Q_m \ket{\psi}$ the projected state according to Eq.~\eqref{Eq_red}. Then, for any pair of states $\ket{\phi},\ket{\chi} \in \SC$, their overlap can be described in terms of the overlap between the corresponding $d_{\SC}$-dimensional reduced states 
\begin{equation}
    \bra{\chi}\ket{\phi} = \bra{\chi^{(\SC)}}\ket{\phi^{(\SC}}\,.
\end{equation}

Then, it is straightforward to see that $\ket{\psi}$,  $\ket{\psi(\thv)}$ and $\ket{\partial_j \psi(\thv)}$ also belong in $\HC_{\SC}$. Hence, the overlaps in Eq.~\eqref{QFIM_jk} can be computed in terms of their reduced counterparts $\ket{\psi(\thv)\S}$ and $\ket{\partial_j \psi(\thv)\S}$. In the following, we will work with everything reduced to such subspace, but to simplify the notation, we will omit the $\SC$ superscript everywhere. For example, whenever we write operator $O$, we actually mean $O^{(\SC)}\in \Cbb^{d_{\SC}\times d_{\SC}}$.

Using the explicit expression for the partial derivatives, we find that the first term in Eq.~\eqref{QFIM_jk} is
\begin{equation}
\begin{split}
    \Re [\bra{\partial_j \psi(\thv)} \ket{\partial_k \psi(\thv)} ] &= \Re[ i(-i)\bra{\psi} \tilde{H_j} U(\thv)\ad U(\thv) \tilde{H_k} \ket{\psi} ]=\Re[\bra{\psi} \tilde{H_j} \tilde{H_k} \ket{\psi} ]\,.
\end{split}
 \end{equation}
Choosing an orthonormal basis containing $\ket{\psi}$ we can rewrite this term as 
\begin{equation}
\begin{split}
      \Re [\bra{\partial_j \psi(\thv)} \ket{\partial_k \psi(\thv)} ] &= \sum_m \Re[ \bra{\psi} \tilde{H_j} \ket{m}\bra{m} \tilde{H_k} \ket{\psi} ]=  \Re [\bra{\psi} \tilde{H_j} \ket{\psi} \bra{\psi} \tilde{H_k} \ket{\psi} ] + \sum_{m\neq\psi} \Re[ \bra{\psi} \tilde{H_j} \ket{m}\bra{m} \tilde{H_k} \ket{\psi} ]\,.
\end{split}
\end{equation}

Proceeding similarly, we find for the second term in Eq.~\eqref{QFIM_jk}
\begin{equation}
\begin{split}
    \Re[\bra{\partial_j \psi(\thv)}  \ket{\psi(\thv)} \bra{\psi(\thv)} \ket{\partial_k \psi(\thv)}] &= \Re[(-i)i \bra{\psi} \tilde{H_j} U(\thv)\ad \ket{\psi(\thv)} \bra{\psi(\thv)} U(\thv)   \tilde{H_k} \ket{\psi} ]\\
    &= \Re [\bra{\psi} \tilde{H_j} \ket{\psi} \bra{\psi} \tilde{H_k} \ket{\psi} ]\,.
\end{split}
\end{equation}

Combining these results we can express the matrix elements of the QFIM as 
\begin{equation}\label{eq:QFIM-elements-1}
\begin{split}
   [F(\thv)]_{jk}&=4\Re[\bra{\psi} \tilde{H_j} \tilde{H_k} \ket{\psi} ]-4\Re [\bra{\psi} \tilde{H_j} \ket{\psi} \bra{\psi} \tilde{H_k} \ket{\psi} ]=4\sum_{m\neq\psi} \Re[ \bra{\psi} \tilde{H_j} \ket{m}\bra{m} \tilde{H_k} \ket{\psi} ]\,.
\end{split}
\end{equation}
Note here that this equations also allows us to express the QFIM elements as $[F]_{jk}=4\Re[ \text{Cov}_{\ket{\psi}}(\tilde{H_j},\tilde{H_k}) ]$. Then, defining the vectors $\vec{R}_{mn}$ and $\vec{I}_{mn}$ with components
\begin{equation}
R_{mn}(j) =\Re[ \bra{m} \tilde{H_j} \ket{n}],\quad I_{mn}(j)=\Im[ \bra{m} \tilde{H_j} \ket{n} ]\,,
\end{equation}
we can express~\eqref{eq:QFIM-elements-1} as
\begin{equation}
\begin{split}
[F(\thv)]_{jk}&=4 \sum_{m\neq \psi} R_{\psi m}(j)R_{m \psi}(k) - I_{\psi m}(j)I_{m\psi}(k)=4 \sum_{m\neq \psi} R_{\psi m}(j)R_{\psi m}(k) + I_{\psi m}(j)I_{\psi m}(k)\,,
\end{split}
\end{equation}
where the second equality follows from the fact that $R_{mn}(j)=R_{nm}(j)$, while $I_{mn}(j)= - I_{nm}(j)$. Thus, one can finally express the QFIM as a sum of $2d-2$ rank-one matrices
\begin{equation}
    \label{eq:F1}
    F(\thv) = -2 \sum_{m=1,m\neq \psi}^d (\vec{R}_{m\psi}\cdot \vec{R}_{m\psi}^{\top} + \vec{I}_{m\psi}\cdot \vec{I}_{m\psi}^{\top})\,.
\end{equation}

Here we recall that, by definition, $H_j$ are elements in the DLA $\liea_\SC$. Then, since the unitaries $U$ are elements of the dynamical Lie group $\mathbb{G}_\SC$ generated by $\liea_\SC$, conjugating $H_j$ by any unitary $U$ results in another element in $\liea_\SC$. That is: $\forall U\in\mathbb{G}_\SC$, and $\forall H_i\in \liea_\SC$ we have $U H_j U\ad\in\liea_\SC$.  Then, by repeating this argument $j$ times, we find that $\tilde{H}_j\in\liea_\SC$, where $\tilde{H}_j$ was defined in Eq.~\eqref{SM_Hj}.

Letting $\{S_\nu \}_{\nu=1}^{\dim(\liea)}$ be a basis of $\liea$, we can express 
\begin{equation}
    \label{eq:Htilde-2}
    \tilde{H_j} = \sum_{\nu=1}^{\rm dim(\liea_{\SC})} a_{\nu}(j) S_{\nu}\,,
\end{equation}
with $a_{\nu}$ real coefficients. Using this fact, we can expand $\vec{R}_{mn}$ and $\vec{I}_{mn}$ in the following ways:
\begin{equation}
\label{eq:Rsubspace}
\begin{split}
    R_{mn}(j) &= \Re[ \sum_{\nu}^{\rm dim(\liea_{\SC})} \bra{m} a_{\nu}(j) S_{\nu} \ket{n}] = \sum_{\nu}^{\rm dim(\liea_{\SC})} \Re [\bra{m} S_{\nu} \ket{n}] a_{\nu}(j)\,,\\
    I_{mn}(j)  &= \Im[ \sum_{\nu}^{\rm dim(\liea_{\SC})} \bra{m} a_{\nu}(j) S_{\nu} \ket{n}] = \sum_{\nu}^{\rm dim(\liea_{\SC})} \Im[ \bra{m} S_{\nu} \ket{n}] a_{\nu}(j)\,.
\end{split}
\end{equation}
More succinctly, we find
\begin{equation}
    \label{eq:F2}
    \vec{R}_{mn} = \sum_{\nu=1}^{\dim(\liea_{\SC})} \Re \bra{m} S_{\nu} \ket{n} \vec{a}_{\nu},\quad \text{and} \quad \vec{I}_{mn} = \sum_{\nu=1}^{\dim(\liea_{\SC})} \Im \bra{m} S_{\nu} \ket{n} \vec{a}_{\nu}\,.
\end{equation}
These equations show that the vectors $\vec{R}_{mn}$ and $\vec{I}_{mn}$ can be expressed as a linear combination of  $\dim(\liea_{\SC})$ other vectors $\{\vec{a}_{\nu} \}$. Then, since the $\vec{R}_{mn}$ and $\vec{I}_{mn}$ generate the $2d-2$ rank-one matrices in the QFIM, we have that $F(\thv)$ has a support on a subspace with a basis that has, at most, $\dim(\liea_\SC)$ elements. Thus, we find 
\begin{equation}
    \rank[F_\mu(\thv)]\leq \dim(\liea_\SC)\,,
\end{equation} 
where we have recovered the $\mu$ dependence of the QFIM.
\end{proof}

Let us here note that the proof of Theorem~\ref{theo:1_SM} holds for all states in the training set $\ket{\psi_{\mu}} \in \SC$. Thus, for all $\ket{\psi_{\mu}}$ we know that the associated QFIM $F_{\mu}(\thv)$ has a column space contained within some fixed $\rm dim(\liea_{\SC})$ dimensional space. More precisely, from the previous proof, we have that the following Proposition holds. 
\begin{proposition}
    \label{Ffixedsubspace}
    There is some vector space spanned by $\rm dim(\liea_{\SC})$ vectors $\{ \vec{a}_{\nu} \}_{\nu=1}^{\dim(\liea_{\SC})}$, such that for any state in the training set $\ket{\psi_{\mu}} \in \SC$, the associated QFIM $F_{\mu}(\thv)$ has a column space contained within this vector space.
\end{proposition}
We will make use of this proposition in the following section. 

\section{Proof of Theorem 2}

In the following, we provide a proof for Theorem \ref{theo:2_SM}. For convenience, we here recall the definition of overparametrization as well as the statement for Theorem~\ref{theo:2_SM}.
\begin{definition}[Overparametrization]\label{def:overparametrization-SI}
    A QNN is said to be overparametrized if the number of parameters $M$ is such that the QFI matrices, for all the states in the training set, simultaneously saturate their achievable rank $R_\mu$ at least in one point of the loss landscape. That is, if increasing the number of parameters past some minimal (critical) value $M_c$ does not further increase the rank of any QFIM:
    \begin{equation}\label{eq:Rmu-SI}
        \max_{M\geq M_c,\thv}\rank[F_\mu(\thv)] = R_\mu\,.
    \end{equation}
\end{definition}

Let us also recall two definitions of a QNN's effective dimension. First, following~\cite{haug2021capacity}, we can define the average effective quantum dimension of a QNN:
\begin{equation}\label{eq:SM_eff_dim_1}
    D_1(\thv)=\mathbb{E}\left[\sum_{i=1}^M\IC(\lambda^{i}_\mu(\thv))\right]\,,
\end{equation}
where $\lambda^{i}_\mu(\thv)$ are the eigenvalues of $F_{\mu}(\thv)$, and where $\IC(x)=0$ for $x=0$, and $\IC(x)=1$ for $x\neq1$. Here the expectation value is taken over the probability distribution  that samples input states from the dataset. 

The second definition follows from~\cite{abbas2020power}. In the $n\rightarrow\infty$ limit, the effective quantum dimension of~\cite{abbas2020power} converges to
\begin{equation}\label{eq:SM_eff_dim_2}
    D_2=\max_{\thv}\left(\rank\left[\widetilde{F}(\thv)\right]\right)\,,
\end{equation}
where $\widetilde{F}(\thv)$ is the classical Fisher Information matrix obtained as
\small
\begin{equation}
    \widetilde{F}(\thv)=\mathbb{E}\left[\frac{\partial\log(p(\ket{\psi_\mu},y_\mu;\thv))}{\partial\thv}\frac{\partial\log(p(\ket{\psi_\mu},y_\mu;\thv))}{\partial\thv}^T\right].
\end{equation}
\normalsize 
Here, $p(\ket{\psi},y;\thv)$, describes the joint relationship between an input $\ket{\psi}$ and an output $y$ of the QNN. In addition, the expectation value is taken over the probability distribution that samples input states from the dataset.

Then, consider the following theorem in the main text.

\begin{theorem}\label{theo:2_SM}
The model capacity, as quantified by the effective dimensions of Eqs.~\eqref{eq:SM_eff_dim_1} or~\eqref{eq:SM_eff_dim_2}, is upper bounded as
\begin{equation}
    D_1(\thv)\leq \dim(\liea_\SC),\quad D_2\leq \dim(\liea_\SC)\,.
\end{equation}
Moreover, when the QNN is overparametrized according to Definition~\ref{def:overparametrization}, $D_1(\thv)$ achieves its maximum value on at least one point of the landscape. 
\end{theorem}

\begin{proof}
    When the QNN is overparametrized according to Definition~\ref{def:overparametrization}, there exists some $\thv$ such that the ranks of the QFIMs     are maximized. That is, $\rank[F\mu(\thv)]=R_\mu$.
    
    Note that the effective dimension $D_1(\thv)$ can be expressed as $D_1(\thv)=\mathbb{E}\left[\rank[F_\mu(\thv)]\right]$. Then, since the ranks are maximal at the overparametrization, so is $D_1(\thv)$. More precisely, we have $D_1(\thv)=\mathbb{E}\left[R_\mu\right]$. Additionally, as shown in Theorem~\ref{theo:1}, $\rank[F_{\mu}(\thv)] \leq \rm dim(\liea_{\SC})$ for all $\thv$ and $\ket{\psi_{\mu}} \in \SC$, so $D_1(\thv)\leq \rm dim(\liea_{\SC})$. 
    
    Next we consider the effective dimension $D_2$ (Eq.~\eqref{eq:SM_eff_dim_2}). $D_2$ specifically quantifies the maximal rank of the expectation value of the classical Fisher information matrices $\widetilde{F}_{\mu}(\thv)$. Because the operator $F_{\mu}(\thv) - \widetilde{F}_{\mu}(\thv)$ is positive semidefinite (\cite{meyer_2021}, Section 5), the following holds:
    \begin{equation}
        \widetilde{F}(\thv) = \mathbb{E}_{\mu}[\widetilde{F}_{\mu}(\thv)] \leq \mathbb{E}_{\mu}[F_{\mu}(\thv)]\,.
    \end{equation}
    In addition for any two symmetric matrices $A$ and $B$, having $A \leq B$ implies that $A^0 \leq B^0$ (\cite{chan1985hermitian}, Theorem 3). Thus,
    \begin{equation}
        (\widetilde{F}(\thv))^0 \leq (\mathbb{E}_{\mu}[F_{\mu}(\thv)])^0 \,,
    \end{equation}
    implying that
    \begin{equation}
        \rank[\widetilde{F}(\thv)] = \Tr[(\widetilde{F}(\thv))^0] \leq \Tr[(\mathbb{E}_{\mu}[F_{\mu}(\thv)])^0] = \rank[\mathbb{E}_{\mu}[F_{\mu}(\thv)]]\,.
    \end{equation}
    By applying Proposition~\ref{Ffixedsubspace} to $\mathbb{E}_{\mu}[F_{\mu}(\thv)]$, we arrive at the desired result:
    \begin{equation}
        \rank[\widetilde{F}(\thv)] \leq \rank[\mathbb{E}_{\mu}[F_{\mu}(\thv)]] \leq \rm dim(\liea_{\SC})\,.
    \end{equation}
\end{proof}

\section{Proof of Theorem 3}
In the following we prove Theorem~\ref{theo:3-si}, which bounds the rank of the Hessian for an observable minimization loss function of the form
\begin{equation}\label{eq:lineal-loss_si}
    \LC(\thv)=\sum_{\ket{\psi_\mu}\in\SC}c_\mu \Tr[U(\thv)\dya{\psi_\mu}U\ad(\thv)O]\,,
\end{equation}
at its optimum. Here, the terms $c_\mu$ are real coefficients associated with each state $\ket{\psi_\mu}$ in $\SC$, and where the operator $O$ is Hermitian. Let us restate the theorem for convenience.

\begin{theorem}\label{theo:3-si}
  Let $\nabla^2\LC(\thv_*)$  be the Hessian for a loss function of the form of Eq.~\eqref{eq:lineal-loss_si} evaluated at the optimum set of parameters $\thv_*$. Then, its rank is upper bounded as
    \begin{equation}
        \rank[\nabla^2\LC(\thv_*)]\leq\min \{ \dim(\liea_\SC), 2dr - r^2 - r\}\,,
    \end{equation}
    where $r=\min\{\rank[\sum_{\mu} c_{\mu} \dya{\psi_{\mu}}],\rank[O]\}$, and $d$ is the Hilbert space dimension.
\end{theorem}
\begin{proof}
Let us define $\rho =\sum_{\mu} c_{\mu} \ket{\psi_{\mu}} \bra{\psi_{\mu}}$. First, the gradient has the following form:
\begin{equation}
\label{eq:obsgrad}
\begin{split}  
    \partial_j \LC(\thv) &= -\Tr\left[\partial_j U(\thv) \rho U\ad(\thv) O - U(\thv) \rho \partial_j U\ad(\thv) O\right]= i\Tr \left[ [\H_j, \rho] \,O_f \right]\,,
\end{split} 
\end{equation}
where $\tilde{H_j}$ is defined in Eq.~\eqref{SM_Hj}, and where we defined $O_f=U(\thv)\ad O U(\thv)$ (we henceforth drop the explicit dependence on $\thv$). Going forward, we similarly drop a the explicit dependence on $\thv$ on terms which are not being differentiated. If we assume that $i \leq j$, and note that $\partial_i \H_j =i[\H_i,\H_j]$ in this case, then we can express the matrix elements of the Hessian as
\begin{equation}
\label{eq:obshessian}
\begin{split}
    \partial_i \partial_j \LC(\thv) &= i \partial_i \Tr\left[U(\thv) \tilde{H_j}(\thv) \rho U\ad(\thv) O - U(\thv) \rho \tilde{H_j}(\thv) U\ad(\thv) O\right] \\
    &= \Tr[\tilde{H_i} \tilde{H_j} \rho O_f] - \Tr[[\tilde{H_i}, \tilde{H_j}] \rho O_f] - \Tr[ \tilde{H_j} \rho \tilde{H_i} O_f] \\
    &- \Tr[ \tilde{H_i} \rho \tilde{H_j} O_f] + \Tr[ \rho [\tilde{H_i}, \tilde{H_j}] O_f] + \Tr[\rho \tilde{H_j} \tilde{H_i}O_f] \\
    &= 2 \Re\left[ \Tr[\rho \tilde{H_i} \tilde{H_j} O_f]\right] - 2\Re\left[ \Tr[\tilde{H_i} \rho \tilde{H_j}O_f]\right]\,. \\
\end{split}
\end{equation}

We now evaluate the Hessian at the optimum $\thv_*$. Here, the propagator has the form \cite{chakrabarti2007quantum} $U(\thv_*) = R\ad Q$ for unitaries $R$ and $Q$ that respectively diagonalize $\rho$ and $O$, i.e. $\rho = Q\ad e Q$ and $O = R\ad o R $, such that $e$ ($o$) is a diagonal matrix containing  the eigenvalues of $\rho$ ($O$) in decreasing (increasing) order. Therefore, we can rewrite Eq.~\eqref{eq:obshessian} at the optimum as
\begin{equation}
\label{eq:final_func_form_obs}
\begin{split}
    \partial_i \partial_j \LC(\thv_*) &= 2 \Re \left[ \Tr[\tilde{H_i} \tilde{H_j} O_f \rho] \right] - 2\Re \left[\Tr[\tilde{H_i} \rho \tilde{H_j} O_f] \right] \\
    &= 2 \R \left[\Tr[\tilde{H_i} \tilde{H_j} Q\ad oe Q] \right] - 2\Re \left[ \Tr[\tilde{H_i} Q\ad e Q \tilde{H_j} Q\ad o Q] \right]\\
    &= 2  \sum_{m,n=1}^d (o_m e_m - o_me_n) \Re \left[ \bra{m} Q \tilde{H_i} Q\ad  \ket{n} \bra{n}Q \tilde{H_j} Q\ad \ket{m} \right]\\
    &= 2 \sum_{m,n = 1}^d (o_m e_m - o_me_n) (R_{mn}'(i)R_{mn}'(j)+I_{mn}'(i)I_{mn}'(j)) \,,
\end{split}
\end{equation}
where we have used $O_f = Q \ad o Q$ at the optimum and defined $R_{mn}'(j) = \Re [\bra{m} Q \tilde{H_j} Q\ad \ket{n}]$ and $I_{mn}'(j) = \Im[ \bra{m} Q \tilde{H_j} Q\ad \ket{n}]$. Because Eq.~\eqref{eq:final_func_form_obs} is symmetric in indices $i$ and $j$, we can remove the assumption that $i \leq j$. By following a proof similar to that in Theorem~\ref{theo:1_SM}, we have an upper bound of $\rm dim(\liea_\SC)$ on the rank of the Hessian $\nabla^2\LC(\thv_*)$ because $R_{mn}'(\cdot)$ and $I_{mn}'(\cdot)$ reside in a $\rm dim(\liea_{\SC})$ dimensional space; see Eq.~\eqref{eq:F2}.

We will now establish the additional $2dr - r^2 -r$ upper bound stated in the theorem. We will use the short hand $r(o)$ and $r(e)$ for ranks of $o$ and $e$, respectively. Assume that $r(e) \leq r(o)$ (the case of $r(o) \leq r(e)$ proceeds similarly). We would like to split Eq.~\eqref{eq:final_func_form_obs} into disjoint summations over $m$ and $n$. Toward that goal, let us define $t_{mn}(i,j) = R_{mn}'(i)R_{mn}'(j)+I_{mn}'(i)I_{mn}'(j)$ to rewrite Eq.~\eqref{eq:final_func_form_obs}:

\begin{equation}
    \begin{split}
        \frac{1}{2} \partial_i \partial_j \LC(\thv_*) =& \sum_{m} o_me_m t_{mm}(i,j) + \sum_{m > n} o_m e_m t_{mn}(i,j) + \sum_{m < n} o_m e_m t_{mn}(i,j) \\
        &- \sum_{m} o_{m} e_{m} t_{mm}(i,j) - \sum_{m > n} o_{m} e_{n} t_{mn}(i,j) - \sum_{m < n} o_{m} e_{n} t_{mn}(i,j)  \\
        =& \sum_{m > n, m \leq r(e)} o_m e_m t_{mn}(i,j)
        + \sum_{m < n, m \leq r(e)} o_m e_m t_{mn}(i,j)
        - \sum_{m > n, m \leq r(o), n \leq r(e)} o_{m} e_{n} t_{mn}(i,j) \\
        &- \sum_{m < n, n \leq r(e)} o_{m} e_{n} t_{mn}(i,j) \,,
    \end{split}
\end{equation}
where in the first equality we have simply split the sums among $m=n$, $m>n$, and $m<n$. Then, in the second sum we have attached more specific subscripts to the summation and used the fact that $r(e) \leq r(o)$. We now combine the sums over $m > n$ (and the same for $m < n$, separately) to arrive at

\begin{equation}
\begin{split}
    \frac{1}{2} \partial_i \partial_j \LC(\thv_*) =& \sum_{m > n; m, n \leq r(e)} (o_m e_m - o_m e_n) t_{mn}(i,j)
    + \sum_{m < n; m, n \leq r(e)} (o_m e_m - o_m e_n) t_{mn}(i,j) \\
    &- \sum_{m > n, r(e) < m \leq r(o), n \leq r(e)} o_{m} e_{n} t_{mn}(i,j)
    + \sum_{m < n, m \leq r(e)} o_{m} e_{m} t_{mn}(i,j)\,,
\end{split}
\end{equation}
where the second summation contains the leftover terms when combining over $m > n$ and the fourth summations contains the leftover terms when combining over $m < n$. Note that $R_{mn}'=R_{nm}'$ and $I_{mn}' = - I_{nm}'$. This means that we can also combine more terms between the first and second summations, and also combine terms between the third and fourth summations. By rewriting terms so that $m > n$ and combining, we arrive at

\begin{equation}
\begin{split}
    \frac{1}{2} \partial_i \partial_j \LC(\thv_*) =& \sum_{m > n; m, n \leq r(e)} (o_m e_m - o_m e_n + o_n e_n - o_n e_m) t_{mn}(i,j) \\
    &- \sum_{m > n, r(e) < m \leq r(o), n \leq r(e)} (-o_m e_n + o_{n} e_{n}) t_{mn}(i,j) \\
    &+ \sum_{m > n, m > r(o), n \leq r(e)} o_{n} e_{n} t_{mn}(i,j)\,.
\end{split}
\end{equation}
As a result, the Hessian can be expressed as
\begin{equation}
\begin{split}
    \frac{1}{2} \nabla^2\LC(\thv_*) &= \sum_{m > n; m, n \leq r(e)} (o_m e_m - o_m e_n + o_n e_n - o_n e_m) (\vec{R}_{mn}\cdot \vec{R}_{mn}^{\top} + \vec{I}_{mn}\cdot \vec{I}_{mn}^{\top}) \\
    &- \sum_{m > n, r(e) < m \leq r(o), n \leq r(e)} (-o_m e_n + o_{n} e_{n}) (\vec{R}_{mn}\cdot \vec{R}_{mn}^{\top} + \vec{I}_{mn}\cdot \vec{I}_{mn}^{\top}) \\
    &+ \sum_{m > n, m > r(o), n \leq r(e)} o_{n} e_{n} (\vec{R}_{mn}\cdot \vec{R}_{mn}^{\top} + \vec{I}_{mn}\cdot \vec{I}_{mn}^{\top})\,,
\end{split}
\end{equation}
where the $j$'th entry of $\vec{R}_{mn}$ and $\vec{I}_{mn}$ are $R_{mn}'(j)$ and $I_{mn}'(j)$, respectively. Now each summation is completely disjoint over $(m, n)$ pairs, so the remaining projectors, $\vec{R_{mn}}\cdot \vec{R_{mn}}^{\top}$ and $\vec{I_{mn}}\cdot \vec{I_{mn}}^{\top}$, are those such that $m > n$ and $n \leq r(e)$. This gives an upper bound on the rank of the Hessian when $r(e) \leq r(o)$ as
\begin{equation}
    \text{rank}(\nabla^2 \LC (\thv_*)) \leq 2 {r(e) \choose 2} + 2 r(e) (d - r(e)) = r(e)^2 - r(e) + 2 r(e) (d - r(e))\,.
\end{equation}
A similar analysis for the case of $r(o) \leq r(e)$ reveals 
\begin{equation}
    \text{rank}(\nabla^2 \LC (\thv_*)) \leq r(o)^2 - r(o) + 2 r(o) (d - r(o))\,.
\end{equation}
Thus, defining $r = \min\{r(e), r(o)\}$, we have an upper bound on the rank of the Hessian as
\begin{equation}
    \text{rank}(\nabla^2 \LC (\thv_*)) \leq 2dr - r^2 - r\,.
\end{equation}
    
\end{proof}

\section{Proof of Theorem 4}

Here we present a proof for Theorem~\ref{theo:4-si}, which upper bounds the rank of the Hessian (evaluated at the solution) for a unitary compilation task. Here the goal is to train a QNN so that its action matches that of a target unitary $V$. We consider two possible loss functions for this task
\begin{equation}
\begin{split}
    \LC_1 (\thv) = 2d - 2 \Re [\Tr[V\ad U(\vec{\theta})]],\quad \text{and} \quad \LC_2 (\thv) = 1 - \frac{1}{d^2}|\Tr[V\ad U(\vec{\theta})]|^2\,.
\end{split}
\end{equation}
Where $\LC_1 (\thv)$ is minimized if $U(\vec{\theta})=V$, while $\LC_2 (\thv)$ is minimized if $U(\vec{\theta})=e^{i \phi}V$, for some any phase $\phi$.

We recall now the statement of Theorem~\ref{theo:4-si},:
\begin{theorem}\label{theo:4-si}
Consider the loss functions for a unitary compilation task
\begin{equation}\label{SM_eq:unitary_cost}
\begin{split}
    \LC_1 (\thv) = 2d - 2 \Re [T(\thv)],\quad \text{and} \quad \LC_2 (\thv) = 1 - \frac{1}{d^2}|T(\thv)|^2, \nonumber
\end{split}
\end{equation}
where $T(\thv)= \Tr[V\ad U(\thv)]$ for a target unitary $V$. Then, let $\nabla^2\LC_1(\thv_*)$ and $\nabla^2\LC_2(\thv_*)$  be the Hessians for the loss functions  $\LC_1 (\thv)$ and  $\LC_1 (\thv)$, respectively evaluated at their solutions $U(\thv_*)=V$ and $U(\thv_*)=e^{i\phi}V$. Then, the maximal ranks of  $\nabla^2\LC_1(\thv_*)$ and $\nabla^2\LC_2(\thv_*)$  are such that $\rank[\nabla^2\LC_1(\thv_*)],\rank[\nabla^2\LC_2(\thv_*)]\leq \dim(\liea)$.
\end{theorem}

\begin{proof}
Let us begin with $\LC_1(\thv)$. The gradient of this loss function is
\begin{equation}
\begin{split}
    \vec{\nabla} L_1(\thv) &= -2\Re\left[ \vec{\nabla} F(\thv)\right] \\
&= -2 \Re\left[\Tr\left[V\ad \vec{\nabla}U(\thv)\right]\right] \\
&= 2\Im\left[ \Tr\left[V\ad U(\thv)\vec{\tilde{H}}(\thv)\right]\right]\,,
\end{split}
\end{equation}
where $\vec{\tilde{H}(\thv)}=(\tilde{H_1}(\thv),\cdots,\tilde{H_M}(\thv))^\top$, and where $\tilde{H_j}(\thv)$ was defined in Eq.~\eqref{SM_Hj}. Similarly, assuming $i\leq j$, we find for the matrix elements of the  Hessian 
\begin{equation}
    \label{eq:thm4beforechange}
    \begin{split}
    \partial_i\partial_j\LC_1(\thv) &= 2\Im\left[\Tr\left[ V\ad \partial_i U(\thv) \H_j + V\ad U(\thv) \partial_i\H_j\right] \right] \\
    &= -2\Re\left[\Tr\left[ V\ad U(\thv) \H_i\H_j - V\ad U(\thv) [\H_i,\H_j] \right]\right]\\
    &=-2\Re\left[\Tr\left[ V\ad U(\thv) \H_i\H_j \right]\right]\,.
    \end{split}
\end{equation}

Evaluating at any optimal set of parameters, that is, such that $U(\thv_*)=V$, we find that Eq.~\eqref{eq:thm4beforechange} is symmetric in indices $i$ and $j$. Thus, we can remove the assumption that $i \leq j$ and express the Hessian more succinctly:
\begin{equation}
\begin{split}
    \nabla^2\LC_1(\thv_*) &= -2\Re \left[\Tr\left[\vec{\tilde{H}}\cdot \vec{\tilde{H}}^\top \right]\right] = -2\sum_{m,n=1}^d \vec{R}_{mn}\cdot\vec{R}_{mn}^\top + \vec{I}_{mn}\cdot\vec{I}_{mn}^\top
\end{split}
\end{equation}
where $\vec{R}_{mn}=\Re[\bra{m} \vec{\tilde{H}}\ket{n}]$ and $\vec{I}_{mn}=\Im[\bra{m} \vec{\tilde{H}}\ket{n}]$. Hence,  we again find that the Hessian is a sum of $d^2$ rank-one matrices. We note that from here onward, we drop the $\thv$ dependence of $\vec{\tilde{H}}$.

Then, following a proof similar to the one used in proving Theorem~\ref{theo:1_SM}, we know that each of the vectors generating the matrices $\vec{R}_{mn}$ and $\vec{I}_{mn}$ can be written as a linear combination of $\dim (\liea)$ other vectors $\{ \vec{a_\nu} \}_{\nu=1}^{\dim(\liea)}$. Thus, the rank of the Hessian of $\nabla^2\LC_1(\thv)$ at the optimum is upper bounded be larger that $\dim(\liea)$.

Now, let us derive the result for $\LC_2(\thv)$. The gradient of the loss function is
\begin{equation}
\begin{split}
    \vec{\nabla}\LC_2(\thv) &= -\frac{2}{d^2}\Re\left[ \vec{\nabla} T(\thv) T^*(\thv)\right] = -2 \Re\left[\Tr\left[V\ad \vec{\nabla}U(\thv) \right]T^*(\thv)\right] = 2\Im\left[ \Tr \left[V\ad U(\thv) \vec{\tilde{H}}\right]T^*(\thv)\right]
\end{split}
\end{equation}
and the Hessian
\begin{equation}
\begin{split}
    \nabla^2\LC_2(\thv) &= \frac{2}{d^2}\Im\left[ \vec{\nabla}\cdot\left(\Tr[ V\ad U(\thv) \vec{\tilde{H}}^\top]\right)T^*(\thv)+\vec{\nabla}T^*(\thv) \cdot\Tr[V\ad U(\thv)\vec{\tilde{H}}^\top ] \right] \\
    &=
    -\frac{2}{d^2}\Re\left[ \Tr[V\ad U(\thv) \vec{\tilde{H}}\cdot \vec{\tilde{H}}^\top]T^*(\thv) + \Tr[\vec{\tilde{H}} U\ad(\thv)V ] \cdot \Tr[V\ad U(\thv) \vec{\tilde{H}}^\top]\right]\,.
\end{split}
\end{equation}
Evaluating at a solution $U(\thv)=e^{i\phi}V$
\begin{equation}
\begin{split}
    \nabla^2\LC_2(\thv_*) &= -\frac{2}{d^2}\Re\left[ \Tr[\vec{\tilde{H}} \cdot \vec{\tilde{H}}^\top]d + \Tr[\vec{\tilde{H}}] \cdot \Tr[\vec{\tilde{H}}^\top]\right] \\
    &= -\frac{2}{d} \sum_{m,n=1}^d \vec{R}_{mn}\cdot\vec{R}_{mn}^\top +\vec{I}_{mn}\cdot\vec{I}_{mn}^\top +\frac{1}{d}\left( \vec{R}_{mm}\cdot\vec{R}_{nn}^\top-\vec{I}_{mm}\cdot\vec{I}_{nn}^\top \right) \,.
\end{split}
\end{equation}
Again, this is a sum of rank-one matrices that live in the span of $\{ \vec{a}_\nu \}$, and following a proof similar to that used in deriving Theorem~\ref{theo:1_SM}, the rank of $ \nabla^2\LC_2(\thv_*)$ is at upper bounded by $\dim(\liea)$.
\end{proof}

\section{Details of the numerical simulations}

The simulations in the main text were  carried out in double precision using the open-source library \texttt{Qibo} \cite{efthymiou2020qibo,efthymiou2021qibo_zenodo} (version 0.1.6). All circuits have been run on CPU because the overhead of transferring the state vector  between the host and the device makes the usage of GPUs not suitable for circuits with less than 15-20 qubits. This is specially true for the case of VQAs, where back and forth communication between host and device results in a deteriorated performance. Simulations have been performed using single-thread multiprocessing to parallelize the execution of different instances of circuits with different number of qubits and depths in multiple cores. In particular, IntelCore i7-9750H, IntelCore i7-10750H and IntelCore i9-9900K cores have been employed.

The optimization method chosen in all cases has been the Adaptive Moment Estimation (Adam) algorithm \cite{kingma2015adam}, which is a variant of Stochastic Gradient Descent (SGD) widely used in classical machine learning, that adaptively adjusts the learning rate for each optimization parameter based on information coming from first and second moments of the gradients. This choice has been motivated by the fact that the works that have reported overparametrization in VQAs used this algorithm \cite{kiani2020learning, wiersema2020exploring}, and also because we consider that a gradient-based optimizer is an appropriate choice to probe relevant features of the optimization landscape, like the disappearance of suboptimal local minima.
In order to leverage automatic differentiation for the computation of gradients, the simulation backend in \texttt{Qibo} has been set to \texttt{tensorflow}. This backend, although slower than the \texttt{qibojit} and \texttt{qibotf} custom backends, allows to seamlessly deploy \texttt{Tensorflow}'s implementation of the Adam optimizer.
The hyper-parameter values employed in all cases are: initial learning rate $=10^{-2}$, $\beta_1=0.9$, $\beta_2=0.999$ and $\hat{\epsilon}=10^{-7}$. The optimization was stopped whenever we reached machine precision.

The minimizations have been carried out in all cases without considering sampling noise, \textit{i.e.} using the full state vector in the simulation to compute expectation values of observables. The main reason for this is that we are here interested in the optimization landscape itself, and not in the stochasticity introduced by finite sampling.

Finally, we mention that parameter-shift rules \cite{mitarai2018quantum,schuld2019evaluating} have been employed in all cases for the computation of the quantum Fisher information and Hessian matrices, and that the simulation backend was switched to the faster \texttt{qibojit} for that.

\section{Formulas for computing the QFIM and Hessian}

We present here the explicit formulas employed in the computation of the Quantum Fisher Information Matrix, $F(\thv)$, and the Hessian, $\nabla^2 \LC(\thv)$, in each of the examples in our numerical simulations. For convenience, we recall the definitions of the elements of these two matrices,
\be [F_\mu(\thv)]_{ij} = 4\Re[\braket{\partial_i\psi_\mu(\thv)}{\partial_j\psi_\mu(\thv)}\\
-\braket{\partial_i\psi_\mu(\thv)}{\psi_\mu(\thv)}\braket{\psi_\mu(\thv)}{\partial_j\psi_\mu(\thv)}] \quad,\quad [\nabla^2\LC(\thv)]_{ij}=\partial_i\partial_j \LC(\thv)\,,\ee
where we use the notation $\partial_{i}=\frac{\partial}{\partial \theta_{i}}$ and where the subscript $\mu$ indicates the quantum state $\ket{\psi_\mu}$ the QNN acts on.
The QFIM can be interpreted, at each point $\vec{\theta}$ in the landscape (and up to a constant factor), as the Hessian matrix of a pure state transfer problem where the target state is $\ket{\psi(\vec{\theta})}$ itself. This opens up the possibility of employing quantum circuits to evaluate the QFIM on quantum hardware using parameter shift rules \cite{mitarai2018quantum,schuld2019evaluating}, which may be useful for instance when computing the natural gradient during the optimization of a VQA \cite{stokes2020quantum}, or when doing variational quantum simulation of imaginary-time evolution \cite{mcardle2019variational}. 
The parameter shift rules are simple recipes to analytically compute partial derivatives of a given loss function with respect to parametrized quantum gates. In the case of the QFIM, its elements are given by~\cite{mari2021estimating}
\be [F_\mu(\thv)]_{ij}= \frac{1}{2}\,\partial_i\partial_j\LC(\vec{\theta})\Big\vert_{\vec{\theta'} = \vec{\theta}}
\ee
where
\be \LC(\vec{\theta}) = 1-|\bra{\psi(\vec{\theta'})}\ket{\psi(\vec{\theta})}|^2.\ee
We start by computing the QFIM and the Hessian for the Hamiltonian Variational Ansatz (HVA) employed in the Variational Quantum Eigensolver (VQE) implementation. We recall that a HVA is an ansatz of the form in Eq. \eqref{eq:PSA_ansatz-SM} where the generators $\GC$, for a given Hamiltonian $H = \sum_{k=1}^N a_k A_k$ (with $A_k$ Hermitian operators and $a_k$ real numbers), are simply $\GC=\{ A_k \}_{k=1}^N$.
We employed this type of ansatz in the main text to minimize the loss function
\be E(\thv) = \langle \psi(\thv)|H_{\TFIM}|\psi(\thv)\rangle \,,\ee
where $H_{\TFIM}$ is the Hamiltonian of the Transverse Field Ising Model (TFIM).
Making use of the fact that for a HVA applied to the TFIM Hamiltonian, all the terms commute within a given $e^{-i \theta_{lk} H_k}$ operator (where $H_k = \frac{1}{2}\sum_i \sigma^z_i\sigma^z_{i+1}$ or $H_k=\frac{1}{2}\sum_i \sigma^x_i$), we find that
\be \label{eqn:hva_partial_derivative} \begin{split}
\partial_{lk}\, e^{-i \theta_{lk} H_k} &= \partial_{lk}\,\prod_i e^{-i \theta_{lk} H_{ki} } =  \sum_i \prod_{j\neq i} e^{-i \theta_{lk} H_{kj} }\, \partial_{lk}\, e^{-i \theta_{lk} H_{ki} } =  \sum_i \prod_{j\neq i} e^{-i \theta_{lkj} H_{kj} }\, \partial_{lk}\, e^{-i \theta_{lki} H_{ki} } \,, 
\end{split}\ee
where $H_{ki}=\sigma^z_i\sigma^z_{i+1}$ or $H_{ki}= \sigma^x_i$. For convenience, we have additionally introduced the notation $\theta_{lki}$ to denote the parameter in the $l$-th layer, that parametrizes the $i$-th term of the $k$-th generator. Note that in a periodic-structured ansatz, $\theta_{lki}=\theta_{lk}$ for all $i$.  Now, the partial derivative $\partial_{lk}\, e^{-i \theta_{lk} H_{ki}}$ can be obtained by applying the parameter shift rule (since $H_{ki}$ has only two distinct non-zero eigenvalues), and hence the partial derivative of a loss function with respect to $\theta_{lk}$ is
\be \partial_{lk} \, \LC(\vec{\theta}) = \frac{1}{2} \sum_{i} \left(\LC\left(\vec{\theta}_{\overline{lki}}, \theta_{lki}^{\frac{\pi}{2}}\right) - \LC\left(\vec{\theta}_{\overline{lki}}, \theta_{lki}^{-\frac{\pi}{2}} \right) \right)\,,\ee
where $\overline{lki}$ denotes all the indices distinct from $l,k,i$, and $\theta_{lki}^{s}= \theta_{lki} +s$. Therefore, applying the parameter shift rule twice, the matrix elements of the QFIM are given by
\begin{equation} \label{eqn:hva_Hessian}
 \begin{split}
[F(\thv)]_{lk,\,l'k'} =  \frac{1}{8} \sum_{i,j} \Big(&\LC\left(\vec{\theta}_{\overline{lki,\, l'k'j}}, \theta_{lki}^{\frac{\pi}{2}}, \theta_{l'k'j}^{\frac{\pi}{2}}\right)+ \LC\left(\vec{\theta}_{\overline{lki,\, l'k'j}}, \theta_{lki}^{-\frac{\pi}{2}}, \theta_{l'k'j}^{-\frac{\pi}{2}}\right)\\ - & \, \LC\left(\vec{\theta}_{\overline{lki,\, l'k'j}}, \theta_{lki}^{\frac{\pi}{2}}, \theta_{l'k'j}^{-\frac{\pi}{2}}\right) - \, \LC\left(\vec{\theta}_{\overline{lki,\, l'k'j}}, \theta_{lki}^{-\frac{\pi}{2}}, \theta_{l'k'j}^{\frac{\pi}{2}}\right)\Big)\,. 
\end{split} 
\end{equation}

To analytically compute the Hessian matrix of the loss function $E(\thv)$ for the HVA, we can also apply twice the parameter shift rule. The matrix elements $\nabla^2 E(\thv)_{lk,\,l'k'}  =\partial_{lk}\, \partial_{l'k'}\, E(\vec{\theta})$ of the Hessian are thus given by (see \textit{e.g.} \cite{huembeli2021characterizing,cerezo2020impact})
\begin{equation} \label{eqn:hva_Hessian2}
 \begin{split}
[\nabla^2 E(\thv)]_{lk,\,l'k'} =  \frac{1}{4} \sum_{i,j} \Big(&E\left(\vec{\theta}_{\overline{lki,\, l'k'j}}, \theta_{lki}^{\frac{\pi}{2}}, \theta_{l'k'j}^{\frac{\pi}{2}}\right)+ \, E\left(\vec{\theta}_{\overline{lki,\, l'k'j}}, \theta_{lki}^{-\frac{\pi}{2}}, \theta_{l'k'j}^{-\frac{\pi}{2}}\right)\\ - & \, E\left(\vec{\theta}_{\overline{lki,\, l'k'j}}, \theta_{lki}^{\frac{\pi}{2}}, \theta_{l'k'j}^{-\frac{\pi}{2}}\right)- \, E\left(\vec{\theta}_{\overline{lki,\, l'k'j}}, \theta_{lki}^{-\frac{\pi}{2}}, \theta_{l'k'j}^{\frac{\pi}{2}}\right)\Big)\,. 
\end{split} 
\end{equation}

We now turn our attention to the Hardware Efficient Ansatz (HEA) that we employed for unitary compilation and quantum autoencoding in the main text. In this case, every $Rx$ or $Ry$ gate in the ansatz is a generator of the form $e^{-i\frac{\theta}{2}\sigma^k}$ (where $\sigma^k\in\{\sigma^x,\sigma^y\}$ has eigenvalues $\pm1$), with an independent angle.
Hence, the QFIM elements are simply
\be \begin{split}  [F(\thv)]_{lk,l'k'} = \frac{1}{8} \Big(&\LC\left( \vec{\theta}_{\overline{lk,\,l'k'}}, \theta_{lk}^{\frac{\pi}{2}},\theta_{l'k'}^{\frac{\pi}{2}} \right)  + \,\LC\left(\vec{\theta}_{\overline{lk,\,l'k'}}, \theta_{lk}^{-\frac{\pi}{2}},\theta_{l'k'}^{-\frac{\pi}{2}}\right)\\ - &\,\LC\left(\vec{\theta}_{\overline{lk,\,l'k'}}, \theta_{lk}^{\frac{\pi}{2}}, \theta_{l'k'}^{-\frac{\pi}{2}}\right) - \,\LC\left(\vec{\theta}_{\overline{lk,\,l'k'}}, \theta_{lk}^{-\frac{\pi}{2}},\theta_{l'k'}^{\frac{\pi}{2}}\right)\Big)\,.  \end{split}\ee
We recall that the QFIM is independent of the loss function, and hence the above formula is valid for both the unitary compilation task and the quantum autoencoder. The difference between these two is that the initial states are different, and this has an impact on the QFIM. The Hessian matrices, on the contrary, depend on the loss function and thus are different in each case, but in our simulations we only computed the Hessian for the unitary compilation case. The loss here is given by $\LC(\thv) = 1 - \frac{1}{d^2} |\Tr(W^\dagger U(\vec{\theta}))|^2$, where $W$ is the unitary being compiled and $d=2^n$.
In this case, the term $L(\theta)=\frac{1}{d^2} |\Tr(W^\dagger U(\vec{\theta}))|^2$ can be directly evaluated on a quantum computer, and so its second partial derivative $\partial_{lk}\,\partial_{l'k'}\, L(\vec{\theta})$ is
\be \partial_{lk}\,\partial_{l'k'}\, L(\vec{\theta}) = \frac{1}{4} \Big(L\left( \vec{\theta}_{\overline{lk,\,l'k'}}, \theta_{lk} ^{\frac{\pi}{2}},\theta_{l'k'}^{\frac{\pi}{2}} \right)  + L\left(\vec{\theta}_{\overline{lk,\,l'k'}}, \theta_{lk}^{-\frac{\pi}{2}},\theta_{l'k'}^{-\frac{\pi}{2}}\right) - L\left(\vec{\theta}_{\overline{lk,\,l'k'}}, \theta_{lk}^{\frac{\pi}{2}}, \theta_{l'k'}^{-\frac{\pi}{2}}\right) - L\left(\vec{\theta}_{\overline{lk,\,l'k'}}, \theta_{lk}^{-\frac{\pi}{2}},\theta_{l'k'}^{\frac{\pi}{2}}\right)\Big)\,. \ee
Then, applying the chain rule twice on $\LC(\vec{\theta})=1-L(\vec{\theta})$ as
\be \partial_{lk}\,\partial_{l'k'}\, \LC(\vec{\theta}) = \LC'(L(\vec{\theta}))\, \partial_{lk}\,\partial_{l'k'}\,L(\vec{\theta})  + \LC''(L(\vec{\theta}))\, \partial_{lk}\, L(\vec{\theta})\,\partial_{l'k'}\,L(\vec{\theta})\,,  \ee
where $\LC'$ ($\LC''$) is the first (second) derivative of $\LC$ with respect to $L$, the expression for the Hessian matrix is found to be
\be [\nabla^2\LC(\thv)]_{lk,l'k'}  = -\frac{1}{4} \Big(L\left( \vec{\theta}_{\overline{lk,\,l'k'}}, \theta_{lk}^{\frac{\pi}{2}},\theta_{l'k'}^{\frac{\pi}{2}} \right)  + L\left(\vec{\theta}_{\overline{lk,\,l'k'}}, \theta_{lk}^{-\frac{\pi}{2}},\theta_{l'k'}^{-\frac{\pi}{2}}\right) - L\left(\vec{\theta}_{\overline{lk,\,l'k'}}, \theta_{lk}^{\frac{\pi}{2}}, \theta_{l'k'}^{-\frac{\pi}{2}}\right) -L\left(\vec{\theta}_{\overline{lk,\,l'k'}}, \theta_{lk}^{-\frac{\pi}{2}},\theta_{l'k'}^{\frac{\pi}{2}}\right)\Big) \,. \ee

\section{Additional Numerical Results}
We present here some additional numerical results that we obtained in simulations. In Sup. Fig. \ref{fig:hva_ranks_qfim_hess_supp}, we show the eigenvalues of the QFIM and Hessian computed at the global optimum for the loss function $E(\thv)=\langle \psi(\thv)|H_{\TFIM}|\psi(\thv)\rangle$ and the Hamiltonian variational ansatz with closed boundary conditions, for $n=10$ qubits. The interest in showing these plots is that therein one can better appreciate that the spectrum does not form a continuum, but rather, that there is a large gap between the non-zero and the zero eigenvalues, so that there is no ambiguity when defining the rank, stemming from numerical precision issues. Moreover, we computed the eigenvalues at random points in the landscape, where the rank of the QFIM is also bounded  by $\dim(\liea_\SC)$, unlike the rank of the Hessian. The spectra of the Hessian at random points in the landscape further informed us that the landscape is highly non-convex, as it contains both positive and negative eigenvalues. It is also interesting to note that in order to obtain a rank of the Hessian that is bounded by $\dim(\liea_\SC)$, one needs to compute it at the global minimum; otherwise, we encountered fairly-good local minima that did not fulfill this result. We remark that all these features were found in all cases where we computed and diagonalized the QFIM and the Hessian.

Furthermore, in Sup. Fig. \ref{fig:hva_rank_vs_parameters_supp} we computed the rank of the QFIM at 30 random points in the landscape for the HVA with open boundaries and different depths (\textit{i.e.} number of parameters), for $n=4,6$ qubits. This figure shows that the rank quickly statures at all points once a critical number of parameters $M_c$ is reached, suggesting that, at least in this case, overparametrization largely arises simultaneously across the entire landscape. We note as well that before having $M_c$ parameters, the average rank is equal to the number of parameters, \textit{i.e.} there is a perfectly linear relation between the two quantities. Adding more parameters beyond $M_c$ however seems to have no effect on the rank, as it only adds null eigenvalues.

\begin{figure}[t!]
    \centering
    \includegraphics[width=1\linewidth]{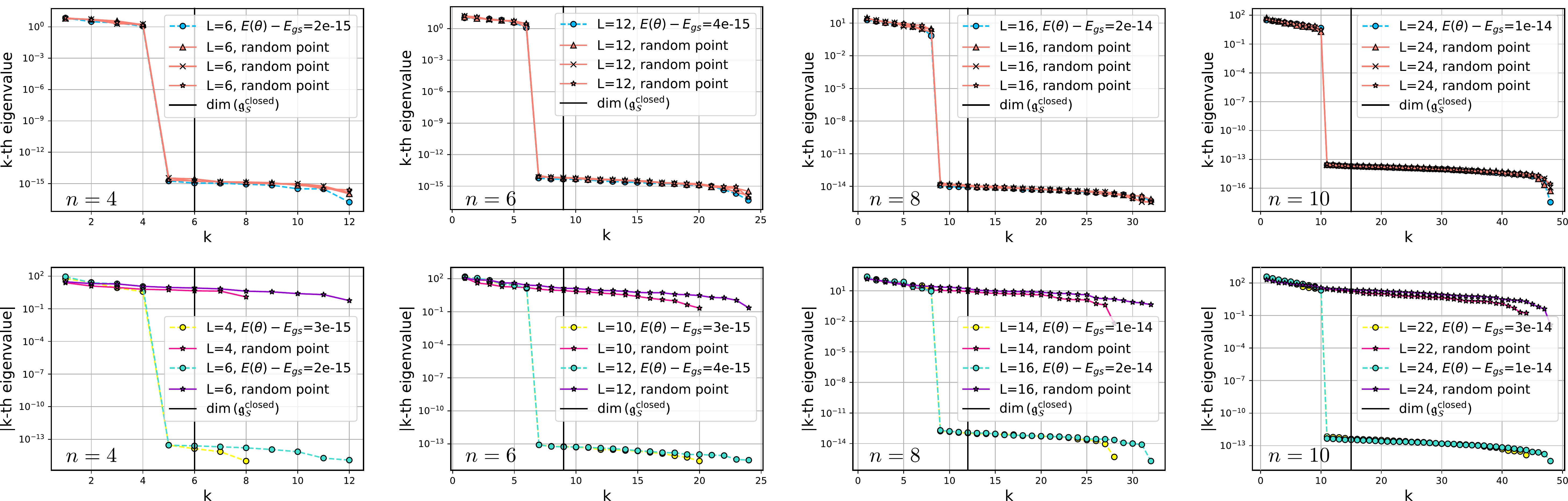}
    \caption{{\bf Spectra of the QFIM and the Hessian for the VQE implementations.} Top row: QFIM spectra for the Hamiltonian variational ansatz with closed boundary conditions, for $n=4,6,8,10$ qubits and $L=6,12,16,24$ layers, both at the global optima and at three random points in the landscape. Bottom row: Hessian spectra for the Hamiltonian variational ansatz with closed boundary conditions, for $n=4,6,8,10$ qubits and $L=4,6,10,12,14,16,22,24$ layers, both at the global optima and at a random point in the landscape.}
    \label{fig:hva_ranks_qfim_hess_supp}
\end{figure}

\begin{figure}[b!]
    \centering
    \includegraphics[scale=0.5]{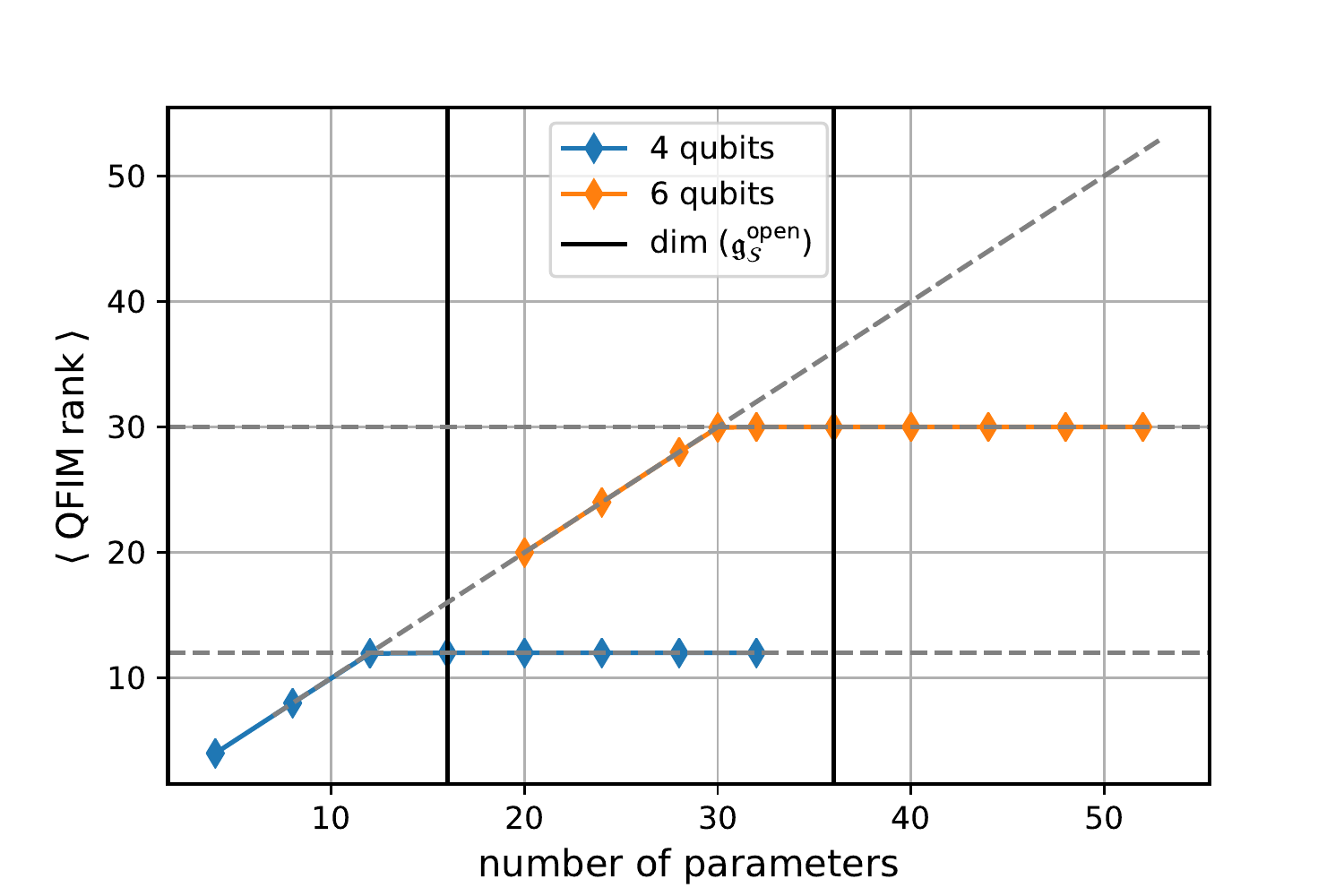}
    \caption{{\bf Average QFIM rank versus number of parameters for the VQE implementations.} Average rank of the QFIM across 30 random points in the landscape, for the Hamiltonian variational ansatz with open boundary conditions and $n=4,6$ qubits. The horizontal dashed lines mark the maximal ranks that the QFIMs achieve, and the tilted dashed line is the line $\langle {\rm QFIM \;rank} \rangle =$ number of parameters. The vertical black lines correspond to the respective $\dim(\liea_\SC)$ (leftmost for $n=4$, and rightmost for $n=6$). We remark that the standard deviation is exactly 0 for all the points in the plot, except for the case $n=4$ (6) and $M=12$ (30) parameters, where $\rank[F(\thv)]=12$ (30) for 29 out of 30 points and $\rank[F(\thv)]=11$ (29) for the remaining one.}
    \label{fig:hva_rank_vs_parameters_supp}
\end{figure}

\end{document}